%% file: main.tex
\documentclass[final,11pt,3p]{elsarticle}

\usepackage{graphicx}
\usepackage[table]{xcolor}
\usepackage{hhline}
\usepackage{tikz,tikz-qtree}\usetikzlibrary{arrows,positioning,calc}

\usepackage{amsmath,amssymb,amsthm}

\usepackage{multicol}\usepackage{multirow}

\usepackage{algorithm}
\usepackage[noend]{algpseudocode}
\algrenewcommand{\algorithmiccomment}[1]{\hfill$\triangleleft$ {\footnotesize \textsl{#1}}}
\algloopdefx{Return}[1]{\textbf{return} #1}
\algblock{Either}{EndOr}
\algcblock{Either}{Or}{EndOr}
\algrenewcommand\algorithmicindent{1.0em}
\makeatletter
\renewcommand{\ALG@beginalgorithmic}{\small}
\makeatother

\DeclareMathAlphabet{\mathpzc}{OT1}{pzc}{m}{it}
\DeclareMathOperator{\hsA}{\langle A\rangle}
\DeclareMathOperator{\hsL}{\langle L\rangle}
\DeclareMathOperator{\hsB}{\langle B\rangle}
\DeclareMathOperator{\hsE}{\langle E\rangle}
\DeclareMathOperator{\hsD}{\langle D\rangle}
\DeclareMathOperator{\hsO}{\langle O\rangle}
\DeclareMathOperator{\hsX}{\langle X\rangle}
\DeclareMathOperator{\hsAt}{\langle \overline{A}\rangle}
\DeclareMathOperator{\hsLt}{\langle \overline{L}\rangle}
\DeclareMathOperator{\hsBt}{\langle \overline{B}\rangle}
\DeclareMathOperator{\hsEt}{\langle \overline{E}\rangle}
\DeclareMathOperator{\hsDt}{\langle \overline{D}\rangle}
\DeclareMathOperator{\hsOt}{\langle \overline{O}\rangle}

\DeclareMathOperator{\Trk}{Trk}
\DeclareMathOperator{\Pref}{Pref}
\DeclareMathOperator{\Suff}{Suff}
\DeclareMathOperator{\states}{states}
\DeclareMathOperator{\intstates}{intstates}
\DeclareMathOperator{\lst}{lst}
\DeclareMathOperator{\fst}{fst}
\DeclareMathOperator{\nestb}{Nest_B}

\DeclareMathOperator{\Root}{root}

\usepackage{cancel}
\DeclareMathOperator{\Rt}{R_t}
\DeclareMathOperator{\notRt}{\cancel{R}_t}

\DeclareMathOperator{\expand}{expand}

\newcommand{\AAbar}{\mathsf{A\overline{A}}}
\newcommand{\AAbarBbarEbar}{\mathsf{A\overline{A}\overline{B}\overline{E}}}
\newcommand{\ABbar}{\mathsf{A\overline{B}}}
\newcommand{\HSexi}{\mathsf{\exists A\overline{A}BE}}
\newcommand{\HSforall}{\mathsf{\forall A\overline{A}BE}}
\newcommand{\AAbarBBbarEbar}{\mathsf{A\overline{A}B\overline{B}\overline{E}}}
\newcommand{\AAbarEBbarEbar}{\mathsf{A\overline{A}E\overline{B}\overline{E}}}
\newcommand{\AAbarBE}{\mathsf{A\overline{A}BE}}
\newcommand{\D}{\mathsf{D}}
\newcommand{\BE}{\mathsf{BE}}
\newcommand{\BED}{\mathsf{BED}}
\newcommand{\ABbarL}{\mathsf{A\overline{B}L}}
\newcommand{\HSprop}{\mathsf{Prop}}

\newcommand{\Th}{\mbox{P}^{\mbox{\scriptsize NP}[O(\log n)]}}
\newcommand{\Thsq}{\mbox{P}^{\mbox{\scriptsize NP}[O(\log^2 n)]}}

\newtheorem{theorem}{Theorem}
\newtheorem{definition}[theorem]{Definition}
\newtheorem{proposition}[theorem]{Proposition}
\newtheorem{lemma}[theorem]{Lemma}
\newtheorem{corollary}[theorem]{Corollary}
\newtheorem{example}{Example}
\newtheorem*{thm**}{Theorem~\mynumber}
\newenvironment{theorem*}[1]
  {\newcommand{\mynumber}{\ref{#1}}\begin{thm**}}
  {\end{thm**}}

\usepackage[colorlinks]{hyperref}

\begin{document}

\begin{frontmatter}

\title{Model Checking for Fragments of\\ Halpern and Shoham's Interval Temporal Logic\\ Based on Track Representatives\tnoteref{label0}}
\tnotetext[label0]{This paper is an extended and revised version of \cite{MMP15} and \cite{MMP15B}.}

\author[uniud]{Alberto Molinari}
\ead{molinari.alberto@gmail.com}
\address[uniud]{Department of Mathematics, Computer Science, and Physics, University of Udine, Italy}

\author[uniud]{Angelo Montanari}
\ead{angelo.montanari@uniud.it}

\author[unina]{Adriano Peron}
\ead{adrperon@unina.it}
\address[unina]{Department of Electronic Engineering and IT, University of Napoli ``Federico II'', Italy}

\begin{abstract}
\input{abstract} 
\end{abstract}


\end{frontmatter}


\input{intro}
\input{background}
\input{section01}
\input{section02}

\input{section03}

\input{section04}

\input{section05}
\input{conclus}

\vfill

\section*{Acknowledgements}
The work by Adriano Peron has been supported by the SHERPA collaborative project, 
which has received funding from the European Community $7$-th Framework Programme 
(FP7/2007-2013) under grant agreements ICT-600958. He is solely responsible 
for its content. The paper does not represent the opinion of the European Community and 
the Community is not responsible for any use that might be made of the information 
contained therein. The work by Alberto Molinari and Angelo Montanari has been supported 
by the GNCS project \emph{Logic, Automata, and Games for Auto-Adaptive Systems}.

\clearpage

\begin{small}
    \input{biblioFLAT}
\end{small}

\clearpage
\appendix
\input{appendix}

\end{document}

%% file: abstract.tex
Model checking allows one to automatically verify a specification of the expected 
properties of a system against a formal model of its behaviour (generally, a Kripke 
structure). 
Point-based temporal logics, such as LTL, CTL, and CTL$^*$, that describe how the 
system evolves state-by-state, are commonly used as specification languages. 
They proved themselves quite successful in a variety of application domains.
However, properties constraining the temporal ordering of temporally 
extended events as well as properties involving temporal aggregations, which are 
inherently interval-based, can not be properly dealt with by them. 
Interval temporal logics (ITLs), that take intervals as their primitive temporal 
entities, turn out to be well-suited for the specification and verification of 
interval properties of computations (we interpret all the tracks of a Kripke structure
as computation intervals).

In this paper, we study the model checking problem for some fragments of 
Halpern and Shoham's modal logic of time intervals (HS). HS features one modality 
for each possible ordering relation between pairs of intervals (the so-called 
Allen's relations).
First, we describe an EXPSPACE model checking algorithm for the HS fragment of Allen's relations 
\emph{meets}, \emph{met-by}, \emph{starts}, \emph{started-by}, and \emph{finishes}, which 
exploits the possibility of finding, for each track (of unbounded length), an equivalent 
bounded-length track representative. 
While checking a property, it only needs to consider tracks whose length does not exceed 
the given bound. Then, we prove the model checking problem for such a fragment to be 
PSPACE-hard.
Finally, we identify other well-behaved HS fragments which are expressive enough to capture meaningful interval properties of systems, such as mutual exclusion, state reachability, and non-starvation, and whose computational complexity is less than or equal to that of LTL.

%% file: intro.tex
\section{Introduction}

One of the most notable techniques for system verification is
model checking, which allows one to verify the desired properties of a system
against a model of its behaviour~\cite{CGP02}.
Properties are usually formalized by means of temporal logics, such as LTL and CTL, 
and systems are represented as labelled state-transition graphs (Kripke structures).
Model checking algorithms perform, in a fully automatic way, an (implicit or explicit) exhaustive enumeration of all the states reachable by the system, and either terminate positively, proving that all properties are met, or produce a counterexample, witnessing that some behavior falsifies a property. 

The model checking problem has systematically been investigated in the context of 
classical, point-based temporal logics, like LTL, CTL, and CTL$^*$, which predicate 
over single computation points/states, while it is still largely unexplored in the 
interval logic setting.

Interval temporal logics (ITLs) have been proposed as a formalism for temporal 
representation and reasoning more expressive than standard point-based ones~\cite{HS91,Ven90,chopping_intervals}.
They take intervals, instead of points, as their primitive temporal entities. Such a choice 
gives them the ability to cope with advanced temporal properties, such as actions with duration, accomplishments, and temporal aggregations, which can not be properly dealt with by standard, point-based temporal logics.

Expressiveness of ITLs makes them well suited for many applications in a variety of computer science fields, 
including artificial intelligence (reasoning about action and change, qualitative reasoning, planning, configuration and multi-agent systems, and computational linguistics), theoretical computer science (formal verification, synthesis), 
and databases (temporal and spatio-temporal databases) \cite{DBLP:journals/logcom/BowmanT03,DBLP:conf/ecp/GiunchigliaT99,DBLP:conf/tacas/LomuscioR06,DBLP:journals/ai/Pratt-Hartmann05,DBLP:series/eatcs/ChaochenH04,digital_circuits_thesis,DBLP:journals/corr/MontanariS14,DBLP:journals/fuin/MarcinkowskiM14,roadmap_intervals}.
However, this great expressiveness is a double-edged sword: in most cases the satisfiability problem for ITLs turns out to be undecidable, and, in the few cases of decidable ITLs, the standard proof machinery, like Rabin's theorem, is usually not applicable.

The most prominent ITL is Halpern and Shoham's modal logic of time intervals (HS, for short)~\cite{HS91}. HS features one modality for each of the 13 possible ordering
relations between pairs of intervals (the so-called Allen's relations~\cite{All83}), apart from the equality relation.
In~\cite{HS91}, it has been shown that the satisfiability problem for 
HS interpreted over all relevant (classes of) linear orders is 
undecidable. Since then, a lot of work has been done on the satisfiability
problem for HS fragments, which has shown that undecidability prevails over
them (see~\cite{DBLP:journals/amai/BresolinMGMS14} for an up-to-date
account of undecidable fragments).
However, meaningful exceptions exist, including the interval logic of temporal
neighbourhood $\AAbar$ and the interval logic of sub-intervals $\D$~\cite{BGMS10,BGMS09,BMSS11,MPS10}.

In this paper, we focus our attention on the model checking problem for HS, for which, as we said, 
little work has been done~\cite{DBLP:conf/time/MontanariMPP14,MMMPP15,LM13,LM14,LM15}
(it is worth pointing out that, in contrast to the case of point-based, linear temporal logics, there is not 
an easy reduction from the model checking problem to validity/satisfiability for ITL).

\paragraph{Related work}
In the classical formulation of the model checking problem~\cite{CGP02}, 
point-based temporal logics are used to analyze, for each path in a Kripke structure, 
how proposition letters labelling the states change from one state to the next one along the path. 
In interval-based model checking, in order to check interval properties of computations, one needs to collect 
information about states into computation stretches. This amounts to interpreting each 
finite path of a Kripke structure (a track) as an interval, and to suitably defining its labelling 
on the basis of the proposition letters that hold on the states composing it.

In \cite{DBLP:conf/time/MontanariMPP14}, Montanari et al.\ give a first characterization of the model checking problem for full HS, interpreted over finite Kripke structures (under the homogeneity assumption~\cite{Roe80}, according to which a proposition letter holds on an interval if and only if it holds on all its sub-intervals). In that paper, the authors introduce the basic elements of the general picture, namely, the interpretation of HS formulas over (abstract) interval models, the mapping of finite Kripke structures into (abstract) interval models, the notion of track descriptor, and a small model theorem proving (with a non-elementary procedure) the decidability of the model checking problem for full HS against finite Kripke structures. Many of these notions will be recalled in the following section.
In \cite{MMMPP15}, Molinari et al.\ work out the model checking problem for full HS in all its details, and prove that it is EXPSPACE-hard, if a succinct encoding of formulas is allowed, and PSPACE-hard otherwise.

In~\cite{LM13,LM14,LM15}, Lomuscio and Michaliszyn address the model checking problem for some fragments of HS extended with epistemic modalities. Their semantic assumptions differ from those made in \cite{DBLP:conf/time/MontanariMPP14}, making it difficult to compare the outcomes of the two research directions. In both cases, formulas of interval temporal logic are evaluated over finite paths/tracks 
obtained from the unravelling of a finite Kripke structure. However, while in \cite{DBLP:conf/time/MontanariMPP14} a proposition letter holds over an interval (track) if and only if it holds over all its states (homogeneity assumption), in~\cite{LM13,LM14} truth of proposition letters on a track/interval depends only on their values at its endpoints. 

In~\cite{LM13}, the authors focus their attention on the HS fragment $\BED$ of Allen's relations \emph{started-by}, \emph{finished-by}, and \emph{contains} (since modality $\hsD$ is definable in terms of modalities $\hsB$ and $\hsE$, $\BED$ is actually as expressive as $\BE$), extended with epistemic modalities. They consider a restricted form of model checking, which verifies the given specification against a single (finite) initial computation interval. Their goal is indeed to reason about a given computation of a multi-agent system, rather than on all its admissible computations.
They prove that the considered model checking problem is PSPACE-complete; moreover, they show that the same problem restricted to the pure temporal fragment $\BED$, that is, the one obtained by removing epistemic modalities, is in PTIME. These results do not come as a surprise as they trade expressiveness for efficiency: modalities $\hsB$ and $\hsE$ allow one to access only sub-intervals of the initial one, whose number is quadratic in the length (number of states) of the initial interval.

In~\cite{LM14}, they show that the picture drastically changes with other fragments of HS, that allow one to access infinitely many tracks/intervals. In particular, they prove that the model checking problem for the HS fragment $\ABbarL$ of Allen's relations \emph{meets}, \emph{starts}, and \emph{before} (since modality $\hsL$ is definable in terms of modality $\hsA$, $\ABbarL$ is actually as expressive as $\ABbar$), extended with epistemic modalities, is decidable with a non-elementary upper bound. Note that, thanks to modalities $\hsA$ and $\hsBt$, formulas of $\ABbarL$ can possibly refer to infinitely many (future) tracks/intervals.

Finally, in~\cite{LM15}, Lomuscio and Michaliszyn show how to use regular expressions in order to specify the way in which tracks/intervals of a Kripke structure get labelled. Such an extension leads to a significant increase in 
expressiveness, as the labelling of an interval is no more determined by that of its endpoints, but it depends on the ordered sequence of states the interval consists of. They also prove that there is not a corresponding
increase in computational complexity, as the complexity bounds given in \cite{LM13,LM14} still hold with the new semantics: the model checking problem for $\BED$ is still in PSPACE, and it is non-elementarily decidable for $\ABbarL$.

\paragraph{Main contributions}
In this paper, we elaborate on the approach to ITL model checking outlined in \cite{DBLP:conf/time/MontanariMPP14} and
we propose an original solution to the problem for some relevant HS fragments based on the notion of track representative. We first prove that the model checking problem for two large HS fragments,
namely, the fragment $\AAbarBBbarEbar$ (resp., $\AAbarEBbarEbar$) of Allen's relations \emph{meets}, \emph{met-by}, \emph{started-by} (resp., \emph{finished-by}), \emph{starts} and \emph{finishes}, is in EXPSPACE. Moreover, we show that it is PSPACE-hard (NEXP-hard, if a succinct encoding of formulas is used). Then, we identify some well-behaved HS fragments, which are still expressive enough to capture meaningful interval properties of state-transition systems, such as mutual exclusion, state reachability, and non-starvation, whose model checking problem exhibits a considerably lower computational complexity, notably, $(i)$ the fragment $\AAbarBbarEbar$, whose model checking problem is PSPACE-complete, and $(ii)$ the fragment $\HSforall$, including formulas of $\AAbarBE$ where only universal modalities are allowed and negation can be applied to propositional formulas only, whose model checking problem is coNP-complete.
%
%

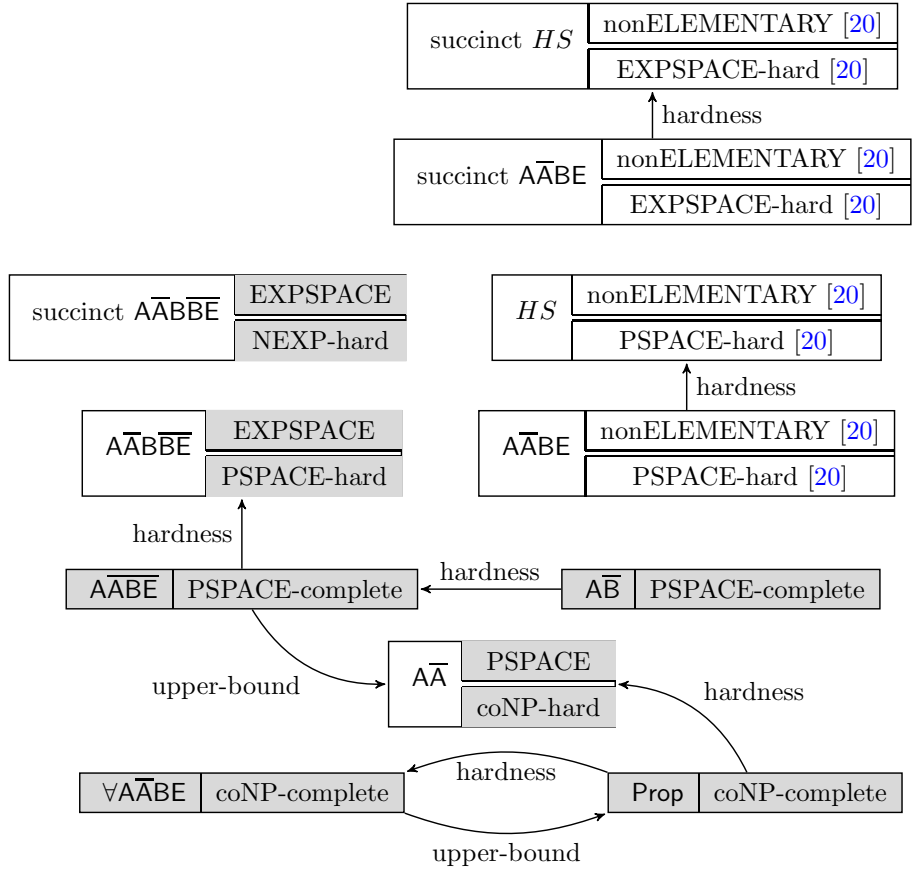
\begin{figure}[t]
\centering
\scalebox{0.9}{
\input{overvGraph.tex}}
\caption{Complexity of model checking for HS fragments.}\label{fig:overv}
\end{figure}
In Figure~\ref{fig:overv}, we summarize known (white boxes) and new (grey boxes) results about complexity of model checking for HS fragments.


The main technical contributions of the paper can be summarized as follows.
\begin{itemize}
\item \emph{Track descriptors}. We start with some background knowledge about HS and Kripke structures, and then we show how the latter can be mapped into interval-based structures, called \emph{abstract interval models}, over which 
HS formulas are evaluated. Each track in a Kripke structure is interpreted as an interval, which becomes an (atomic) object of the domain of an abstract interval model. The labeling of an interval is defined on the basis of the states that compose it, according to the homogeneity assumption~\cite{Roe80}. Then, we introduce \emph{track descriptors}~\cite{DBLP:conf/time/MontanariMPP14}. A track descriptor is a tree-like structure providing information about a possibly infinite set of tracks (the number of admissible track descriptors for a given Kripke structure is  finite). Being associated with the same descriptor is indeed a sufficient condition for two tracks to be indistinguishable with respect to satisfiability of $\AAbarBBbarEbar$ formulas, provided that the nesting depth of $\hsB$ modality is less than or equal to the depth of the descriptor itself. Finally, we introduce the key notions of \emph{descriptor sequence for a track} and \emph{cluster}, and the relation of \emph{descriptor element indistinguishability}, which allow us to determine when two prefixes of some track are associated with the same descriptor, avoiding the expensive operation of explicitly constructing track descriptors.

\item \emph{A small model theorem}. The main result of the paper is a small model theorem, showing that we can 
restrict the verification of an $\AAbarBBbarEbar$ formula to a finite number of bounded-length \emph{track representatives}. A track representative is a track that can be analyzed in place of all---possibly infinitely many---tracks associated with its descriptor. 		
We use track representatives to devise an EXPSPACE 
model checking algorithm for $\AAbarBBbarEbar$. Descriptor element indistinguishability plays a fundamental role in the proof of the bound to the maximum length of representatives, and it allows us to show the completeness of the algorithm, which considers all the possible representatives. 
In addition, we prove that the model checking problem for $\AAbarBBbarEbar$ is PSPACE-hard, NEXP-hard if a \emph{succinct encoding} of formulas is used (it is worth noticing that the proposed algorithm requires exponential working space also in the latter case).

\item \emph{Well-behaved HS fragments}. We first show that the proposed model checking algorithm can verify formulas with a constant nesting depth of $\hsB$ modality by using polynomial working space. This allows us to conclude that 
the model checking problem for $\AAbarBbarEbar$ formulas (which lack modality $\hsB$) is in PSPACE.
Then, we prove that the model checking problem for $\ABbar$ is PSPACE-hard. PSPACE-completeness of $\AAbarBbarEbar$ (and $\ABbar$) immediately follows. Next, we deal with the fragment $\HSforall$. We first provide a coNP model checking algorithm for $\HSforall$, and then we show that model checking for the pure propositional fragment $\HSprop$ is coNP-hard. The two results together allow us to conclude that the model checking problem for both $\HSprop$ and $\HSforall$ is coNP-complete. In addition, upper and lower bounds to the complexity of the problem for $\AAbar$ (the logic of temporal neighbourhood) directly follow: since $\AAbar$ is a fragment of $\AAbarBbarEbar$ and $\HSprop$ is a fragment of $\AAbar$, complexity of model checking for $\AAbar$ is in between coNP and PSPACE.
\end{itemize}

\paragraph{Organization of the paper} In Section~\ref{sec:backgr}, we provide some background knowledge. Then, in
Section~\ref{sec:Bdescript}, we introduce track descriptors~\cite{DBLP:conf/time/MontanariMPP14} and, 
in Section~\ref{sec:clusters}, we formally define the key relation of indistinguishability over descriptor elements.
In Section~\ref{sec:representatives}, we describe an EXPSPACE model checking algorithm for $\AAbarBBbarEbar$ based
on track representatives. We also show how to obtain a PSPACE model checking algorithm for $\AAbarBbarEbar$ by
suitably tailoring the one for $\AAbarBBbarEbar$. In Section~\ref{subsec:AAbarBbarEbar}, we prove that model checking for $\AAbarBbarEbar$ is PSPACE-hard; PSPACE-completeness immediately follows. Moreover, we get for free a lower bound to the complexity of the model checking problem for $\AAbarBBbarEbar$, which turns out to be PSPACE-hard (in the appendix, we show that the problem is NEXP-hard if a succinct encoding of formulas is used). Finally, in Section~\ref{subsec:forallAAbarBE} we provide a coNP model checking algorithm for $\HSforall$ and then we show that the problem is actually coNP-complete. Conclusions give a short assessment of the work done and describe future research directions.

%% file: overvGraph.tex
\newcommand{\cellThree}[3]{
\begin{tabular}{c|c}
\rule[-1ex]{0pt}{3.5ex}
\multirow{2}{*}{#1} & #2 \\ 
\hhline{~=}\rule[-1ex]{0pt}{3.5ex}
 & #3 
\end{tabular}}

\newcommand{\cellTwo}[2]{\begin{tabular}{c|c}
\rule[-1ex]{0pt}{3.5ex}
#1 & #2 \\
\end{tabular}}

\begin{tikzpicture}[->,>=stealth',shorten >=1pt,auto,semithick,main node/.style={rectangle,draw, inner sep=0pt}]  

\tikzstyle{gray node}=[fill=gray!30]

    \node [main node,gray node](0) at (-4,0) {\cellTwo{$\AAbarBbarEbar$}{PSPACE-complete}};
    \node [main node,gray node](1) at (3,0)  {\cellTwo{$\ABbar$}{PSPACE-complete}};
    \node [main node](2) at (-0.2,-1.4) {\cellThree{$\AAbar$}{\cellcolor{gray!30}{PSPACE}}{\cellcolor{gray!30}{coNP-hard}}};
    \node [main node,gray node](3) at (-4,-3) {\cellTwo{$\HSforall$}{coNP-complete}};
    \node [main node,gray node](4) at (3.5,-3) {\cellTwo{$\HSprop$}{coNP-complete}};
    
    \node [main node](5) at (-4,2) {\cellThree{$\AAbarBBbarEbar$}{\cellcolor{gray!30}{EXPSPACE}}{\cellcolor{gray!30}{PSPACE-hard}}};
    \node [main node](6) at (-4.5,4) {\cellThree{succinct $\AAbarBBbarEbar$}{\cellcolor{gray!30}{EXPSPACE}}{\cellcolor{gray!30}{NEXP-hard}}};
    
    \node [main node](7) at (2.5,2) {\cellThree{$\AAbarBE$}{nonELEMENTARY \cite{MMMPP15}}{PSPACE-hard \cite{MMMPP15}}};
    \node [main node](8) at (2.5,4) {\cellThree{$HS$}{nonELEMENTARY \cite{MMMPP15}}{PSPACE-hard \cite{MMMPP15}}};
    \node [main node](9) at (2,6) {\cellThree{succinct $\AAbarBE$}{nonELEMENTARY \cite{MMMPP15}}{EXPSPACE-hard \cite{MMMPP15}}};
    \node [main node](10) at (2,8) {\cellThree{succinct $HS$}{nonELEMENTARY \cite{MMMPP15}}{EXPSPACE-hard \cite{MMMPP15}}};
   
    \path
    (1) edge [swap] node {hardness} (0) 
    (0) edge [bend right,swap] node {upper-bound} (2.west)
        edge node {hardness} (5)
    (4) edge [bend right,swap] node {hardness }(2.east)
    (4.north west) edge [out=160,in=20] node {hardness} (3.north east)
    (3.south east) edge [swap,out=340,in=200] node {upper-bound} (4.south west)
    (7) edge [swap] node {hardness} (8)
    (9) edge [swap] node {hardness} (10);

\end{tikzpicture}

%% file: background.tex
\section{Preliminaries}\label{sec:backgr}

\subsection{The interval temporal logic HS}
Interval-based approaches to temporal representation and reasoning have been successfully pursued in computer science and artificial intelligence. An interval algebra to reason about intervals and their relative order was first proposed by Allen~\cite{All83}. Then, a systematic logical study of ITLs was done by Halpern and Shoham, who introduced the 
logic HS featuring one modality for each Allen interval relation~\cite{HS91}, except for equality.
\begin{table}[tbp]
\centering
\caption{Allen's interval relations and corresponding HS modalities.}\label{allen}
\begin{tabular}{cclc}
\hline
Allen's relation & HS & Definition w.r.t. interval structures &  Example\\ 
\hline

 &   &   & \multirow{7}{*}{\input{allensRels.tex}}\\ 

\emph{meets} & $\hsA$ & $[x,y]\mathpzc{R}_A[v,z]\iff y=v$ &\\ 

\emph{before} & $\hsL$ & $[x,y]\mathpzc{R}_L[v,z]\iff y<v$ &\\ 
 
\emph{started-by} & $\hsB$ & $[x,y]\mathpzc{R}_B[v,z]\iff x=v\wedge z<y$ &\\ 

\emph{finished-by} & $\hsE$ & $[x,y]\mathpzc{R}_E[v,z]\iff y=z\wedge x<v$ &\\ 

\emph{contains} & $\hsD$ & $[x,y]\mathpzc{R}_D[v,z]\iff x<v\wedge z<y$ &\\ 

\emph{overlaps} & $\hsO$ & $[x,y]\mathpzc{R}_O[v,z]\iff x<v<y<z$ &\\ 
\hline
\end{tabular}
\end{table}
Table~\ref{allen} depicts 6 of the 13 Allen's relations
together with the corresponding HS (existential) modalities. 
The other 7 are equality and the 6 inverse relations 
(given a binary relation $\mathpzc{R}$, the inverse relation $\overline{\mathpzc{R}}$ is such that $b \overline{\mathpzc{R}} a$ if and only if $a \mathpzc{R} b$). 

The language of HS features a set of proposition letters $\mathpzc{AP}$, the Boolean connectives $\neg$ and $\wedge$, and a temporal modality for each of the (non trivial) Allen's relations, namely, $\hsA$, $\hsL$, $\hsB$, $\hsE$, $\hsD$, $\hsO$, $\hsAt$, $\hsLt$, $\hsBt$, $\hsEt$, $\hsDt$ and $\hsOt$.
HS formulas are defined by the following grammar:
\begin{equation*}
    \psi ::= p \;\vert\; \neg\psi \;\vert\; \psi \wedge \psi \;\vert\; \langle X\rangle\psi \;\vert\; \langle \overline{X}\rangle\psi, \ \ \mbox{ with } p\in\mathpzc{AP},\; X\in\{A,L,B,E,D,O\}.
\end{equation*}
We will make use of the standard abbreviations of propositional logic, e.g., we will write $\psi \vee \phi$ for $\neg(\neg\psi \wedge \neg\phi)$, $\psi \rightarrow \phi$ for $\neg \psi \vee \phi$, and $\psi \leftrightarrow \phi$ for $\left(\psi \rightarrow \phi\right)\wedge\left(\phi \rightarrow \psi\right)$.
Moreover, for all $X$, dual universal modalities $[X]\psi$ and $[\overline{X}]\psi$ are defined as $\neg\langle X\rangle\neg\psi$ and $\neg\langle \overline{X} \rangle\neg\psi$, respectively. 

We will assume the \emph{strict semantics} of HS: only intervals consisting of at least two points are allowed. Under that assumption, HS modalities are \emph{mutually exclusive} and \emph{jointly exhaustive}, 
that is, exactly one of them holds between any two intervals. However, the strict semantics can easily be ``relaxed'' to include point intervals, and all results we are going to prove hold for the non-strict HS semantics as well.
All HS modalities can be expressed in terms of $\hsA$, $\hsB$, and $\hsE$, and the inverse modalities $\hsAt, \hsBt$, and $\hsEt$, as follows:
\begin{equation*}
\begin{array}{cc}
\hsL\psi\equiv\hsA\hsA\psi & \qquad \hsLt\psi\equiv\hsAt\hsAt\psi \\ 
\hsD\psi\equiv\hsB\hsE\psi\equiv \hsE\hsB\psi & \qquad \hsDt\psi\equiv\hsBt\hsEt\psi \equiv\hsEt\hsBt\psi \\ 
\hsO\psi\equiv\hsE\hsBt\psi & \qquad \hsOt\psi\equiv\hsB\hsEt\psi
\end{array}.
\end{equation*}

We denote by $\mathsf{X_1\cdots X_n}$ the fragment of HS that features modalities $\langle X_1\rangle,\cdots, \langle X_n\rangle $ only. 

HS can be viewed as a multi-modal logic with the 6 primitive modalities $\hsA$, $\hsB$, $\hsE$, $\hsAt$, $\hsBt$, and $\hsEt$. Accordingly, HS semantics can be defined over a multi-modal Kripke structure, called here an \emph{abstract interval model}, in which (strict) intervals are treated as atomic objects and Allen's relations as simple binary relations between pairs of them.

\begin{definition}[\cite{MMMPP15}]
An \emph{abstract interval model} is a tuple $\mathpzc{A}=(\mathpzc{AP},\mathbb{I},A_\mathbb{I},B_\mathbb{I},E_\mathbb{I},\sigma)$, 
where
$\mathpzc{AP}$ is a finite set of proposition letters, 
$\mathbb{I}$ is a possibly infinite set of atomic objects (worlds),
$A_\mathbb{I}$, $B_\mathbb{I}$, and $E_\mathbb{I}$ are three binary relations over $\mathbb{I}$, and
$\sigma:\mathbb{I}\mapsto 2^{\mathpzc{AP}}$ is a (total) labeling function which assigns a set of proposition letters to each world.
\end{definition}

Intuitively, in the interval setting, $\mathbb{I}$ is a set of intervals, $A_\mathbb{I}$, $B_\mathbb{I}$, and $E_\mathbb{I}$ are interpreted as Allen's interval relations $A$ (\emph{meets}), $B$
(\emph{started-by}), and $E$ (\emph{finished-by}), respectively, and $\sigma$ assigns to each interval the set of proposition letters that hold over it.

Given an abstract interval model $\mathpzc{A}=(\mathpzc{AP},\mathbb{I},A_\mathbb{I},B_\mathbb{I},E_\mathbb{I}, \sigma)$
and an interval $I\in\mathbb{I}$, truth of an HS formula over $I$ is defined by structural induction on the formula as follows:
\begin{itemize}
    \item $\mathpzc{A},I\models p$ if and only if $p\in \sigma(I)$, for any proposition letter $p\in\mathpzc{AP}$;
    \item $\mathpzc{A},I\models \neg\psi$ if and only if it is not true that $\mathpzc{A},I\models \psi$ (also denoted as $\mathpzc{A},I\not\models \psi$);
    \item $\mathpzc{A},I\models \psi \wedge \phi$ if and only if $\mathpzc{A},I\models \psi$ and $\mathpzc{A},I\models \phi$;
    \item $\mathpzc{A},I\models \langle X\rangle\psi$, for $X \in\{A,B,E\}$, if and only if there exists $J\in\mathbb{I}$ such that $I\, X_\mathbb{I}\, J$ and $\mathpzc{A},J\models \psi$;
    \item $\mathpzc{A},I\models \langle \overline{X}\rangle\psi$, for $\overline{X} \in\{\overline{A},\overline{B},\overline{E}\}$, if and only if there exists $J\in\mathbb{I}$ such that $J\, X_\mathbb{I}\, I$ and $\mathpzc{A},J\models \psi$.
\end{itemize}

\subsection{Kripke structures and abstract interval models}

In this section, we define a mapping 
from Kripke structures to abstract interval models that makes it possible to specify properties
of systems by means of HS formulas. 

\begin{definition}
A finite Kripke structure $\mathpzc{K}$ is a tuple $(\mathpzc{AP},W, \delta,\mu,w_0)$, where $\mathpzc{AP}$ is a set of proposition letters, $W$ is a finite set of states, 
$\delta\subseteq W\times W$ is a left-total relation between pairs of states,
$\mu:W\mapsto 2^\mathpzc{AP}$ is a total labelling function, and $w_0\in W$ is the initial state.
\end{definition}
For all $w\in W$, $\mu(w)$ is the set of proposition letters which hold at that state,
while $\delta$ is the transition relation which constrains the evolution of the system over time.

\begin{figure}[ht]
\centering
\begin{tikzpicture}[->,>=stealth,thick,shorten >=1pt,auto,node distance=3cm,every node/.style={circle,draw}]
    \node [style={double}](v0) {$\stackrel{v_0}{p}$};
    \node (v1) [right of=v0] {$\stackrel{v_1}{q}$};
    \draw (v0) to [bend right] (v1);
    \draw (v1) to [bend right] (v0);
    \draw (v0) to [loop left] (v0);
    \draw (v1) to [loop right] (v1);
\end{tikzpicture}
\caption{The Kripke structure $\mathpzc{K}_{2}$.}\label{KEquiv}
\end{figure}
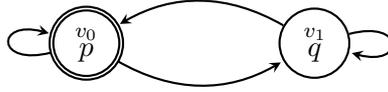
Figure \ref{KEquiv} depicts a Kripke structure, 
$\mathpzc{K}_{2}$, with two states (the initial state is identified by a double circle).
Formally, $\mathpzc{K}_{2}$ is defined by the following quintuple:
\[(\{p,q\},\{v_0,v_1\},\{(v_0,v_0),(v_0,v_1),(v_1,v_0),(v_1,v_1)\},\mu,v_0),\]
where $\mu(v_0)=\{p\}$ and $\mu(v_1)=\{q\}$.

\begin{definition}
A track $\rho$ over a finite Kripke structure $\mathpzc{K}=(\mathpzc{AP},W,\delta,\mu,w_0)$ is a \emph{finite} sequence of states $v_0\cdots v_n$, with $n\geq 1$, such that for all $i\in\{0,\cdots ,n-1\}$, $(v_i,v_{i+1})\in \delta$.
\end{definition}

Let $\Trk_\mathpzc{K}$ be the (infinite) set of all tracks over a finite Kripke structure $\mathpzc{K}$. For any track $\rho=v_0\cdots v_n \in \Trk_\mathpzc{K}$, we define:
\begin{itemize}
\item $|\rho|=n+1$;
\item $\rho(i)=v_i$, for $0\leq i\leq |\rho|-1$;
\item $\states(\rho)=\{v_0,\cdots,v_n\}\subseteq W$;
\item $\intstates(\rho)=\{v_1,\cdots,v_{n-1}\}\subseteq W$;
\item $\fst(\rho)=v_0$ and $\lst(\rho)=v_n$;
\item $\rho(i,j)=v_i\cdots v_j$ is a subtrack of $\rho$, for $0\leq i < j\leq |\rho|-1$;
\item $\Pref(\rho)=\{\rho(0,i) \mid 1\leq i\leq |\rho|-2\}$ is the set of all proper prefixes of $\rho$. Note that $\Pref(\rho)=\emptyset$ if $|\rho|=2$;
\item $\Suff(\rho)=\{\rho(i,|\rho|-1) \mid 1\leq i\leq |\rho|-2\}$ is the set of all proper suffixes of $\rho$. Note that $\Suff(\rho)=\emptyset$ if $|\rho|=2$.
\end{itemize}
It is worth pointing out that the length of tracks, prefixes, and suffixes is greater than 1, as they will be mapped into strict intervals.
If $\fst(\rho)=w_0$ (the initial state of $\mathpzc{K}$), $\rho$ is said to be an \emph{initial track}. In the following, we will denote by $\rho\cdot \rho'$ the concatenation of the tracks $\rho$ and $\rho'$, assuming that $(\lst(\rho),\fst(\rho'))\in\delta$ hence $\rho\cdot \rho' \in \Trk_\mathpzc{K}$; moreover,
by $\rho^n$ we will denote the track obtained by concatenating $n$ copies of $\rho$.

An abstract interval model (over $\Trk_\mathpzc{K}$) can be naturally associated with a finite Kripke structure by interpreting every track as an interval bounded by its first and last states.
\begin{definition}[\cite{MMMPP15}]\label{def:inducedmodel}
The abstract interval model induced by a finite Kripke structure $\mathpzc{K}=(\mathpzc{AP},W,\allowbreak \delta,\mu,w_0)$ is the abstract interval model $\mathpzc{A}_\mathpzc{K}=(\mathpzc{AP},\mathbb{I},A_\mathbb{I},B_\mathbb{I},E_\mathbb{I},\sigma)$, where: 
\begin{itemize}
	\item $\mathbb{I}=\Trk_\mathpzc{K}$,
	\item $A_\mathbb{I}=\left\{(\rho,\rho')\in\mathbb{I}\times\mathbb{I}\mid \lst(\rho)=\fst(\rho')\right\}$,
	\item $B_\mathbb{I}=\left\{(\rho,\rho')\in\mathbb{I}\times\mathbb{I}\mid \rho'\in\Pref(\rho)\right\}$, 
	\item $E_\mathbb{I}=\left\{(\rho,\rho')\in\mathbb{I}\times\mathbb{I}\mid \rho'\in\Suff(\rho)\right\}$, and
	\item $\sigma:\mathbb{I}\mapsto 2^\mathpzc{AP}$ where $\sigma(\rho)=\bigcap_{w\in\states(\rho)}\mu(w)$, for all $\rho\in\mathbb{I}$.
\end{itemize}
\end{definition}

In Definition \ref{def:inducedmodel}, relations $A_\mathbb{I},B_\mathbb{I}$, and $E_\mathbb{I}$ are interpreted as Allen's interval relations \emph{meets}, \emph{started-by}, and \emph{finished-by}, respectively. Moreover, according to the definition of $\sigma$, a proposition letter $p\in\mathpzc{AP}$ holds over $\rho=v_0\cdots v_n$ if and only if it holds over all the states $v_0, \ldots , v_n$ of $\rho$. This conforms to the \emph{homogeneity principle}, according to which a proposition letter holds over an interval if and only if it holds over all of its subintervals.

\emph{Satisfiability} of an HS formula over a finite  Kripke structure can be given in terms of induced abstract interval models.
\begin{definition}
Let $\mathpzc{K}$ be a finite Kripke structure, $\rho$ be a track in $\Trk_\mathpzc{K}$, and
$\psi$ be an HS formula. We say that the pair $(\mathpzc{K},\rho)$ satisfies $\psi$, denoted by $\mathpzc{K},\rho\models \psi$, if and only if it holds that $\mathpzc{A}_\mathpzc{K},\rho\models \psi$.
\end{definition}

\begin{definition}
Let $\mathpzc{K}$ be a finite Kripke structure and $\psi$ be an HS formula. We say that
$\mathpzc{K}$ models $\psi$, denoted by $\mathpzc{K}\models \psi$, if and only if 
\emph{for all initial tracks} $\rho\in\Trk_\mathpzc{K}$, it holds that $\mathpzc{K},\rho\models \psi.$
\end{definition}
The \emph{model checking problem} for HS over finite Kripke structures is the problem of deciding whether $\mathpzc{K}\models \psi$.
Since Kripke structures feature an infinite number of tracks, the problem is not trivially decidable.

We end the section by providing some meaningful examples of properties of tracks and/or transition systems that can be expressed in HS.

\begin{example}\label{ex:ellk}
The formula $[B]\bot$ can be used to select all and only the tracks of length $2$. Given any $\rho$, with $|\rho|=2$, independently of $\mathpzc{K}$, it indeed holds that $\mathpzc{K},\rho\models [B]\bot$, because $\rho$ has no (strict) prefixes. On the other hand, it holds that $\mathpzc{K},\rho\models \hsB\top$ if (and only if) $|\rho|>2$.
Finally, let $\ell(k)$ be a shorthand for $[B]^{k-1}\bot \wedge \hsB^{k-2}\top$. 
It holds that $\mathpzc{K},\rho\models \ell(k)$ if and only if $|\rho|=k$.
\end{example}

\begin{example}
Let us consider the finite Kripke structure $\mathpzc{K}_{2}$ depicted in Figure~\ref{KEquiv}. 
The truth of the following statements can be easily checked:
\begin{itemize}
    \item $\mathpzc{K}_{2},(v_0v_1)^2\models \hsA q$;
    \item $\mathpzc{K}_{2},v_0v_1v_0\not\models \hsA q$;
    \item $\mathpzc{K}_{2},(v_0v_1)^2\models \hsAt p$;
    \item $\mathpzc{K}_{2},v_1v_0v_1\not\models \hsAt p$.
\end{itemize}
The above statements show that modalities $\hsA$ and $\hsAt$ can be used to distinguish between tracks that start or end at different states. In particular, note that $\hsA$ (resp., $\hsAt$) allows one to ``move'' to \emph{any} track branching on the right (resp., left) of the considered one, e.g., if $\rho=v_0v_1v_0$, then $\rho\, A_\mathbb{I}\, v_0v_0$, $\rho\, A_\mathbb{I}\, v_0v_1$, $\rho\, A_\mathbb{I}\, v_0v_0v_0$, $\rho\, A_\mathbb{I}\, v_0v_0v_1$, $\rho\, A_\mathbb{I}\, v_0v_1v_0v_1$, and so on.

Modalities $\hsB$ and $\hsE$ can be used to distinguish between tracks encompassing a different number of iterations of a given loop. This is the case, for instance, with the following statements:
\begin{itemize}    
    \item $\mathpzc{K}_{2},(v_1v_0)^3 v_1\models \hsB \big(\hsA p \wedge \hsB \left(\hsA p \wedge \hsB\hsA p\right)\big)$;
    \item $\mathpzc{K}_{2},(v_1v_0)^2 v_1\not\models \hsB \big(\hsA p \wedge \hsB \left(\hsA p \wedge \hsB\hsA p\right)\big)$.
\end{itemize}

Finally, HS makes it possible to distinguish between $\rho_1=v_0^3v_1v_0$ and $\rho_2=v_0v_1v_0^3$, which feature the same number of iterations of the same loops, but differ in the order of loop occurrences: $\mathpzc{K}_{2},\rho_1\models \hsB\left(\hsA q \wedge \hsB\hsA p\right)$ but $\mathpzc{K}_{2},\rho_2\not\models \hsB\left(\hsA q \wedge \hsB\hsA p\right)$.
\end{example}

\begin{example}\label{ExampleSched}
In Figure \ref{KSched}, we give an example of a finite Kripke structure $\mathpzc{K}_{Sched}$ that models the behaviour of a scheduler serving three processes which are continuously requesting the use of a common resource. The initial state 
is $v_0$: no process is served in that state. In any other state $v_i$ and $\overline{v}_i$, with $i \in \{1,2,3\}$, the $i$-th process is served (this is denoted by the fact that $p_i$ holds in those 
states). For the sake of readability, edges are marked either by $r_i$, for $request(i)$, or by $u_i$, for $unlock(i)$. However, edge labels do not have a semantic value, i.e., they are neither part of the structure definition, nor proposition letters; they are simply used to ease reference to edges. 
Process $i$ is served in state $v_i$, then, after ``some time'', a transition $u_i$ from $v_i$ to $\overline{v}_i$ is taken; subsequently, process $i$ cannot be served again immediately, as $v_i$ is not directly reachable from $\overline{v}_i$ (the scheduler cannot serve the same process twice in two successive rounds). A transition $r_j$, with $j\neq i$, from $\overline{v}_i$ to $v_j$ is then taken and process $j$ is served. This structure can be easily generalised to a higher number of processes.

\begin{figure}[ht]
\centering
\begin{tikzpicture}[->,>=stealth',shorten >=1pt,auto,node distance=3cm,thick,main node/.style={circle,draw}]

  \node[main node,style={double}] (1) {$\stackrel{v_0}{\emptyset}$};
  \node[main node,fill=gray!35] (3) [below of=1] {$\stackrel{v_2}{p_2}$};
  \node[main node,fill=gray!50] (2) [left of=3] {$\stackrel{v_1}{p_1}$};
  \node[main node,fill=gray!20] (4) [right of=3] {$\stackrel{v_3}{p_3}$};
  \node[main node,fill=gray!50] (5) [below of=2] {$\stackrel{\overline{v_1}}{p_1}$};
  \node[main node,fill=gray!35] (6) [below of=3] {$\stackrel{\overline{v_2}}{p_2}$};
  \node[main node,fill=gray!20] (7) [below of=4] {$\stackrel{\overline{v_3}}{p_3}$};

  \path[every node/.style={font=\small}]
    (1) edge [bend right] node[left] {$r_1$} (2)
        edge node {$r_2$} (3)
        edge [bend left] node[right] {$r_3$} (4)
    (2) edge [bend right] node [left] {$u_1$} (5)
    (3) edge node {$u_2$} (6)
    (4) edge [bend left] node [right] {$u_3$} (7)
    (5) edge node[very near end,left] {$r_2$} (3)
    (5) edge [out=270,in=270,looseness=1.3] node [near start,swap] {$r_3$} (4)
    (6) edge node[very near end,right] {$r_1$} (2)
    (6) edge node[very near end,left] {$r_3$} (4)
    (7) edge [out=270,in=270,looseness=1.3] node [near start] {$r_1$} (2)
    (7) edge node[very near end,right] {$r_2$} (3)
    ;
\end{tikzpicture}
\vspace{-1.5cm}
\caption{The Kripke structure $\mathpzc{K}_{Sched}$.}\label{KSched}
\end{figure}
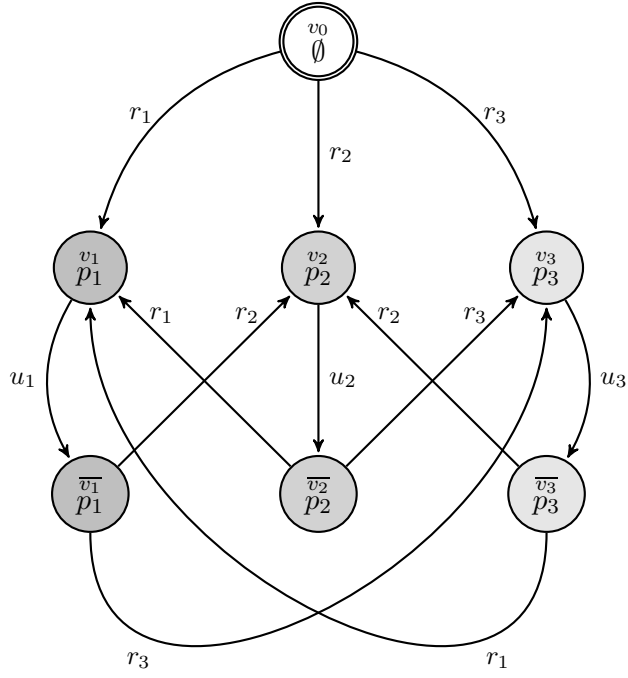

We show how some meaningful properties to check against 
$\mathpzc{K}_{Sched}$ 
can be expressed in HS, and, in particular, by means of formulas of the fragment $\mathsf{\overline{A}E}$---a subfragment of the fragment $\AAbarEBbarEbar$, on which we will focus in the following. 
In all formulas, we force the validity of the considered property over all legal computation sub-intervals by using modality $[E]$ (all computation sub-intervals are suffixes of at least one initial track). 
Truth of the following statements can be easily checked:
\begin{itemize}
    \item $\mathpzc{K}_{Sched}\models[E]\big(\hsE^4\top \rightarrow (\chi(p_1,p_2) \vee \chi(p_1,p_3) \vee \chi(p_2,p_3))\big)$,\\ with $\chi(p,q):=\hsE\hsAt p \wedge \hsE\hsAt q$;
    \item $\mathpzc{K}_{Sched}\not\models[E](\hsE^{10}\top \rightarrow \hsE\hsAt p_3)$;
    \item $\mathpzc{K}_{Sched}\not\models[E](\hsE^6 \rightarrow (\hsE\hsAt p_1 \wedge \hsE\hsAt p_2 \wedge \hsE\hsAt p_3))$.
\end{itemize}
The first formula requires that in any suffix of length at least 6 of an initial track, at least 2 proposition letters are witnessed. $\mathpzc{K}_{Sched}$ satisfies the formula since a process cannot be executed twice consecutively. 

The second formula requires that in any suffix of length at least 12 of an initial track, process 3 is executed at least once in some internal states. $\mathpzc{K}_{Sched}$ does not satisfy the formula since the scheduler, being unfair, can avoid executing a process ad libitum. 

The third formula requires that in any suffix of length at least 8 of an initial track, $p_1$, $p_2$, and $p_3$ are all witnessed. The only way to satisfy this property would be to constrain the scheduler to execute the three processes in a strictly periodic manner, which is not the case.
\end{example}

\section{The notion of $B_k$-descriptor}\label{sec:Bdescript}
For any finite Kripke structure $\mathpzc{K}$, one can find a corresponding induced abstract interval model $\mathpzc{A}_\mathpzc{K}$, featuring one interval for each track of $\mathpzc{K}$. As we already pointed out, since $\mathpzc{K}$ has loops (each state must have at least one successor, as the transition relation $\delta$ is left-total), the number of its tracks, and thus the number of intervals of $\mathpzc{A}_\mathpzc{K}$, is infinite.

In \cite{MMMPP15}, Molinari et al.\ showed that, given a bound $k$ on the structural complexity of HS formulas (that is, on the nesting depth of $\hsB$ and $\hsE$ modalities), it is
possible to obtain a \emph{finite} representation for $\mathpzc{A}_\mathpzc{K}$, which is equivalent to
$\mathpzc{A}_\mathpzc{K}$ with respect to satisfiability of HS formulas with structural complexity less than or equal to $k$.
By making use of such a representation, they prove that the model checking problem for (full) HS is decidable (with a non-elementary upper bound).

In this paper, we first restrict our attention to $\AAbarBBbarEbar$ and 
provide a model checking algorithm of lower complexity. All the results we are going to prove hold also for the fragment $\AAbarEBbarEbar$ by symmetry.
We start with the definition of some basic notions.

\begin{definition} 
\label{def:B-nesting}
Let $\psi$ be an $\AAbarBBbarEbar$ formula. The B-nesting depth of $\psi$, denoted by $\nestb(\psi)$, is defined by induction on the complexity of the formula as follows:
\begin{itemize}
        \item $\nestb(p)=0$, for any proposition letter $p\in\mathpzc{AP}$;
        \item $\nestb(\neg\psi)=\nestb(\psi)$;
        \item $\nestb(\psi\wedge\phi)=\max\{\nestb(\psi),\nestb(\phi)\}$;
        \item $\nestb(\hsB\psi)=1+\nestb(\psi)$;
        \item $\nestb(\hsX\psi)=\nestb(\psi)$, for $X\in\{A, \overline{A}, \overline{B}, \overline{E}\}$.
\end{itemize}
\end{definition}

Making use of Definition \ref{def:B-nesting},
we can introduce the relation(s) of $k$-equivalence over tracks.
\begin{definition}\label{def:k-equivalence}
Let $\mathpzc{K}$ be a finite Kripke structure, $\rho$ and $\rho'$ be two tracks in $\Trk_\mathpzc{K}$, and $k\in\mathbb{N}$. We say that $\rho$ and $\rho'$ are $k$-equivalent if and only if, for every $\AAbarBBbarEbar$ formula $\psi$ with $\nestb(\psi)=k$, $\mathpzc{K},\rho\models \psi$ if and only if $\mathpzc{K},\rho'\models \psi$.
\end{definition}
It can be easily proved that $k$-equivalence propagates downwards.
\begin{proposition}\label{propdown}
Let $\mathpzc{K}$ be a finite Kripke structure, $\rho$ and $\rho'$ be two tracks in $\Trk_\mathpzc{K}$, and $k\in\mathbb{N}$. If $\rho$ and $\rho'$ are $k$-equivalent, then they are $h$-equivalent, 
for all $0\leq h\leq k$.
\end{proposition}
\begin{proof}
Let us assume that $\mathpzc{K},\rho\models \psi$, with $0\leq \nestb(\psi)=h\leq k$. Consider the formula $\hsB^k\top$, whose B-nesting depth is equal to $k$. It holds that either $\mathpzc{K},\rho\models \hsB^k\top$ or $\mathpzc{K},\rho\models \neg\hsB^k\top$. In the first case, we have that $\mathpzc{K},\rho\models \hsB^k\top \wedge \psi$. Since $\nestb(\hsB^k\top \wedge \psi)=k$, from the hypothesis, it immediately follows that $\mathpzc{K},\rho'\models \hsB^k\top \wedge \psi$, and thus $\mathpzc{K},\rho'\models \psi$. The other case can be dealt with in a symmetric way.
\end{proof}

We are now ready to define the key notion of \emph{descriptor} for a track of a Kripke structure.

\begin{definition}[\cite{MMMPP15}] \label{def:trackdescr}
Let $\mathpzc{K}=(\mathpzc{AP},W,\delta,\mu,v_0)$ be a finite Kripke structure, $\rho \in \Trk_\mathpzc{K}$,
and $k\in\mathbb{N}$. The $B_k$-descriptor for $\rho$ is a labelled tree $\mathpzc{D}=(V,E,\lambda)$ of depth $k$, where $V$ is a finite set of vertices, $E\subseteq V\times V$ is a set of edges, 
and $\lambda:V\mapsto W\times 2^W\times W$ is a node labelling function, inductively defined as follows:
    \begin{itemize}
        \item for $k=0$, the $B_k$-descriptor for $\rho$ is the tree                   $\mathpzc{D} = (\{\Root(\mathpzc{D})\},\emptyset,\lambda)$, where 
            $\lambda(\Root(\mathpzc{D}))=(\fst(\rho),\intstates(\rho), \lst(\rho));$
        
        \item for $k>0$, the $B_k$-descriptor for $\rho$ is the tree                   $\mathpzc{D} = (V,E,\lambda)$, where 
            $\lambda(\Root(\mathpzc{D}))=\allowbreak (\fst(\rho),\allowbreak \intstates(\rho),\lst(\rho)),$
        which satisfies the following conditions:
            \begin{enumerate}
                \item for each prefix $\rho'$ of $\rho$, there exists $v\in V$ such that $(\Root(\mathpzc{D}),v)\in E$ and the subtree rooted in $v$ is the $B_{k-1}$-descriptor  for $\rho'$;
                \item for each vertex $v\in V$ such that $(\Root(\mathpzc{D}),v)\in E$, there exists a prefix $\rho'$ of $\rho$ such that the subtree rooted in $v$ is the $B_{k-1}$-descriptor for $\rho'$;
                \item\label{noiso} for all pairs of edges $(\Root(\mathpzc{D}),v'), (\Root(\mathpzc{D}),v'')\in E$, if the subtree rooted in $v'$ is isomorphic to the subtree rooted in $v''$, then $v'=v''$ \footnote{Here and in the following, we write subtree for maximal subtree. Moreover, isomorphism between descriptors accounts for node labels, as well (not only for the structure of descriptors).}.
            \end{enumerate}
    \end{itemize}
\end{definition}

Condition \ref{noiso} of Definition \ref{def:trackdescr} simply states that no two subtrees whose roots are siblings can be isomorphic.
A $B_0$-descriptor $\mathpzc{D}$ for a track consists of its root only, which is denoted by $\Root(\mathpzc{D})$. 
A label of a node will be referred to as a \emph{descriptor element}: the notion of descriptor element bears analogies with an abstraction technique for discrete time Duration Calculus proposed by Hansen et al.\ in~\cite{HPB14}, which, on its turn, is connected to Parikh images~\cite{Par66} (a descriptor element can be seen as a qualitative analogue of this). 

Basically, for any $k \geq 0$, the label of the root of the $B_k$-descriptor $\mathpzc{D}$ for
$\rho$ is the triple $(\fst(\rho),\intstates(\rho),\lst(\rho))$. 
Each prefix $\rho'$ of $\rho$ is associated with some subtree whose root is labelled with $(\fst(\rho'),\intstates(\rho'),\lst(\rho'))$ and is a child of the root of $\mathpzc{D}$.
Such a construction is then iteratively applied to the children of the root until either depth $k$ is reached or a track of length 2 is being considered on a node.
 
Hereafter equality between descriptors is considered \emph{up to isomorphism}.

As an example, in Figure \ref{removeisom} we show the $B_2$-descriptor for the track $\rho = v_0v_1v_0v_0v_0v_0v_1$ of $\mathpzc{K}_{2}$ (Figure \ref{KEquiv}). It is worth noting that there exist two distinct prefixes of $\rho$, that is, the tracks $\rho'=v_0v_1v_0v_0v_0v_0$ and $\rho''=v_0v_1v_0v_0v_0$, which have the same $B_1$-descriptor. Since, according to Definition \ref{def:trackdescr}, no tree can occur more than once as a subtree of the same node (in this example, the root), in the $B_2$-descriptor for $\rho$, prefixes $\rho'$ and $\rho''$ are represented by the same tree (the first subtree of the root on the left). This shows that, in general, the root of a descriptor for a track with $h$ proper prefixes does not necessarily have $h$ children.
\begin{figure}[htbp] 
\centering
\resizebox{\linewidth}{!}{
\begin{tikzpicture}[level distance=15mm,every node/.style={fill=gray!20}]
\Tree [.$(v_0,\{v_0,v_1\},v_1)$
	[.$(v_0,\{v_0,v_1\},v_0)$
		$(v_0,\{v_0,v_1\},v_0)$
		$(v_0,\{v_1\},v_0)$
		$(v_0,\emptyset,v_1)$
]	[.$(v_0,\{v_0,v_1\},v_0)$
		$(v_0,\{v_1\},v_0)$
		$(v_0,\emptyset,v_1)$
]	[.$(v_0,\{v_1\},v_0)$
		$(v_0,\emptyset,v_1)$
]	$(v_0,\emptyset,v_1)$
]
\end{tikzpicture}}
\caption{The $B_2$-descriptor for the track $v_0v_1v_0v_0v_0v_0v_1$ of $\mathpzc{K}_{2}$.}\label{removeisom}
\end{figure}
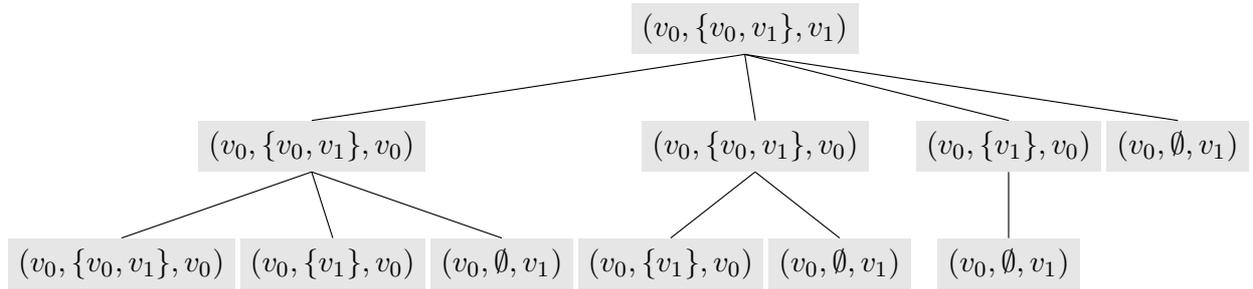

$B$-descriptors do not convey, in general, enough information to determine which track they were built from; however, they can be used to determine which $\AAbarBBbarEbar$ formulas are satisfied by the track from which they were built.

In \cite{MMMPP15}, the authors prove that, for a finite Kripke structure $\mathpzc{K}$, there exists a 
\emph{finite number} (non-elementary w.r.t. $|W|$ and $k$) of possible $B_k$-descriptors. Moreover, the number of nodes of a descriptor has a non-elementary upper bound as well.
Since the number of tracks of $\mathpzc{K}$ is infinite, and for any $k\in\mathbb{N}$ the set of $B_k$-descriptors for its tracks is finite, at least one $B_k$-descriptor must be the $B_k$-descriptor of \emph{infinitely many} tracks. Thus, $B_k$-descriptors naturally induce an equivalence relation of finite index over the set of tracks of a finite Kripke structure (\emph{$k$-descriptor equivalence relation}).

\begin{definition}
Let $\mathpzc{K}$ be a finite Kripke structure, $\rho,\rho' \in \Trk_\mathpzc{K}$,
and $k \in \mathbb{N}$. We say that $\rho$ and $\rho'$ are \emph{$k$-descriptor equivalent} (denoted as $\rho\sim_k\rho'$)
if and only if the $B_k$-descriptors for $\rho$ and $\rho'$ coincide.
\end{definition}

\begin{lemma}\label{symmextlemma}
Let $k\in\mathbb{N}$, $\mathpzc{K}=(\mathpzc{AP},W,\delta,\mu,v_0)$ be a finite Kripke structure and $\rho_1$, $\rho_1'$, $\rho_2$, $\rho_2'$ be tracks in $\Trk_\mathpzc{K}$ such that
$\left(\lst(\rho_1),\fst(\rho_1')\right)\in\delta$, $\left(\lst(\rho_2),\fst(\rho_2')\right)\in\delta$,
$\rho_1\sim_k\rho_2$ and $\rho_1'\sim_k\rho_2'$.
Then $\rho_1\cdot \rho_1'\sim_k\rho_2\cdot\rho_2'$.
\end{lemma}
The proof is reported in \ref{symmextlemmaProof}. The next proposition immediately follows from Lemma~\ref{symmextlemma}.
\begin{proposition}[Left and right extensions]\label{extBk}
Let $\mathpzc{K}=(\mathpzc{AP},W,\delta,\mu,v_0)$ be a finite Kripke structure, $\rho, \rho'$ be two tracks in $\Trk_\mathpzc{K}$ such that $\rho\sim_k\rho'$, and $\overline{\rho}\in \Trk_\mathpzc{K}$. 
If $\left(\lst(\rho),\fst(\overline{\rho})\right)\in\delta$, then $\rho\cdot \overline{\rho}\sim_k\rho'\cdot\overline{\rho}$, and if 
$\left(\lst(\overline{\rho}),\fst(\rho)\right)\in\delta$, then $\overline{\rho}\cdot \rho\sim_k\overline{\rho}\cdot\rho'$.
\end{proposition}

The next theorem proves that, for any pair of tracks $\rho,\rho'\in\Trk_\mathpzc{K}$, if $\rho\sim_k\rho'$, then $\rho$ and $\rho'$ are $k$-equivalent (see Definition \ref{def:k-equivalence}). 
\begin{theorem}[\cite{MMMPP15}]\label{satPresB}
Let $\mathpzc{K}$ be a finite Kripke structure, $\rho$ and $\rho'$ be two tracks in $\Trk_\mathpzc{K}$, and $\psi$ be a formula of $\AAbarBBbarEbar$ with $\nestb(\psi)=k$. If $\rho\sim_k\rho'$, then 
$\mathpzc{K},\rho\models\psi\iff \mathpzc{K},\rho'\models\psi$.
\end{theorem}

Since the set of $B_k$-descriptors for the tracks of a finite Kripke structure $\mathpzc{K}$ is finite, i.e., the equivalence relation $\sim_k$ has a finite index, there always exists a finite number of $B_k$-descriptors that ``satisfy'' an $\AAbarBBbarEbar$ formula $\psi$ with $\nestb(\psi)= k$ (this can be formally proved by a quotient construction \cite{MMMPP15}). 

%% file: allensRels.tex
\begin{tikzpicture}[scale=0.96]
\path[use as bounding box] (-0.4,0.2) rectangle (3.4,-3.0);
\coordinate [label=left:\textcolor{red}{$x$}] (A0) at (0,0);
\coordinate [label=right:\textcolor{red}{$y$}] (B0) at (1.5,0);
\draw[red] (A0) -- (B0);
\fill [red] (A0) circle (2pt);
\fill [red] (B0) circle (2pt);

\coordinate [label=left:$v$] (A) at (1.5,-0.5);
\coordinate [label=right:$z$] (B) at (2.5,-0.5);
\draw[black] (A) -- (B);
\fill [black] (A) circle (2pt);
\fill [black] (B) circle (2pt);

\coordinate [label=left:$v$] (A) at (2,-1);
\coordinate [label=right:$z$] (B) at (3,-1);
\draw[black] (A) -- (B);
\fill [black] (A) circle (2pt);
\fill [black] (B) circle (2pt);

\coordinate [label=left:$v$] (A) at (0,-1.5);
\coordinate [label=right:$z$] (B) at (1,-1.5);
\draw[black] (A) -- (B);
\fill [black] (A) circle (2pt);
\fill [black] (B) circle (2pt);

\coordinate [label=left:$v$] (A) at (0.5,-2);
\coordinate [label=right:$z$] (B) at (1.5,-2);
\draw[black] (A) -- (B);
\fill [black] (A) circle (2pt);
\fill [black] (B) circle (2pt);

\coordinate [label=left:$v$] (A) at (0.5,-2.5);
\coordinate [label=right:$z$] (B) at (1,-2.5);
\draw[black] (A) -- (B);
\fill [black] (A) circle (2pt);
\fill [black] (B) circle (2pt);

\coordinate [label=left:$v$] (A) at (1.3,-3);
\coordinate [label=right:$z$] (B) at (2.3,-3);
\draw[black] (A) -- (B);
\fill [black] (A) circle (2pt);
\fill [black] (B) circle (2pt);

\coordinate (A1) at (0,-3);
\coordinate (B1) at (1.5,-3);
\draw[dotted] (A0) -- (A1);
\draw[dotted] (B0) -- (B1);
\end{tikzpicture}

%% file: section01.tex
\section{Clusters and descriptor element indistinguishability}\label{sec:clusters}

A $B_k$-descriptor provides a finite encoding for a possibly infinite set of tracks (the tracks associated with that descriptor). Unfortunately, the representation of $B_k$-descriptors as trees labelled over descriptor elements is highly redundant. For instance, given any pair of subtrees rooted in some children of the root of a descriptor, it is always the case that one of them is a subtree of the other:
the two subtrees are associated with two (different) prefixes of a track and one of them is necessarily a prefix of the other. In practice, the size of the tree representation of $B_k$-descriptors prevents their direct use in model checking algorithms, and makes it difficult to determine the intrinsic complexity of $B_k$-descriptors.

In this section, we devise a more compact representation of $B_k$-descriptors. 
Each class of the $k$-descriptor equivalence relation is a set of $k$-equivalent tracks. For any such class, we select (at least) one track representative whose length is (exponentially) bounded in both the size of $W$ (the set of states of the Kripke structure) and $k$. 
In order to determine such a bound, we consider suitable ordered sequences (possibly with repetitions) of descriptor elements of a $B_k$-descriptor. Let the \emph{descriptor sequence} for a track be the ordered sequence of descriptor elements associated with its prefixes. It can be easily checked that in a descriptor sequence descriptor elements can be repeated. We devise a criterion to avoid such repetitions whenever they cannot be distinguished by an $\AAbarBBbarEbar$ formula of $B$-nesting depth up to $k$.
\begin{definition}
Let $\rho=v_0v_1 \cdots v_n$ be a track of a finite Kripke structure. The descriptor sequence $\rho_{ds}$ for $\rho$ is $d_0 \cdots d_{n-1}$, where $d_i = \rho_{ds}(i)=(v_0, \intstates(v_0\cdots v_{i+1}),v_{i+1})$, for $i\in\{0,\ldots,n-1\}$. 
We denote by $DElm(\rho_{ds})$ the set of descriptor elements occurring in $\rho_{ds}$.
\end{definition}

\begin{figure}[t]
\centering
\begin{tikzpicture}[->,>=stealth',shorten >=1pt,auto,node distance=2cm,thick,main node/.style={circle,draw}]
    \node[main node,style={double}] (0) {$v_0$};
    \node[main node] (1) [above right of=0] {$v_1$};
    \node[main node] (2) [right of=0] {$v_2$};
    \node[main node] (3) [right of=1] {$v_3$};
  \path[every node/.style={font=\small}]
    (0) edge [loop left] (0)
    (0) edge (1)
        edge (2)
    (1) edge (2)
        edge (3)
    (2) edge [bend left] (1)
        edge [bend left] (3)
    (3) edge [loop right] (3)
    (3) edge (2)
    (2) edge [loop right] (2)
    ;
\end{tikzpicture}
\caption{An example of finite Kripke structure.}\label{ExK1}
\end{figure}
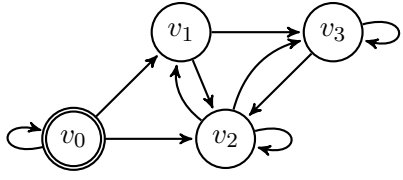

As an example, let us consider the finite Kripke structure of Figure \ref{ExK1} and the track $\rho = v_0v_0v_0v_1v_2v_1v_2v_3v_3v_2v_2$. The descriptor sequence for $\rho$ is:
\begin{multline}
	\rho_{ds}=(v_0,\emptyset , v_0)\boxed{(v_0,\{v_0\},v_0)}(v_0,\{v_0\},v_1)(v_0,\{v_0,v_1\},v_2)
	 \boxed{(v_0,\Gamma,v_1)(v_0,\Gamma,v_2)}\\ (v_0,\Gamma,v_3) \boxed{(v_0,\Delta,v_3)(v_0,\Delta,v_2)(v_0,\Delta,v_2)},
	\tag{*}\label{dsexample}
\end{multline}
where $\Gamma\!=\!\{v_0,v_1,v_2\}$, $\Delta\!=\!\{v_0,v_1,v_2,v_3\}$, and 
$DElm(\rho_{ds})$ is the set $\{(v_0,\emptyset , v_0),(v_0,\{v_0\},v_0),\allowbreak
(v_0,\{v_0\},v_1),\, (v_0,\{v_0,v_1\},v_2),\,
	(v_0,\Gamma,v_1),\, (v_0,\Gamma,v_2),\, (v_0,\Gamma,v_3),\, (v_0,\Delta,v_2),\, (v_0,\Delta,v_3)\}$.
The meaning of boxes in (\ref{dsexample}) will be clear later.

To express the relationships between descriptor elements occurring in a descriptor sequence,
we introduce a binary relation $\Rt$. Intuitively, given two descriptor elements $d'$ and $d''$
of a descriptor sequence, the relation $d'\Rt d''$ holds if $d'$ and $d''$ are the descriptor elements of two 
tracks $\rho'$ and $\rho''$, respectively, and  $\rho'$ is a prefix of $\rho''$.
\begin{definition}
Let $\rho_{ds}$ be the descriptor sequence for a track $\rho$ and let
$d'=(v_{in},S',v_{fin}')$ and $d''=(v_{in},S'',v_{fin}'')$ be two descriptor elements in $\rho_{ds}$. It holds that
$d'\Rt d''$ if and only if $S'\cup\{v_{fin}'\}\subseteq S''$.
\end{definition}
 
Note that the relation $\Rt$ is transitive. In fact for all 
descriptor elements $d'\!=\!(v_{in},S',v_{fin}')$, $d''=(v_{in},S'',v_{fin}'')$ and $d'''=(v_{in},S''',v_{fin}''')$, if $d'\Rt d''$ and $d''\Rt d'''$, then $S'\cup\{v_{fin}'\}\subseteq S''$ and $S''\cup\{v_{fin}''\}\subseteq S'''$; it 
follows that $S'\cup\{v_{fin}'\} \subseteq S'''$, and thus $d'\Rt d'''$. 
The relation $\Rt$ is neither an equivalence relation nor a quasiorder, since $\Rt$ is neither 
reflexive (e.g., $(v_0,\{v_0\},v_1)\notRt (v_0,\{v_0\},v_1)$), nor symmetric 
(e.g., $(v_0,\{v_0\},v_1)\Rt (v_0,\{v_0,v_1\},v_1)$ and $(v_0, \{v_0,v_1\},v_1)\notRt 
(v_0,\{v_0\},v_1)$), nor antisymmetric (e.g., $(v_0,\{v_1,v_2\},v_1)\Rt (v_0,\{v_1,v_2\},
v_2)$ and $(v_0,\{v_1,v_2\},v_2)\Rt (v_0,\{v_1,v_2\},v_1)$, but the two elements are distinct).

It can be easily shown that $\Rt$ 
pairs descriptor elements of increasing prefixes of 
a track.
\begin{proposition}
Let $\rho_{ds}$ be the descriptor sequence for the track $\rho=v_0v_1\cdots v_n$. Then, \linebreak $\rho_{ds}(i)\Rt \rho_{ds}(j)$, for all $0\leq i<j<n$.
\end{proposition}

We now partition descriptor elements into two different types.
\begin{definition}
A descriptor element $(v_{in},S,v_{fin})$ is a Type-1 descriptor element if $v_{fin}\notin S$,
while it is a Type-2 descriptor element if  $v_{fin}\in S$. 
\end{definition}
It can be easily checked that a descriptor element $d=(v_{in},S,v_{fin}$) is Type-1 if and only if $\Rt$ is not reflexive for  $d$. In fact, if $d\notRt d$, then $S\cup\{v_{fin}\}\not\subseteq S$, and thus $v_{fin}\notin S$. Conversely, if $v_{fin}\notin S$, then $d\notRt d$. It follows that a Type-1 descriptor element cannot occur more than once in a descriptor sequence. On the other hand, Type-2 descriptor elements may occur multiple times, and if a descriptor element occurs more than once in a descriptor sequence, then it is necessarily of Type-2.
\begin{proposition}
If both $d'\Rt d''$ and $d''\Rt d'$, for $d' = (v_{in},S',v'_{fin})$ and $d'' = (v_{in},S'',v''_{fin})$, then $v_{fin}'\in S'$, $v_{fin}''\in S''$, and  $S'=S''$; thus, both $d'$ and $d''$ are Type-2 descriptor elements. 
\end{proposition}

We are now ready to give a general characterization of the descriptor sequence $\rho_{ds}$ 
for a track $\rho$: $\rho_{ds}$ is composed of some (maximal) subsequences, consisting 
of occurrences of Type-2 descriptor elements on which $\Rt$ is symmetric, separated by 
occurrences of Type-1 descriptor elements. 
This can be formalized by means of the following notion of cluster.
\begin{definition}
A cluster $\mathpzc{C}$ of (Type-2) descriptor elements is a maximal set of descriptor elements $\{d_1,\ldots , d_s\}\subseteq DElm(\rho_{ds})$ such that  $d_i\Rt d_j$ and $d_j\Rt d_i$ for all $i,j\in\{1,\ldots , s\}$.
\end{definition}
Thanks to maximality, clusters are pairwise disjoint: if $\mathpzc{C}$ and $\mathpzc{C}'$ are distinct clusters, $d\in\mathpzc{C}$ and $d'\in\mathpzc{C}'$, either $d\Rt d'$ and $d'\notRt d$, or $d'\Rt d$ and $d\notRt d'$.

It can be easily checked that
    the descriptor elements of a cluster $\mathpzc{C}$ are contiguous in $\rho_{ds}$ (in other words, they
    form a subsequence of $\rho_{ds}$), that is, 
    occurrences of descriptor elements of $\mathpzc{C}$ are never shuffled with 
    occurrences of descriptor elements not belonging to $\mathpzc{C}$.
\begin{definition}
Let $\rho_{ds}$ be a descriptor sequence and $\mathpzc{C}$ be one of its clusters. The subsequence of $\rho_{ds}$ associated with $\mathpzc{C}$ is the subsequence $\rho_{ds}(i,j)$, with $i\leq j < |\rho_{ds}|$, including all and only the occurrences of the descriptor elements in $\mathpzc{C}$. 
\end{definition}    
Note that two subsequences associated with two distinct clusters $\mathpzc{C}$ and $\mathpzc{C}'$ in a descriptor sequence must be separated by at least one occurrence of a Type-1 descriptor element. 
For instance, 
with reference to the descriptor sequence (\ref{dsexample}) for the track $\rho=v_0v_0v_0v_1v_2v_1v_2v_3v_3v_2v_2$ of the Kripke structure in Figure~\ref{ExK1}, the subsequences associated with clusters are enclosed in boxes.

While $\Rt$ allows us to order any pair of Type-1 descriptor elements, as well as any Type-1 descriptor element with respect to a Type-2 one, it does not give us any means to order Type-2 descriptor elements belonging to the same cluster. This, together with the fact that Type-2 elements may have multiple occurrences in a descriptor sequence, implies that we need to somehow limit the 
number of occurrences of Type-2 elements in order to determine a bound on the length of track representatives of $B_{k}$-descriptors. 

To this end, we introduce an equivalence relation that allows us to put together indistinguishable
occurrences of the same descriptor element in a descriptor sequence, that is, 
to detect those occurrences which are associated with prefixes of the track
with the same $B_{k}$-descriptor. 
The idea is that a track representative for a $B_{k}$-descriptor should not feature indistinguishable occurrences of the same descriptor element.

\begin{definition}\label{def:k-indist}
Let  $\rho_{ds}$ be a descriptor sequence and $k \geq 1$. We say that $\rho_{ds}(i)$ and $\rho_{ds}(j)$, with $0 \leq i<j<|\rho_{ds}|$, are $k$-indistinguishable if (and only if) they are occurrences \emph{of the same descriptor element} $d$ and:
\begin{itemize}
\item (for $k = 1$) $DElm(\rho_{ds}(0, i-1))=DElm(\rho_{ds}(0, j-1))$;
\item (for $k \geq 2$) for all $i \leq \ell \leq j-1$, there exists $0\leq \ell'\leq i-1$ such that $\rho_{ds}(\ell)$ and $\rho_{ds}(\ell')$ are $(k-1)$-indistinguishable.
\end{itemize}
\end{definition}
From Definition \ref{def:k-indist}, it 
follows that two indistinguishable occurrences $\rho_{ds}(i)$ and $\rho_{ds}(j)$ of the same descriptor element necessarily belong to the same subsequence of $\rho_{ds}$ associated with a cluster. 

In general, it is always the case that $DElm(\rho_{ds}(0, i-1))\subseteq DElm(\rho_{ds}(0, j-1))$ for $i<j$. Moreover, note that the two first occurrences of a descriptor element, say $\rho_{ds}(i)$ and $\rho_{ds}(j)$, with $i<j$, are never 1-indistinguishable as a consequence of the fact that \mbox{1-indistinguishability} requires that $DElm(\rho_{ds}(0, i-1))= DElm(\rho_{ds}(0, j-1))$.
 
Proposition \ref{prop:property1} and \ref{propA} state some basic properties of the $k$-indistinguishability relation.
\begin{proposition} \label{prop:property1}
Let $k\geq 2$ and $\rho_{ds}(i)$ and $\rho_{ds}(j)$, with $0 \leq i<j<|\rho_{ds}|$, be two \mbox{$k$-indistinguishable} occurrences of the same descriptor element in a descriptor sequence
$\rho_{ds}$. Then, $\rho_{ds}(i)$ and $\rho_{ds}(j)$ are also $(k-1)$-indistinguishable. 
\end{proposition}
\begin{proof}
The proof is by induction on $k \geq 2$.\\
%
\emph{Base case} ($k= 2$). Let $\rho_{ds}(i)$ and $\rho_{ds}(j)$ be two $2$-indistinguishable occurrences of a descriptor element $d$. By definition, for any $\rho_{ds}(i')$, with $i\leq i'<j$, an occurrence of the descriptor element $d'=\rho_{ds}(i')$ must exist before position $i$, and thus $DElm(\rho_{ds}(0,i-1))=DElm(\rho_{ds}(0,j-1))$. It immediately follows that $\rho_{ds}(i)$ and $\rho_{ds}(j)$ are $1$-indistinguishable.\\
%
\emph{Inductive step} ($k\geq 3$). By definition, for all $i\leq \ell \leq j-1$, there exists $0\leq \ell' \leq  i-1$ such that $\rho_{ds}(\ell)$ and $\rho_{ds}(\ell')$ are $(k-1)$-indistinguishable. By the inductive hypothesis, $\rho_{ds}(\ell)$ and $\rho_{ds}(\ell')$ are $(k-2)$-indistinguishable, which implies that $\rho_{ds}(i)$ and $\rho_{ds}(j)$ are $(k-1)$-indistinguishable.
\end{proof}

\begin{proposition}\label{propA}
Let $k\geq 1$ and $\rho_{ds}(i)$ and $\rho_{ds}(m)$, with $0\! \leq\! i\!<\!m\!<\!|\rho_{ds}|$, be two $k$-indistinguishable occurrences of the same descriptor element in a descriptor sequence $\rho_{ds}$.
If $\rho_{ds}(j)=\rho_{ds}(m)$, for some $i<j<m$, then $\rho_{ds}(j)$ and $\rho_{ds}(m)$ are also $k$-indistinguishable.
\end{proposition}
\begin{proof}
For $k=1$, we have $DElm(\rho_{ds}(0,i-1))=DElm(\rho_{ds}(0,m-1))$; moreover, $DElm(\rho_{ds}(0,i-1))\subseteq DElm(\rho_{ds}(0,j-1))\subseteq DElm(\rho_{ds}(0,m-1))$. Thus $DElm(\rho_{ds}(0,i-1))=DElm(\rho_{ds}(0, m-1))= DElm(\rho_{ds}(0,j-1))$,
proving the property.

If $k\geq 2$, all occurrences $\rho_{ds}(i')$, with $i\leq i'<m$, are $(k-1)$-indistinguishable from some occurrence of the same descriptor element before $i$, by hypothesis. In particular, this is true for all occurrences $\rho_{ds}(j')$, with $j\leq j'<m$. The thesis trivially follows.
\end{proof}

\begin{example}\label{Example:DesEqFig}
In Figure~\ref{DesEqFig}, we give some examples of $k$-indistinguishability relations, for $k \in \{1,2,3\}$, considering the track $\rho=v_0v_1v_2v_3 v_3v_2v_3v_3 v_2v_3v_2v_3 v_3v_2v_3v_2 v_1v_3v_2v_3 v_2v_1v_2v_1 v_3\allowbreak v_2v_2v_3v_2$ of the finite Kripke structure depicted in Figure \ref{ExK1}. The track $\rho$ generates the descriptor sequence $\rho_{ds}=(v_0,\emptyset,v_1) (v_0,\{v_1\},v_2) (v_0,\{v_1,v_2\},v_3)abaababaababcababcbcabbab$, where $a$, $b$, and $c$ stand for $(v_0,\{v_1,v_2,v_3\},v_3)$, $(v_0,\{v_1,v_2,v_3\},v_2)$, and  $(v_0,\{v_1,v_2,v_3\},v_1)$, respectively. The figure shows the subsequence $\rho_{ds}(3, |\rho_{ds}|-1)$ associated with the cluster $\mathpzc{C}=\{a,b,c\}$. Pairs of $k$-indistinguishable consecutive occurrences of descriptor elements are connected by a rounded edge labelled  by $k$. Edges labelled by $\times$ link occurrences which are not $1$-indistinguishable. The values of all missing edges can easily be derived using the property stated by Corollary~\ref{propC} below. 
The meaning of numerical strings at the bottom of the figure will be clear later.
\end{example}

\begin{figure}[tbp]
\centering
\resizebox{\textwidth}{!}{\input{exDesEq}}
\caption{Examples of $k$-indistinguishability relations.
}\label{DesEqFig}
\end{figure}
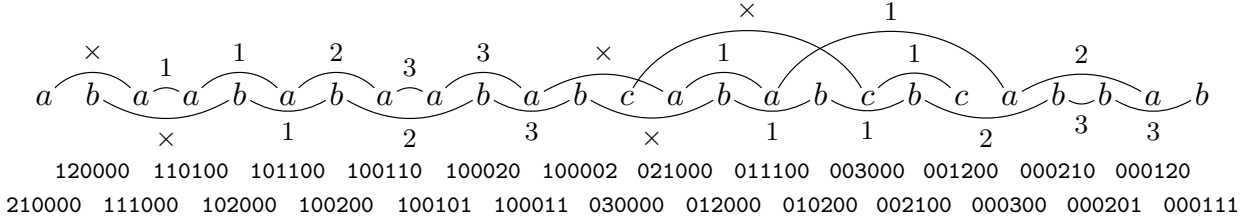

The next theorem establishes a fundamental connection between  
\mbox{$k$-indistinguishability} of descriptor elements and $k$-descriptor equivalence of tracks.

\begin{theorem}\label{teokequiv}
Let $\rho_{ds}$ be the descriptor sequence for a track $\rho$. Two occurrences $\rho_{ds}(i)$ and $\rho_{ds}(j)$, with $0\leq i <j<|\rho_{ds}|$, of the same descriptor element are $k$-indistinguishable if and only if $\rho(0, i+1)\sim_k\rho(0,j+1)$.
\end{theorem}
\begin{proof}
Let us assume that $\rho_{ds}(i)$ and $\rho_{ds}(j)$, with $i<j$, are $k$-indistinguishable. We prove by induction on $k \geq 1$ that $\rho(0, i+1)$ and $\rho(0,j+1)$ are associated with the same $B_k$-descriptor.

\emph{Base case} ($k=1$). Since $\rho_{ds}(i)$ and $\rho_{ds}(j)$ are occurrences of the same descriptor element, the roots of the $B_1$-descriptors for $\rho(0,i+1)$ and for $\rho(0,j+1)$ are labelled by the same descriptor element. Moreover, for each leaf of the $B_1$-descriptor for $\rho(0, i+1)$ there is a leaf of the $B_1$-descriptor for $\rho(0, j+1)$ with the same label, and vice versa, as by 1-indistinguishability $DElm(\rho_{ds}(0,i-1))=DElm(\rho_{ds}(0,j-1))$.

\emph{Inductive step} ($k\geq 2$). Since all the prefixes of $\rho(0,i+1)$ are also prefixes of $\rho(0,j+1)$, we just need to focus on the prefixes $\rho(0,t)$, with $i+1 \leq t \leq j$, in order to show that $\rho(0, i+1)$ and $\rho(0,j+1)$ have the same $B_k$-descriptor. By definition, any occurrence $\rho_{ds}(i')$ with $i\leq i'<j$, is $(k-1)$-indistinguishable from another occurrence $\rho_{ds}(i'')$, with $i''<i$, of the same descriptor element. By the inductive hypothesis, $\rho(0, i'+1)$ and $\rho(0, i''+1)$ are associated with the same $B_{k-1}$-descriptor. It follows that, for any proper prefix of $\rho(0, j+1)$ (of length at least 2), there exists a proper prefix of $\rho(0, i+1)$ with the same $B_{k-1}$-descriptor, which implies that the tracks $\rho(0,i+1)$ and $\rho(0,j+1)$ are associated with the same $B_{k}$-descriptor.
    
Conversely, we prove by induction on $k\geq 1$ that if $\rho_{ds}(i)$ and $\rho_{ds}(j)$, with $i<j$, are \emph{not} $k$-indistinguishable, then the $B_k$-descriptors for $\rho(0, i+1)$ and $\rho(0, j+1)$ 
are different.
We assume $\rho_{ds}(i)$ and $\rho_{ds}(j)$ to be occurrences of the same descriptor element (if this was not the case, the thesis would trivially follow, since the roots of the $B_k$-descriptors for $\rho(0,i+1)$ and $\rho(0,j+1)$ would be labelled by different descriptor elements).

\emph{Base case} ($k=1$). If  $\rho_{ds}(i)$ and $\rho_{ds}(j)$, with $i < j$, are \emph{not} $1$-indistinguishable, $DElm(\rho_{ds}(0,i-1))\subset DElm(\rho_{ds}(0,j-1))$. Hence, there is $d\in DElm(\rho_{ds}(0,j-1))$ such that $d\notin DElm(\rho_{ds}(0,i-1))$, and thus the $B_1$-descriptor for $\rho(0,j+1)$ has a leaf labelled by $d$ which is not present in the $B_1$-descriptor for $\rho(0,i+1)$.
         
        
\emph{Inductive step} ($k\geq 2$). If $\rho_{ds}(i)$ and $\rho_{ds}(j)$, with $i<j$, are \emph{not} $k$-indistinguishable, then there exists (at least) one occurrence $\rho_{ds}(i')$, with $i\leq i'<j$, of a descriptor element $d$ which is \emph{not} $(k-1)$-indistinguishable from any occurrence of $d$ before position $i$. By the inductive hypothesis, $\rho(0,i'+1)$ is associated to a $B_{k-1}$-descriptor which is not equal to any $B_{k-1}$-descriptors associated with proper prefixes of $\rho(0, i+1)$. Thus, in the $B_k$-descriptor for $\rho(0,j+1)$ there exists a subtree of depth $k-1$ such that there is no isomorphic subtree of depth $k-1$ in the $B_k$-descriptor for $\rho(0,i+1)$. 
\end{proof}

Note that $k$-indistinguishability between occurrences of descriptor elements is defined \emph{only for pairs of prefixes of the same track}, while the relation of $k$-descriptor equivalence can be applied to pairs of any tracks of a Kripke structure.

The next corollary easily follows from Theorem \ref{teokequiv}.

\begin{corollary}\label{propC}
Let $\rho_{ds}(i)$, $\rho_{ds}(j)$, and $\rho_{ds}(m)$, with $0 \leq i<j<m<|\rho_{ds}|$, be three occurrences of the same descriptor element in a descriptor sequence $\rho_{ds}$. If both the pair $\rho_{ds}(i)$ and $\rho_{ds}(j)$ and the pair $\rho_{ds}(j)$ and $\rho_{ds}(m)$ are $k$-indistinguishable, for some $k\geq 1$, then $\rho_{ds}(i)$ and $\rho_{ds}(m)$ are also $k$-indistinguishable. 
\end{corollary}

%% file: exDesEq.tex
\begin{tikzpicture}[node distance=3mm]
\begin{scope}[inner sep=1pt]
\node (1) at (0,0) {$a$};
\node (2) [right=of 1.south east,anchor=south west] {$b$};
\node (3) [right=of 2.south east,anchor=south west] {$a$};
\node (4) [right=of 3.south east,anchor=south west] {$a$};
\node (5) [right=of 4.south east,anchor=south west] {$b$};
\node (6) [right=of 5.south east,anchor=south west] {$a$};
\node (7) [right=of 6.south east,anchor=south west] {$b$};
\node (8) [right=of 7.south east,anchor=south west] {$a$};
\node (9) [right=of 8.south east,anchor=south west] {$a$};
\node (10) [right=of 9.south east,anchor=south west] {$b$};
\node (11) [right=of 10.south east,anchor=south west] {$a$};
\node (12) [right=of 11.south east,anchor=south west] {$b$};
\node (13) [right=of 12.south east,anchor=south west] {$c$};
\node (14) [right=of 13.south east,anchor=south west] {$a$};
\node (15) [right=of 14.south east,anchor=south west] {$b$};
\node (16) [right=of 15.south east,anchor=south west] {$a$};
\node (17) [right=of 16.south east,anchor=south west] {$b$};
\node (18) [right=of 17.south east,anchor=south west] {$c$};
\node (19) [right=of 18.south east,anchor=south west] {$b$};
\node (20) [right=of 19.south east,anchor=south west] {$c$};
\node (21) [right=of 20.south east,anchor=south west] {$a$};
\node (22) [right=of 21.south east,anchor=south west] {$b$};
\node (23) [right=of 22.south east,anchor=south west] {$b$};
\node (24) [right=of 23.south east,anchor=south west] {$a$};
\node (25) [right=of 24.south east,anchor=south west] {$b$};
\end{scope}

\begin{scope}[font=\footnotesize]
\path   (1) edge [bend left,out=45,in=135] node [above] {$\times$} (3)
        (2) edge [bend right] node [below] {$\times$} (5)
        (3) edge [bend left] node [above] {1} (4)
        (4) edge [bend left,out=45,in=135] node [above] {1} (6)
        (6) edge [bend left,out=45,in=135] node [above] {2} (8)
        (8) edge [bend left] node [above] {3} (9)
        (9) edge [bend left,out=45,in=135] node [above] {3} (11)
        (5) edge [bend right] node [below] {1} (7)
        (7) edge [bend right] node [below] {2} (10)
        (10) edge [bend right] node [below] {3} (12)
        
        (11) edge [bend left] node [above] {$\times$} (14)
        (12) edge [bend right] node [below] {$\times$} (15)
        (14) edge [bend left,out=45,in=135] node [above] {1} (16)
        (16) edge [out=60,in=120] node [above] {1} (21)
        (21) edge [bend left] node [above] {2} (24)
        (15) edge [bend right] node [below] {1} (17)
        (17) edge [bend right] node [below] {1} (19)
        (19) edge [bend right] node [below] {2} (22)
        (22) edge [bend right] node [below] {3} (23)
        (23) edge [bend right] node [below] {3} (25)
        (13) edge [out=60,in=120] node [above] {$\times$} (18)
        (18) edge [bend left,out=45,in=135] node [above] {1} (20);
\end{scope}

\begin{scope}[inner sep=1pt,font=\ttfamily\scriptsize]
\node [below=1cm of 1] {210000};
\node [below=0.6cm of 2] {120000};
\node [below=1cm of 3] {111000};
\node [below=0.6cm of 4] {110100};
\node [below=1cm of 5] {102000};
\node [below=0.6cm of 6] {101100};
\node [below=1cm of 7] {100200};
\node [below=0.6cm of 8] {100110};
\node [below=1cm of 9] {100101};
\node [below=0.6cm of 10] {100020};
\node [below=1cm of 11] {100011};
\node [below=0.6cm of 12] {100002};
\node [below=1cm of 13] {030000};
\node [below=0.6cm of 14] {021000};
\node [below=1cm of 15] {012000};
\node [below=0.6cm of 16] {011100};
\node [below=1cm of 17] {010200};
\node [below=0.6cm of 18] {003000};
\node [below=1cm of 19] {002100};
\node [below=0.6cm of 20] {001200};
\node [below=1cm of 21] {000300};
\node [below=0.6cm of 22] {000210};
\node [below=1cm of 23] {000201};
\node [below=0.6cm of 24] {000120};
\node [below=1cm of 25] {000111};
\end{scope}

\end{tikzpicture}

%% file: section02.tex
\section{A model checking procedure for $\AAbarBBbarEbar$ based on track representatives} \label{sec:representatives}

In this section, we will exploit the $k$-indistinguishability relation(s) between descriptor elements 
in a descriptor sequence $\rho_{ds}$ for a track $\rho$ to possibly replace $\rho$ by a $k$-descriptor equivalent, \emph{shorter} track $\rho'$ of bounded length. This allows us to find, for each 
$B_k$-descriptor $\mathpzc{D}_{B_k}$ (witnessed by a track of a finite Kripke structure $\mathpzc{K}$), a \emph{track representative} $\tilde{\rho}$ in $\mathpzc{K}$ such that $(i)$ $\mathpzc{D}_{B_k}$ is the $B_k$-descriptor
for $\tilde{\rho}$ and $(ii)$ the length of $\tilde{\rho}$ is bounded. Thanks to property $(ii)$, we can check all the track representatives of a finite Kripke structure by simply visiting its unravelling up to a bounded depth.

The notion of track representative can be explained as follows. Let $\rho_{ds}$ be the descriptor sequence for a track $\rho$. If there are two occurrences of the same descriptor element $\rho_{ds}(i)$ and $\rho_{ds}(j)$, with $i<j$, which are $k$-indistinguishable---let $\rho= \rho(0,j+1)\cdot\overline{\rho}$, with $\overline{\rho}=\rho(j+2,|\rho|-1)$---then we can replace 
$\rho$ by the $k$-descriptor equivalent, shorter track $\rho(0, i+1)\cdot\overline{\rho}$.
By Theorem \ref{teokequiv}, $\rho(0, i+1)$ and $\rho(0, j+1)$ have the same $B_k$-descriptor and thus, by Proposition \ref{extBk}, 
$\rho=\rho(0, j+1)\cdot\overline{\rho}$ and $\rho(0, i+1)\cdot\overline{\rho}$ have the same $B_k$-descriptor. 
Moreover, since $\rho_{ds}(i)$ and $\rho_{ds}(j)$ are occurrences of the same descriptor element, $\rho(i+1)=\rho(j+1)$ and thus the track $\rho(0,i+1)\cdot\overline{\rho}$ is witnessed in the finite Kripke structure.
By iteratively applying such a \emph{contraction method}, we can find a track $\rho'$ which is $k$-descriptor equivalent to $\rho$, whose descriptor sequence is devoid of $k$-indistinguishable occurrences of descriptor elements. \emph{A track representative} is a track that fulfils this property.

We now show how to calculate a bound to the length of track representatives.
We start by stating some technical properties. The next proposition provides a bound to the distance within which we necessarily observe a repeated occurrence of some descriptor element in the descriptor sequence for a track. We preliminarily observe that, for any track $\rho$, $|DElm(\rho_{ds})|\leq 1+|W|^2$, where $W$ is the set of states of the finite Kripke structure. Indeed, in the descriptor sequence, the sets of internal states of prefixes of $\rho$ increase monotonically with respect to the ``$\subseteq$'' relation. As a consequence, at most $|W|$ distinct sets may occur---excluding $\emptyset$ which can occur only in the first descriptor element. Moreover, these sets can be paired with all possible final states, which are at most $|W|$.

\begin{proposition}\label{propzerohalt}
For each track $\rho$ of $\mathpzc{K}$, associated with a descriptor element $d$, there exists a track $\rho'$ of $\mathpzc{K}$, associated with the same descriptor element $d$, such that $|\rho'|\leq 2+|W|^2$.
\end{proposition}
\begin{proof}
By induction on the length $\ell \geq 2$ of $\rho$.\\
\emph{Base case} ($\ell=2$). The track $\rho$ satisfies the condition $\ell\leq 2+|W|^2$.\\
\emph{Inductive step} ($\ell>2$). We distinguish two cases. If $\rho_{ds}$ has no duplicated occurrences of the same descriptor element, then $|\rho_{ds}|\leq 1+|W|^2$, since $|DElm(\rho_{ds})|\leq 1+|W|^2$, and thus 
$\ell\leq 2+|W|^2$ (the length of $\rho$ is equal to the length of $\rho_{ds}$ plus $1$). 

On the other hand,
if $\rho_{ds}(i)=\rho_{ds}(j)$, for some $0\leq i<j<|\rho_{ds}|$, $\rho(0,i+1)$ and $\rho(0,j+1)$ are associated with the same descriptor element. Now, $\rho'=\rho(0,i+1)\cdot\rho(j+2, |\rho|-1)$ is a track of $\mathpzc{K}$ since $\rho(i+1)=\rho(j+1)$, and, by Proposition \ref{extBk}, $\rho=\rho(0, j+1)\cdot\rho(j+2, |\rho|-1)$ and $\rho'$ are associated with the same descriptor element. By the inductive hypothesis, there exists a track $\rho''$ of $\mathpzc{K}$, associated with the same descriptor element of $\rho'$ (and of $\rho$), with $|\rho''|\leq 2+|W|^2$.
\end{proof}

Proposition \ref{propzerohalt} will be used in the following unravelling Algorithm~\ref{unr} as a termination criterion (referred to as \emph{0-termination criterion}) for unravelling a finite Kripke structure when it is not necessary to observe multiple occurrences of the same descriptor element:
\emph{to get a track representative for every descriptor element with initial state $v$, witnessed in a finite Kripke structure with set of states $W$, 
we can avoid considering tracks longer than $2+|W|^2$ while exploring the unravelling of the Kripke structure from $v$.}

Let us now consider the (more difficult) problem of establishing a bound 
for tracks devoid of pairs of $k$-indistinguishable occurrences of descriptor elements.
We first note that, in a descriptor sequence $\rho_{ds}$ for a track $\rho$, there are at most $|W|$ occurrences of Type-1 descriptor elements. On the other hand, Type-2 descriptor elements can occur multiple times and thus, 
to bound the length of $\rho_{ds}$, one has to 
constrain the \emph{number} and the \emph{length} of the subsequences of $\rho_{ds}$ associated with clusters. As for their number, it suffices to observe that they are separated by Type-1 descriptor elements, and hence at most $|W|$ of them, related to distinct clusters, can occur in a descriptor sequence. 

As for their length, we can proceed as follows. First, for any cluster $\mathpzc{C}$, it holds that $|\mathpzc{C}|\leq |W|$, as all (Type-2) descriptor elements of $\mathpzc{C}$ share the same set $S$ of internal states and their final states $v_{fin}$ must belong to $S$.
In the following, we consider the (maximal) subsequence $\rho_{ds}(u,v)$ of $\rho_{ds}$ associated with a specific cluster $\mathpzc{C}$, for some $0\leq u\leq v\leq |\rho_{ds}|-1$, and when we mention an index $i$, we implicitly assume that $u\leq i\leq v$, that is, $i$ refers to a position in the subsequence. 

We sequentially scan such a subsequence suitably recording the multiplicity of occurrences
of descriptor elements into an auxiliary structure. 
To detect indistinguishable occurrences of descriptor elements up to indistinguishability $s \geq 1$, we use $s + 3$ arrays  $Q_{-2}()$, $Q_{-1}()$, $Q_0()$, $Q_1()$, 
$\ldots$, $Q_s()$. Array elements are sets of descriptor elements of $\mathpzc{C}$: given an index $i$,  the sets at position $i$, $Q_{-2}(i)$, $Q_{-1}(i)$, $Q_0(i)$, $Q_1(i)$, 
$\ldots$, $Q_s(i)$, store information about indistinguishability for multiple occurrences of descriptor elements in the subsequence up to position $i>u$. 
To exemplify, if we find an occurrence of the descriptor element $d \in\mathpzc{C}$ at position $i$, that is, $\rho_{ds}(i)=d$, we have that:
\begin{enumerate}
	\item  $Q_{-2}(i)$ contains all descriptor elements of $\mathpzc{C}$ which have never occurred in $\rho_{ds}(u,i)$;
	\item  $d\in Q_{-1}(i)$ if $d$ has never occurred in $\rho_{ds}(u,i-1)$ and $\rho_{ds}(i)=d$, that is, $\rho_{ds}(i)$ is the first occurrence of $d$ in $\rho_{ds}(u,i)$;
	\item  $d \in Q_{0}(i)$ if $d$ occurs at least twice in $\rho_{ds}(u,i)$ and the occurrence $\rho_{ds}(i)$ of $d$ is \emph{not} 1-indistinguishable from the last occurrence of $d$ in $\rho_{ds}(u,i-1)$;
	\item $d \in Q_{t}(i)$ (for some $t\geq 1$) if the occurrence $\rho_{ds}(i)$ of $d$  is $t$-indistinguishable, but \emph{not also} $(t+1)$-indistinguishable, from the last occurrence of $d$ in $\rho_{ds}(u,i-1)$.
\end{enumerate}

In particular, at position $u$ (the first of the subsequence), $Q_{-1}(u)$ contains only the descriptor element $d =\rho_{ds}(u)$, $Q_{-2}(u)$ is the set $\mathpzc{C}\setminus\{d\}$, and $Q_0(u)$, $Q_1(u)$, $\dots$ are empty sets. 

Arrays $Q_{-2}()$, $Q_{-1}()$, $Q_0()$, $Q_1()$, $\ldots$, $Q_s()$ satisfy the following constraints:
    for all positions $i$, $\bigcup^s_{m=-2} Q_m(i)=\mathpzc{C}$ and,
    for all $i$ and all $m\neq m'$, $Q_m(i)\cap Q_{m'}(i)=\emptyset$.

Intuitively, at every position $i$, $Q_{-2}(i)$, $Q_{-1}(i), Q_0(i), Q_1(i)$, $\ldots$, $Q_s(i)$ describe a \emph{state} of the scanning process of the subsequence. 
The change of state produced by the transition from position $i-1$ to position $i$ while scanning the subsequence is formally defined by the function $f$, reported in Figure \ref{fig:f}, which maps the descriptor sequence $\rho_{ds}$ and a position $i$ to the tuple of sets $\big(Q_{-2}(i),Q_{-1}(i),Q_0(i), Q_1(i), \ldots, Q_s(i)\big)$.

\begin{figure}[tbp]
\centering
\fbox{%
\begin{minipage}{0.98\textwidth}
\begin{small}
$f(\rho_{ds},u)=\big(\mathpzc{C}\setminus\{d\}, \{d\}, \emptyset, \cdots , \emptyset\big) \text{ with } \rho_{ds}(u)=d$;

\medskip
For all $i>u$: $f(\rho_{ds},i)= \big(Q_{-2}(i),Q_{-1}(i),Q_0(i),\ldots ,Q_s(i)\big) =$
\[ \left\{
\begin{minipage}{0.96\textwidth}
$\big(Q_{-2}(i-1)\setminus\{d\},\{d\}\cup\bigcup^s_{m=-1}Q_m(i-1), \emptyset, \ldots , \emptyset\big)$ if $\rho_{ds}(i)$ is the first occurrence of $d$ in $\rho_{ds}(u,i)$; \textbf{(a)}\\

$\big(Q_{-2}(i-1),Q_{-1}(i-1)\setminus\{d\},\{d\}\cup\bigcup^s_{m=0}Q_m(i-1), \emptyset, \ldots , \emptyset\big)$ if $\rho_{ds}(i)=d$, $d\in Q_{-1}(i-1)$, and $\rho_{ds}(i)$ is at least the second occurrence of $d$ in $\rho_{ds}(u,i)$ and it is \emph{not} 1-indistinguishable from the immediately preceding occurrence of $d$; \textbf{(b)}\\

$\big(Q_{-2}(i-1),Q_{-1}(i-1),\{d\}\cup Q_0(i-1), Q_1(i-1)\setminus\{d\}, \ldots , Q_s(i-1)\setminus\{d\}\big)$ if $\rho_{ds}(i)=d$, $d\in \bigcup^s_{m=0} Q_{m}(i-1)$, and $\rho_{ds}(i)$ is at least the second occurrence of $d$ in $\rho_{ds}(u,i)$ and it is \emph{not} 1-indistinguishable from the immediately preceding occurrence of $d$; \textbf{(c)}\\

$\big(Q_{-2}(i-1)\setminus\{d\},\ldots,Q_{t-1}(i-1)\setminus\{d\},\{d\}\cup\bigcup^s_{m=t}Q_m(i-1), \emptyset, \ldots , \emptyset\big)$ if $\rho_{ds}(i)=d$, $\rho_{ds}(i)$ is $t$-indistinguishable (for some $t\geq 1$), but \emph{not} also $(t+1)$-indistinguishable, to the immediately preceding occurrence of 
$d$, and $d\in \bigcup^{t-1}_{m=-2}Q_{m}(i-1)$; \textbf{(d)}\\

$\big(Q_{-2}(i-1),\cdots,Q_{t-1}(i-1),\{d\}\cup Q_t(i-1), Q_{t+1}(i-1) \setminus\{d\}, \ldots , Q_s(i-1)\setminus\{d\}\big)$ if $\rho_{ds}(i)=d$, $\rho_{ds}(i)$ is $t$-indistinguishable (for some $t\geq 1$), but \emph{not} also $(t+1)$-indistinguishable, to the immediately preceding occurrence of $d$, and $d\in \bigcup^s_{m=t}Q_{m}(i-1)$. \textbf{(e)}
\end{minipage}\right. \]
\end{small}
\end{minipage}}
\caption{Definition of the scan function $f$.}\label{fig:f}
\end{figure}

Note that, whenever a descriptor element $\rho_{ds}(i)=d$ is such that $d\in Q_z(i-1)$ and $d\in Q_{z'}(i)$, with $z<z'$ (cases (a), (b), and (d) of the definition of $f$), all $Q_{z''}(i)$, with $z''>z'$, are empty sets and, for all $z''\geq z'$,  all elements in $Q_{z''}(i-1)$ belong to $Q_{z'}(i)$.
As an intuitive explanation, consider, for instance, the following scenario:
in a subsequence of $\rho_{ds}$, associated with some cluster $\mathpzc{C}$, $\rho_{ds}(h)=\rho_{ds}(i)=d\in\mathpzc{C}$ and $\rho_{ds}(h')=\rho_{ds}(i')=d'\in\mathpzc{C}$, for some $h<h'<i<i'$ and $d \neq d'$, and there are not other occurrences of $d$ and $d'$ in $\rho_{ds}(h,i')$. If $\rho_{ds}(h)$ and $\rho_{ds}(i)$ are exactly $z'$-indistinguishable, by definition of the indistinguishability relation, $\rho_{ds}(h')$ and $\rho_{ds}(i')$ can be no more than $(z'+1)$-indistinguishable. Thus, if $d'$ is in $Q_{z''}(i-1)$, for some $z''>z'$, we can safely ``downgrade'' it to $Q_{z'}(i)$, 
because we know that, when we meet the next occurrence of $d'$ ($\rho_{ds}(i')$), 
$\rho_{ds}(h')$ and $\rho_{ds}(i')$ will be no more than $(z'+1)$-indistinguishable.

In the following, we will make use of an abstract characterization of the state of arrays at a given position $i$, as determined by the scan function $f$, called \emph{configuration}, that   
only accounts for the cardinality of sets in arrays. Theorem \ref{thdesc} states that, when a descriptor 
subsequence is scanned, configurations never repeat, since the sequence of configurations is strictly decreasing according to the lexicographical order $>_{lex}$.
This property will allow us to establish the desired bound on the length of track representatives.

\begin{definition}\label{def:config} 
Let $\rho_{ds}$ be the descriptor sequence for a track $\rho$ and $i$ be a position in the subsequence of $\rho_{ds}$ associated with a given cluster. The \emph{configuration at position} $i$, denoted as $c(i)$, is the tuple
\[c(i)=(|Q_{-2}(i)|,|Q_{-1}(i)|,|Q_{0}(i)|,|Q_{1}(i)|,\cdots ,|Q_{s}(i)|),\]
where $f(\rho_{ds},i) = (Q_{-2}(i),Q_{-1}(i),Q_{0}(i),Q_{1}(i),\cdots ,Q_{s}(i))$.
\end{definition}
An example of a sequence of configurations is given in Figure~\ref{DesEqFig} of Example~\ref{Example:DesEqFig}, where, for each position in the subsequence $\rho_{ds}(3, |\rho_{ds}|-1)$, we give the associated configuration: $c(3)=(2,1,0,0,0,0)$, $c(4)=(1,2,0,0,0,0)$, and so forth.

\begin{theorem}\label{thdesc}
Let $\rho_{ds}$ be the descriptor sequence for a track $\rho$ and $\rho_{ds}(u,v)$, for some $u <v$, be the subsequence associated with a cluster $\mathpzc{C}$. For all $u<i\leq v$, if $\rho_{ds}(i)=d$, then it holds that $d\in Q_t(i-1)$, $d\in Q_{t+1}(i)$, for some $t\in\{-2,-1\}\cup\mathbb{N}$, and $c(i-1)>_{lex}c(i)$.
\end{theorem}
The proof is given in~\ref{thdescProof}.

We show now how to select all and only those tracks which do not feature any pair of $k$-indistinguishable occurrences of descriptor elements. To this end, we make use of a scan function $f$ which uses $k + 3$ arrays (the value $k+3$ accounts for the parameter $k$ of descriptor element indistinguishability, plus the three arrays $Q_{-2}()$, $Q_{-1}()$, $Q_0()$). Theorem \ref{thdesc} guarantees that, while scanning a subsequence, configurations never repeat. This allows us to set an upper bound to the length of a track such that, whenever exceeded, the descriptor sequence for the track features at least a pair of $k$-indistinguishable occurrences of some descriptor element. The bound is essentially given by the number of possible configurations for $k + 3$ arrays.

By an easy combinatorial argument, we can prove the following proposition.
\begin{proposition}\label{starsbars}
For all $n,t\in\mathbb{N}\setminus\{0\}$, the number of distinct $t$-tuples of natural 
numbers whose sum equals $n$ is 
$\varepsilon(n,t)=\binom{n+t-1}{n}=\binom{n+t-1}{t-1}$.
\end{proposition}
\begin{proof}
The following figure suggests an alternative representation of a tuple, in the form of a configuration of separators/bullets: \begin{center}
\fbox{
\begin{tabular}{cccccccccccccc}
$\circ$ & $\circ$ & $\circ$ & $\circ$ & $\circ$ & $\mid$ & $\circ$ & $\circ$ & $\circ$ & $\mid$ & $\circ$ & $\mid$ & $\mid$ & $\circ$  
\end{tabular}}
\hspace{0.5cm}
$\leftrightsquigarrow$ \hspace{0.5cm} $(5,3,1,0,1)$
\end{center} 
It can be easily checked that such a representation is \emph{unambiguous}, i.e., there exists a bijection between configurations of separators/bullets and tuples.

The sum of the natural numbers of the tuple equals the number of bullets, and the size of the tuple is the number of separators plus 1. Since there are $\varepsilon(n,t)=\binom{n+t-1}{t-1}$ distinct ways of choosing $t-1$ separators among $n+t-1$ different places---and places which are not chosen must contain bullets---there are exactly $\varepsilon(n,t)$ distinct $t$-tuples of natural numbers whose sum equals $n$.
\end{proof}
Proposition \ref{starsbars} provides two upper bounds for $\varepsilon(n,t)$: $\varepsilon(n,t)\leq (n+1)^{t-1}$ and $\varepsilon(n,t)\leq t^n$.

Since a configuration $c(i)$ of a cluster $\mathpzc{C}$ is a $(k+3)$-tuple whose
elements add up to $|\mathpzc{C}|$, by Proposition~\ref{starsbars} we conclude
that there are at most $\varepsilon(|\mathpzc{C}|,k+3)=\binom{|\mathpzc{C}|+k+2}{k+2}$ 
distinct configurations of size $(k+3)$, whose natural numbers add up to $|\mathpzc{C}|$.
Moreover, since configurations never repeat while scanning a subsequence associated with a cluster $\mathpzc{C}$, $\varepsilon(|\mathpzc{C}|,k+3)$ is an upper bound to the length of such a subsequence.

Now, for any track $\rho$, $\rho_{ds}$ features at most $|W|$ subsequences associated with distinct clusters $\mathpzc{C}_1,\mathpzc{C}_2,\dots$, and thus, if the following upper bound to the length of $\rho$ is exceeded, then there is at least one pair of $k$-indistinguishable occurrences of some descriptor element in $\rho_{ds}$:
$|\rho|\leq 1+(|\mathpzc{C}_1|+1)^{k+2}+(|\mathpzc{C}_2|+1)^{k+2}+\cdots + (|\mathpzc{C}_s|+1)^{k+2}+|W|$, where $s\leq |W|$, and the last addend is to count occurrences of Type-1 descriptor elements. 
Since clusters are disjoint, their union is a subset of $DElm(\rho_{ds})$, and $|DElm(\rho_{ds})|\leq 1+|W|^2$, we get:  
\begin{multline*}
|\rho|\leq 1+(|\mathpzc{C}_1| + |\mathpzc{C}_2| +\cdots + |\mathpzc{C}_s| + |W|)^{k+2}+|W|\leq 
 1+(|DElm(\rho_{ds})| + |W|)^{k+2}+|W| 
\\ \leq 1+(1+|W|^2 + |W|)^{k+2}+|W|\leq 1+(1+|W|)^{2k+4}+|W|.
\end{multline*}
Analogously, by using the alternative bound to $\varepsilon(|\mathpzc{C}|,k+3)$, we have that
\begin{multline*}
|\rho|\leq 1+(k+3)^{|\mathpzc{C_1}|}+(k+3)^{|\mathpzc{C_2}|}+ \cdots + (k+3)^{|\mathpzc{C_s}|} +|W|\leq 
1+(k+3)^{|\mathpzc{C_1}|+|\mathpzc{C_2}|+\cdots + |\mathpzc{C_s}|}+|W|
\\ \leq 1+(k+3)^{|DElm(\rho_{ds})|}+|W|\leq 1+(k+3)^{|W|^2+1}+|W|.
\end{multline*}
The upper bound for $|\rho|$ is then the least of the two given upper bounds:
\[
\tau(|W|,k)= \min\big\{1+(1+|W|)^{2k+4}+|W|, 1+(k+3)^{|W|^2+1}+|W|\big\}.
\]

\begin{theorem}\label{thbound}
Let $\mathpzc{K}=(\mathpzc{AP},W, \delta,\mu,w_0)$ be a finite Kripke structure and $\rho$ be a track in $\Trk_\mathpzc{K}$.
If $|\rho|>\tau(|W|,k)$, then there exists another track in $\Trk_\mathpzc{K}$, whose length is less than or equal to $\tau(|W|,k)$, associated with the same $B_k$-descriptor as $\rho$.	
\end{theorem}
\begin{proof}[Proof (sketch)]
If $|\rho|>\tau(|W|,k)$, then there exists (at least) a subsequence of $\rho_{ds}$, associated with some cluster $\mathpzc{C}$, which contains (at least) a pair of $k$-indistinguishable occurrences of some descriptor element $d\in\mathpzc{C}$, say $\rho_{ds}(i)$ and $\rho_{ds}(j)$, with $j<i$.
By Theorem \ref{teokequiv}, the two tracks $\tilde{\rho}_1=\rho(0, j+1)$ and $\tilde{\rho}_2=\rho(0, i+1)$ have the same $B_k$-descriptor. 
Now, let us rewrite the track $\rho$ as the concatenation $\tilde{\rho}_2\cdot\overline{\rho}$ for some $\overline{\rho}$. By Proposition \ref{extBk}, the tracks $\rho=\tilde{\rho}_2\cdot\overline{\rho}$ and $\rho'=\tilde{\rho}_1\cdot\overline{\rho}$ are associated with the same $B_k$-descriptor. Since $\lst(\tilde{\rho}_1)=\lst(\tilde{\rho}_2)$ ($\rho_{ds}(j)$ and $\rho_{ds}(i)$ are occurrences of the same descriptor element $d$), $\rho'=\tilde{\rho}_1\cdot\overline{\rho}$ is a track of $\mathpzc{K}$ shorter than $\rho$.
If $|\rho'|\leq\tau(|W|,k)$, we have proved the thesis; otherwise, we can iterate the process by applying the above contraction to $\rho'$.
\end{proof}

Theorem \ref{thbound} allows us to define a termination criterion to bound the depth of the unravelling of a finite Kripke structure ($(k\geq 1)$-\emph{termination criterion}), while searching for track representatives for witnessed $B_k$-descriptors:
\emph{for any $k\geq 1$, to get a track representative for every $B_k$-descriptor, with initial state $v$, and witnessed in a finite Kripke structure with set of states $W$, we can avoid taking into consideration tracks longer than $\tau(|W|,k)$ while exploring the unravelling of the structure from $v$.}

Thanks to the above results, we are now ready to define a model checking algorithm for $\AAbarBBbarEbar$ formulas.
First, we introduce the unravelling Algorithm~\ref{unr}, which explores the unravelling of the input Kripke structure $\mathpzc{K}$ to find track representatives for all witnessed $B_k$-descriptors. It features two modalities, \emph{forward mode} (which is active when its fourth parameter, direction, is \textsc{forw}) and \emph{backward mode} (active when the parameter direction is \textsc{backw}), in which the unravelling of $\mathpzc{K}$ is visited following the direction of edges and against their direction (that is equivalent to visiting the transposed graph $\overline{\mathpzc{K}}$ of $\mathpzc{K}$), respectively. In both cases, if there exist $k$-indistinguishable occurrences of a descriptor element in $\rho_{ds}$, the track $\rho$ is never returned.

\begin{algorithm}[tb]
\begin{minipage}{\textwidth}
\begin{algorithmic}[1]
	\If{direction = \textsc{forw}}
	    \State{Unravel $\mathpzc{K}$ starting from $v$ according to $\ll$}\Comment{``$\ll$'' is an arbitrary order of the nodes of $\mathpzc{K}$}
	    \State{For every new node \emph{of the unravelling} met during the visit, return the track $\rho$ from $v$ to the current node only if:}
	    \If{$k=0$}
	        \State{Apply the 0-termination criterion}
	    \Else
	        \If{The last descriptor element $d$ of (the descriptor sequence of) the current track $\rho$ is $k$-indistinguishable from a previous occurrence of $d$} 
	        \State{\emph{skip} $\rho$ and backtrack to $\rho(0, |\rho|-2)\cdot \overline{v}$, where $\overline{v}$ is the minimum state (w.r.t. $\ll$), greater than $\rho(|\rho|-1)$, such that $(\rho(|\rho|-2),\overline{v})$ is an edge of $\mathpzc{K}$.}
	        \EndIf	   
	    \EndIf 
	\ElsIf{direction = \textsc{backw}}
	    \State{Unravel  $\overline{\mathpzc{K}}$ starting from $v$ according to $\ll$}\Comment{$\overline{\mathpzc{K}}$ is $\mathpzc{K}$ with transposed edges}
	    \State{For every new node \emph{of the unravelling} met during the visit, consider the track $\rho$ \emph{from the current node to $v$}, and recalculate descriptor element indistinguishability from scratch (left to right); return the track only if:}
	    \If{$k=0$}
            \State{Apply the 0-termination criterion}
	    \Else
	        \If{There exist two $k$-indistinguishable occurrences of a descriptor element $d$ in (the descriptor sequence of) the current track $\rho$} 
	            \State{\emph{skip} $\rho$}
	        \EndIf
	    \EndIf
	    \State{Do not visit tracks of length greater than $\tau(|W|,k)$}
	\EndIf    
\end{algorithmic}
\end{minipage}
\caption{\texttt{Unrav}$(\mathpzc{K},v,k,\text{direction})$}\label{unr}
\end{algorithm}

In the \emph{forward mode} (which will be used to deal with $\hsA$ and $\hsBt$ modalities), the direction of track exploration and that of indistinguishability checking are the same, so we can stop extending a track as soon as the first pair of \mbox{$k$-indistinguishable} occurrences of a descriptor element is found in the descriptor sequence, suggesting an easy termination criterion for stopping the unravelling of tracks. In the \emph{backward mode} (used in the case of $\hsAt$ and $\hsEt$ modalities), such a straightforward criterion cannot be adopted, because tracks are explored right to left (the opposite direction with respect to edges of the Kripke structure), while  the indistinguishability relation over descriptor elements is computed left to right. In general, changing the prefix of a considered track requires recomputing from scratch the descriptor sequence and the indistinguishability relation over descriptor elements. 
In particular, $k$-indistinguishable occurrences of descriptor elements can be detected in the middle of a subsequence, and not necessarily at the end.
In this latter case, however,
the upper bound $\tau(|W|,k)$ on the maximum depth of the unravelling ensures the termination of the algorithm (line 17).

The next theorem proves soundness and completeness of Algorithm~\ref{unr} for the forward mode. The proof for the backward one is quite similar, and thus omitted.
\begin{theorem}\label{corrunr}
Let $\mathpzc{K}=(\mathpzc{AP},W,\delta,\mu,w_0)$ be a finite Kripke structure, $v\in W$, and $k\in\mathbb{N}$.
For every track $\rho$ of $\mathpzc{K}$, with $\fst(\rho)=v$ and $|\rho|\geq 2$, the unravelling Algorithm~\ref{unr} returns a track $\rho'$ of $\mathpzc{K}$, with $\fst(\rho')=v$, such that $\rho$ and $\rho'$ are associated with the same $B_k$-descriptor and $|\rho'|\leq \tau(|W|,k)$.
\end{theorem}
\begin{proof}


If $k=0$ the thesis follows immediately by the 0-termination criterion. So let us assume $k\geq 1$. The proof is by induction on $\ell=|\rho|$.

(Case $\ell=2$) In this case, $\rho_{ds}=(\fst(\rho),\emptyset,\lst(\rho))$, and the only descriptor element of the sequence is Type-1. Thus, $\rho$ itself is returned by the algorithm.

(Case $\ell>2$) If in $\rho_{ds}$ there are no pairs of $k$-indistinguishable occurrences of some descriptor element, the termination criterion of Algorithm \ref{unr} can never be applied. Thus, $\rho$ itself is returned (as soon as it is visited) and its length is at most $\tau(|W|,k)$.

Otherwise, the descriptor sequence of any track $\rho$ can be split into 3 parts: $\rho_{ds}=\rho_{ds1}\cdot\rho_{ds2}\cdot\rho_{ds3}$, where $\rho_{ds1}$ ends with a Type-1 descriptor element and it does not contain pairs of \mbox{$k$-indistinguishable} occurrences of any descriptor element; $\rho_{ds2}$ is a subsequence associated with a cluster $\mathpzc{C}$ of (Type-2) descriptor elements with at least a pair of \mbox{$k$-indistinguishable} occurrences of descriptor elements; $\rho_{ds3}$ (if it is not the empty sequence) begins with a Type-1 descriptor element. This amounts to say that 
$\rho_{ds2}$ is the ``leftmost'' subsequence of $\rho_{ds}$ consisting of elements of a cluster $\mathpzc{C}$, with at least a pair of $k$-indistinguishable occurrences of some descriptor element.

Therefore, there are two indexes $i,j$, with $j<i$, such that $\rho_{ds2}(j)$ and $\rho_{ds2}(i)$ are two $k$-indistinguishable occurrences of some $d\in\mathpzc{C}$ in $\rho_{ds}$. By Proposition~\ref{propA}, there exists a pair of indexes $i',j'$,  with $j'<i'$, such that $\rho_{ds2}(j')$ and $\rho_{ds2}(i')$ are two \emph{consecutive} $k$-indistinguishable occurrences of $d$ (by consecutive we mean that, for all $t\in[j'+1,i'-1]$, $\rho_{ds2}(t)\neq d$). If there are many such pairs (even for different elements in $\mathpzc{C}$), let us consider the one with the lower index $i'$ (namely, precisely the pair which is found earlier by the unravelling algorithm). By Theorem~\ref{teokequiv}, the two tracks associated with $\rho_{ds1}\cdot \rho_{ds2}(0, j')$ and $\rho_{ds1}\cdot \rho_{ds2}(0, i')$, say $\tilde{\rho}_1$ and $\tilde{\rho}_2$ respectively, have the same $B_k$-descriptor. Then, by Proposition \ref{extBk}, the tracks $\rho=\tilde{\rho}_2\cdot\overline{\rho}$ (for some $\overline{\rho}$) and $\rho'=\tilde{\rho}_1\cdot\overline{\rho}$ have the same $B_k$-descriptor.

Algorithm \ref{unr} does not return $\tilde{\rho}_2$ and, due to the backtrack step, neither $\rho=\tilde{\rho}_2\cdot\overline{\rho}$ is returned. 
But since $\lst(\tilde{\rho}_1)=\lst(\tilde{\rho}_2)$ ($\rho_{ds2}(j')$ and $\rho_{ds2}(i')$ are occurrences of the same descriptor element), the unravelling of $\mathpzc{K}$ features $\rho'=\tilde{\rho}_1\cdot\overline{\rho}$, as well. Now, by induction hypothesis, a track $\rho''$ of $\mathpzc{K}$ is returned, such that $\rho'$ and $\rho''$ have the same $B_k$-descriptor, and $|\rho''|\leq \tau(|W|,k)$. $\rho$ has in turn the same $B_k$-descriptor as $\rho''$.
\end{proof}

The above proof shows how a ``contracted variant'' of a track $\rho$ is (indirectly) computed by Algorithm~\ref{unr}.
As an example,  
$\rho'=v_0v_1v_2v_3v_3v_2v_3v_3v_2v_3v_2v_3v_2v_1v_3v_2v_3v_2v_1v_2v_1v_3v_2$
is returned by Algorithm~\ref{unr}
in place of the track $\rho$ of Example~\ref{Example:DesEqFig}, and 
it can be checked that $\rho'_{ds}$ does not contain any pair of $3$-indistinguishable occurrences of a descriptor element and that $\rho$ and $\rho'$ have the same $B_3$-descriptor.


%% file: section03.tex
\begin{algorithm}[p]
\begin{algorithmic}[1]
	\State{$k\gets \nestb(\psi)$}
	\State{$u\gets New\left(\texttt{Unrav}(\mathpzc{K},w_0,k,\textsc{forw})\right)$}\Comment{$w_0$ is the initial state of $\mathpzc{K}$}
	\While{$u.\texttt{hasMoreTracks()}$}
	    \State{$\tilde{\rho}\gets u.\texttt{getNextTrack()}$}
	    \If{$\texttt{Check}(\mathpzc{K},k,\psi,\tilde{\rho})=0$}
	        \Return{0: ``$\mathpzc{K},\tilde{\rho}\not\models \psi$''}
	    \EndIf
	\EndWhile
	\Return{1: ``$\mathpzc{K}\models \psi$''}	
\end{algorithmic}
\caption{\texttt{ModCheck}$(\mathpzc{K},\psi)$}\label{ModCheck2}
\end{algorithm}

\begin{algorithm}[p]
\resizebox{0.95\textwidth}{!}{
\begin{minipage}{\textwidth}
\begin{multicols}{2}
\begin{algorithmic}[1]
    \If{$\psi=\top$}
        \Return{1}
    \ElsIf{$\psi=\bot$}
        \Return{0}
	\ElsIf{$\psi=p\in\mathpzc{AP}$}
	    \If{$p\in \bigcap_{s\in \states(\tilde{\rho})}\mu(s)$}
	        \State{\textbf{return} 1 \textbf{else} \textbf{return} 0}
	    \EndIf
	\ElsIf{$\psi=\neg\varphi$}
	    \Return{1 $-$ $\texttt{Check}(\mathpzc{K},k,\varphi,\tilde{\rho})$}
    \ElsIf{$\psi=\varphi_1\wedge\varphi_2$}
        \If{$\texttt{Check}(\mathpzc{K},k,\varphi_1,\tilde{\rho})=0$}
	        \Return{0}
	    \Else
	        \Return{$\texttt{Check}(\mathpzc{K},k,\varphi_2,\tilde{\rho})$}
	    \EndIf
	\ElsIf{$\psi=\hsA\varphi$}
	    \State{$u\gets New\left(\texttt{Unrav}(\mathpzc{K},\lst(\tilde{\rho}),k,\textsc{forw})\right)$}
	    \While{$u.\texttt{hasMoreTracks()}$}
	        \State{$\rho\gets u.\texttt{getNextTrack()}$}
	        \If{$\texttt{Check}(\mathpzc{K},k,\varphi,\rho)=1$}
	            \Return{1}
	        \EndIf
	    \EndWhile
	    \Return{0}
	\ElsIf{$\psi=\hsAt\varphi$}
	    \State{$u\gets New\left(\texttt{Unrav}(\mathpzc{K},\fst(\tilde{\rho}),k,\textsc{backw})\right)$}
	    \While{$u.\texttt{hasMoreTracks()}$}
	        \State{$\rho\gets u.\texttt{getNextTrack()}$}
	        \If{$\texttt{Check}(\mathpzc{K},k,\varphi,\rho)=1$}
	            \Return{1}
	        \EndIf
	    \EndWhile
	    \Return{0}%
\columnbreak 
	\ElsIf{$\psi=\hsB\varphi$}
	    \For{each $\overline{\rho}$ prefix of $\tilde{\rho}$}\label{Bcase}
	        \If{$\texttt{Check}(\mathpzc{K},k-1,\varphi,\overline{\rho})=1$}
	            \Return{1}
	        \EndIf
	    \EndFor
	    \Return{0}
	\ElsIf{$\psi=\hsBt\varphi$}
	    \For{each $v\in W$ s.t. $(\lst(\tilde{\rho}),v)\in\delta$}
	        \If{$\texttt{Check}(\mathpzc{K},k,\varphi,\tilde{\rho}\cdot v)=1$}
	                \Return{1}
	        \EndIf
	        \State{$u\gets New\left(\texttt{Unrav}(\mathpzc{K},v,k,\textsc{forw})\right)$}
	        \While{$u.\texttt{hasMoreTracks()}$}
	            \State{$\rho\gets u.\texttt{getNextTrack()}$}
	            \If{$\texttt{Check}(\mathpzc{K},k,\varphi,\tilde{\rho}\cdot\rho)=1$}
	                \Return{1}
	            \EndIf
	        \EndWhile
	    \EndFor
	    \Return{0}
	\ElsIf{$\psi=\hsEt\varphi$}
	    \For{each $v\in W$ s.t. $(v,\fst(\tilde{\rho}))\in\delta$}
	        \If{$\texttt{Check}(\mathpzc{K},k,\varphi,v \cdot\tilde{\rho})=1$}
	                \Return{1}
	        \EndIf
	        \State{$u\gets New\left(\texttt{Unrav}(\mathpzc{K},v,k,\textsc{backw})\right)$}
	        \While{$u.\texttt{hasMoreTracks()}$}
	            \State{$\rho\gets u.\texttt{getNextTrack()}$}
	            \If{$\texttt{Check}(\mathpzc{K},k,\varphi,\rho\cdot\tilde{\rho})=1$}
	                \Return{1}
	            \EndIf
	        \EndWhile
	    \EndFor
	    \Return{0}
	\EndIf
\end{algorithmic}
\end{multicols}
\end{minipage}}
\caption{\texttt{Check}$(\mathpzc{K},k,\psi,\tilde{\rho})$}\label{Chk2}
\end{algorithm}

Algorithm~\ref{unr} can be used to define the model checking procedure \texttt{ModCheck}$(\mathpzc{K},\psi)$ (Algorithm~\ref{ModCheck2}).
\texttt{ModCheck}$(\mathpzc{K},\psi)$ exploits the procedure $\texttt{Check}(\mathpzc{K},k, \psi,\tilde{\rho})$ (Algorithm~\ref{Chk2}), which checks a formula $\psi$ of B-nesting depth $k$ against a track $\tilde{\rho}$ of the Kripke structure $\mathpzc{K}$. $\texttt{Check}(\mathpzc{K},k, \psi,\tilde{\rho})$ basically calls itself recursively on the subformulas of $\psi$, and uses the unravelling Algorithm~\ref{unr} to deal with $\hsA$, $\hsAt$, $\hsBt$, and $\hsEt$ modalities. Soundness and completeness of these two procedures are stated by Lemma \ref{lemmamdc}
and Theorem \ref{thModcheck} below, whose proofs can be found in \ref{explCheck} and \ref{proofModCheck2}, respectively. 

\begin{lemma}\label{lemmamdc}
Let $\psi$ be an $\AAbarBBbarEbar$ formula with $\nestb(\psi)=k$, $\mathpzc{K}$ be a finite Kripke structure, and $\tilde{\rho}$ be a track in $\Trk_\mathpzc{K}$. It holds that $\texttt{Check}(\mathpzc{K},k, \psi,\tilde{\rho})=1$ if and only if $\mathpzc{K},\tilde{\rho}\models \psi$.
\end{lemma}

\begin{theorem}\label{thModcheck}
Let $\psi$ be an $\AAbarBBbarEbar$ formula and $\mathpzc{K}$ be a finite Kripke structure. It holds that \texttt{ModCheck}$(\mathpzc{K},\psi)=1$ if and only if $\mathpzc{K}\models \psi$.
\end{theorem}

The model checking algorithm \texttt{ModCheck} requires \emph{exponential working space}, as it uses an instance of the unravelling algorithm and some additional space for a track $\tilde{\rho}$. Analogously, every recursive call to \texttt{Check} (possibly) needs an instance of the unravelling algorithm and space for a track. There are at most $|\psi|$ jointly active calls to \texttt{Check} (plus one to \texttt{ModCheck}), thus the maximum space needed by the considered algorithms is $\left(|\psi|+1\right)\cdot O(|W|+\nestb(\psi))\cdot \tau(|W|,\nestb(\psi))$ bits overall, where $\tau(|W|,\nestb(\psi))$ is the maximum length of track representatives, and $O(|W|+\nestb(\psi))$ bits are needed to represent a state of $\mathpzc{K}$, a descriptor element, and a counter for $k$-indistinguishability. 

In conclusion,
we have proved that the model checking problem for formulas of the HS fragment $\AAbarBBbarEbar$ over finite Kripke structures is in EXPSPACE.
As a particular case, formulas $\psi$ of the fragment $\AAbarBbarEbar$ can be checked in \emph{polynomial working space} by \texttt{ModCheck}, as its formulas do not feature $\hsB$ modality (hence $\nestb(\psi)=0$). Thus, the model checking problem for $\AAbarBbarEbar$ is in PSPACE. In the next section, we prove that 
it is actually PSPACE-complete. As a direct consequence, $\AAbarBBbarEbar$ turns out to be PSPACE-hard. 

The next theorem proves that the model checking problem for $\AAbarBBbarEbar$ is NEXP-hard if a \emph{succinct} encoding of formulas is adopted (the proof is given in \ref{sec:succAAbarBBbarEbarHard}).
\begin{theorem*}{threduction}
The model checking problem for succinctly encoded formulas of $\AAbarBBbarEbar$ over finite Kripke structures is NEXP-hard (under polynomial-time reductions).
\end{theorem*}

%% file: section04.tex
\section{The fragment $\AAbarBbarEbar$}\label{subsec:AAbarBbarEbar}
In this section, we prove that the model checking algorithm described in the previous section, applied to $\AAbarBbarEbar$ formulas, is optimal by showing that model checking for $\ABbar$ is a PSPACE-hard problem (Theorem~\ref{th:ABbarHard}). PSPACE-completeness of $\AAbarBbarEbar$ (and $\ABbar$) immediately follows. As a by-product, model checking for $\AAbarBBbarEbar$ is PSPACE-hard as well.

Before proving Theorem~\ref{th:ABbarHard}, we give an example showing that 
the three HS fragments $\AAbarBbarEbar$, $\HSforall$, and 
$\AAbar$, on which we focus in this (and the next) section,
are expressive enough to capture meaningful properties of state-transition systems.

\begin{example}\label{exlong}
\input{example}
\end{example}

Now, in order to prove Theorem~\ref{th:ABbarHard}, we provide a reduction from the QBF problem (i.e., the problem of determining the truth of a \emph{fully-quantified} Boolean formula in \emph{prenex normal form})---which is known to be PSPACE-complete (see, for example, \cite{Sip12})---to the model checking problem for $\ABbar$ formulas over finite Kripke structures. 

We consider a quantified Boolean formula $\psi\! = \! Q_n x_n Q_{n-1} x_{n-1} \cdots Q_1 x_1 \phi(x_n,x_{n-1},\!\cdots\! ,x_1)$ where $Q_i\in\{\exists, \forall\}$ for all $i=1,\cdots ,n$, and $\phi(x_n,x_{n-1},\cdots ,x_1)$ is a quantifier-free Boolean formula.
Let $Var = \{x_n,\ldots ,x_1\}$ be the set of variables of $\psi$. We define the Kripke structure $\mathpzc{K}_{QBF}^{Var}$, whose initial tracks represent all the possible assignments to the variables of $Var$. For each $x \in Var$, $\mathpzc{K}_{QBF}^{Var}$ features four states,  
$w_x^{\top 1}$, $w_x^{\top 2}$, $w_x^{\bot 1}$, and $w_x^{\bot 2}$: the first two
represent a $\top$ truth assignment to $x$ and the last two a $\bot$ one.
$\mathpzc{K}_{QBF}^{Var}=(\mathpzc{AP},W, \delta,\mu,w_0)$ is formally defined as follows:
\begin{itemize}
    \item $\mathpzc{AP}= Var \cup \{start\} \cup \{ x_{i\,aux} \mid 1\leq i\leq n\}$;
    \item $W= \{w_{x_i}^\ell \mid 1\leq i\leq n, \ell \in \{\bot_1,\bot_2,\top_1,\top_2\}\} \cup \{w_0,w_1,sink\}$;
    \item if $n=0$, $\delta=\{(w_0,w_1),(w_1,sink),(sink,sink)\}$; \\
            if $n>0$,             
            $\delta = \{(w_0,w_1),(w_1,w_{x_n}^{\top_1}),(w_1,w_{x_n}^{\bot_1})\}\cup
            \{(w_{x_i}^{\top_1},w_{x_i}^{\top_2}), (w_{x_i}^{\bot_1},w_{x_i}^{\bot_2}) \mid 1 \leq i \leq n\} \cup 
            \{(w_{x_i}^\ell,w_{x_{i-1}}^m) \mid \ell \in \{\bot_2,\top_2\},  m \in \{\bot_1,\top_1\}, 2 \leq i \leq n\} \cup 
            \{(w_{x_1}^{\top_2},sink),(w_{x_1}^{\bot_2},sink)\}\cup
            \{(sink,sink) \}$. 
\item $\mu(w_0) = \mu(w_1) = Var \cup \{start\}$;\\
            $\mu(w_{x_i}^\ell) = Var \cup \{x_{i\, aux}\}$, for $1\leq i\leq n$ and $\ell \in \{\top_1,\top_2\}$;\\
            $\mu(w_{x_i}^\ell) = (Var \setminus \{x_i\}) \cup \{x_{i\, aux}\}$, for $1\leq i\leq n$ and $\ell \in \{\bot_1,\bot_2\}$;\\
            $\mu(sink) = Var$.
\end{itemize}

\begin{figure}
\centering
\resizebox{\textwidth}{!}{
\begin{tikzpicture}[->,>=stealth',shorten >=1pt,auto,node distance=2.7cm,semithick,every node/.style={circle,draw,inner sep=2pt},every loop/.style={max distance=8mm}]  
    \node[style={double}] (a) {$\stackrel{w_0}{x,y,z,start}$};
    \node (0) [right of=a] {$\stackrel{w_1}{x,y,z,start}$};
    \node (1a) [above right = -0.3cm and 0.6cm of 0] {$\stackrel{w_x^{\top 1}}{x,y,z,x_{aux}}$};
    \node (1b) [below right = -0.3cm and 0.7cm of 0] {$\stackrel{w_x^{\bot 1}}{y,z,x_{aux}}$};
    \node (1a1) [right of=1a] {$\stackrel{w_x^{\top 2}}{x,y,z,x_{aux}}$};
    \node (1b1) [right of=1b] {$\stackrel{w_x^{\bot 2}}{y,z,x_{aux}}$};
    
    \node (2a) [right of=1a1] {$\stackrel{w_y^{\top 1}}{x,y,z,y_{aux}}$};
    \node (2b) [right of=1b1] {$\stackrel{w_y^{\bot 1}}{x,z,y_{aux}}$};
    \node (2a1) [right of=2a] {$\stackrel{w_y^{\top 2}}{x,y,z,y_{aux}}$};
    \node (2b1) [right of=2b] {$\stackrel{w_y^{\bot 2}}{x,z,y_{aux}}$};
    
    \node (3a) [right of=2a1] {$\stackrel{w_z^{\top 1}}{x,y,z,z_{aux}}$};
    \node (3b) [right of=2b1] {$\stackrel{w_z^{\bot 1}}{x,y,z_{aux}}$};
    \node (3a1) [right of=3a] {$\stackrel{w_z^{\top 2}}{x,y,z,z_{aux}}$};
    \node (3b1) [right of=3b] {$\stackrel{w_z^{\bot 2}}{x,y,z_{aux}}$};
    
    \node (pit) [below right of=3a1] {$\stackrel{sink}{x,y,z}$};
  
  \path
    (a) edge (0)
    (0) edge (1a)
        edge (1b)
        
    (1a) edge (1a1)
    (1b) edge (1b1)
    (1a1) edge (2a)
        edge (2b)
    (1b1) edge (2a)
        edge (2b)
        
    (2a) edge (2a1)
    (2b) edge (2b1)
    (2a1) edge (3a)
        edge (3b)
    (2b1) edge (3a)
        edge (3b)
        
     (3a) edge (3a1)
    (3b) edge (3b1)
    (3a1) edge (pit)
    (3b1) edge (pit)
    (pit) edge [loop above] (pit)
    ;
\end{tikzpicture}}
\caption{Kripke structure $\mathpzc{K}_{QBF}^{x,y,z}$ associated with a quantified Boolean formula with variables $x$, $y$, $z$.}\label{Kqbf}
\end{figure}
An example of such a Kripke structure, for $Var=\{x,y,z\}$, is given in Figure \ref{Kqbf}.

From $\psi$, we obtain the $\ABbar$ formula $\xi=start\rightarrow \xi_n$, where 
\begin{equation*}
\xi_i=
\begin{cases}
\phi(x_n,x_{n-1},\cdots ,x_1) & i=0\\
\hsBt\big((\hsA x_{i\, aux}) \wedge \xi_{i-1}\big) & i>0 \wedge Q_i=\exists\\
[\overline{B}]\big((\hsA x_{i\, aux}) \rightarrow \xi_{i-1}\big) & i>0 \wedge Q_i=\forall
\end{cases}
\end{equation*}
Both $\mathpzc{K}_{QBF}^{Var}$ and $\xi$ can be built by using logarithmic working space. 
We will show (proof of Theorem~\ref{th:ABbarHard}) that $\psi$
is true if and only if $\mathpzc{K}_{QBF}^{Var}\models\xi$.
%
%
%
%
%
%
%
%
\input{longproof}

%% file: example.tex
\begin{figure}[t]
\centering
\begin{tikzpicture}[->,>=stealth',shorten >=1pt,auto,node distance=2.8cm,semithick,every node/.style={circle,draw,inner sep=1pt,minimum width=1.5cm},every loop/.style={max distance=8mm}]  
    \useasboundingbox (-4.2,0.7) rectangle (11.3,-6.5);

    \node[style={double}] (0) {$\stackrel{w_0}{x_0}$};
    \node (1) [below right of=0] {$\stackrel{w_1}{r_0,x_0}$};
    \node (2) [below left of=0] {$\stackrel{w_2}{r_1}$};
    \node (3) [below right of=2] {$\stackrel{w_3}{r_0,r_1}$};
    
    \node (4) [right of=1] {$\stackrel{w_4}{r_0,r_1,e_1}$};
    \node (5) [below of=4] {$\stackrel{w_5}{e_1}$};
    
    \node (6) [right of=4] {$\stackrel{w_6}{r_0,r_1,e_0,x_0}$};
    \node (7) [below of=6] {$\stackrel{w_7}{e_0,x_0}$};
    
    \node (8) [right of=6] {$\stackrel{w_8}{r_0,r_1,e_0,e_1}$};
    \node (9) [below of=8] {$\stackrel{w_9}{e_0,e_1}$};

  \path
    (0) edge [loop right] (0)
        edge (1)
        edge (2)
    (1) edge [loop left] (1)
        edge (3)
        edge [out=30,in=150] (6)
    (2) edge [loop right] (2)
        edge (3)
        edge [out=330,in=210] (4)
    (3) edge [out=25,in=220] (4)
        edge [out=20,in=210] (6)
        edge [out=15,in=210] (8)
    (4) edge (5)
    (6) edge (7)
    (8) edge (9);
    
    \draw (5) to[out=198, in=195, looseness=2.12] (0);
    \draw (7) to[out=205, in=190, looseness=2.1] (0);
    \draw (9) to[out=210, in=185, looseness=2.1] (0);
    
    \draw[color=gray,dashed] (-3,0.9) rectangle (3,-4.8); 
    \draw[color=gray,dashed] (3.5,-0.8) rectangle (11.8,-5.7);
    
    \draw[dashed,color=gray,-] (6.1,-0.9) -- (6.1,-5.7);
    \draw[dashed,color=gray,-] (9,-0.9) -- (9,-5.7);
    
    \node [draw=none,above = 0.5cm of 4,minimum width=0] (pp) {$\mathcal{P}_1$};
    \node [draw=none,minimum width=0] at (pp -| 6) {$\mathcal{P}_0$};
    \node [draw=none,minimum width=0] at (pp -| 8) {$\mathcal{P}_0,\; \mathcal{P}_1$};
    \node [draw=none,minimum width=0] at (-3.3,0.4) {$\mathcal{S}$};
   
\end{tikzpicture}
\caption{A simple state-transition system.}\label{KEx}
\end{figure}

Let $\mathpzc{K}=(\mathpzc{AP},W, \delta,\mu,w_0)$, with $\mathpzc{AP}=\{r_0, r_1,e_0,e_1,x_0\}$, be 
the Kripke structure of Figure \ref{KEx}, that models the interactions between a scheduler $\mathpzc{S}$ and two processes, $\mathcal{P}_0$ and $\mathcal{P}_1$, which possibly ask for a shared resource. 
At the initial state $w_0$, $\mathpzc{S}$ has not received any request from the processes yet, while
in $w_1$ (resp., $w_2$) only $\mathcal{P}_0$ (resp., $\mathcal{P}_1$) has sent a request, and thus $r_0$ (resp., $r_1$) holds. As long as at most one process has issued a request, $\mathpzc{S}$ is not forced to allocate the resource ($w_1$ and $w_2$ have self loops). At state $w_3$, both $\mathcal{P}_0$ and $\mathcal{P}_1$ are waiting for the shared resource (both $r_0$ and $r_1$ hold).  State $w_3$ has transitions only towards $w_4$, $w_6$, and $w_8$. At state $w_4$ (resp., $w_6$) $\mathcal{P}_1$ (resp., $\mathcal{P}_0$) can access the resource and $e_1$ (resp., $e_0$) holds in the interval $w_4w_5$ (resp., $w_6w_7$). In addition, a faulty transition may be taken from $w_3$ leading to states $w_8$ and $w_9$ where both $\mathcal{P}_0$ and $\mathcal{P}_1$ use the resource (both $e_0$ and $e_1$ hold in the interval $w_8w_9$). Finally, from 
$w_5$, $w_7$, and $w_9$ the system can only move to $w_0$, where $\mathpzc{S}$ waits for new requests from $\mathcal{P}_0$ and $\mathcal{P}_1$.

Let $\mathpzc{P}$ be the set $\{r_0,r_1,e_0,e_1\}$ and let $x_0$ be an auxiliary proposition letter labelling the states $w_0$, $w_1$, $w_6$, and $w_7$, where $\mathpzc{S}$ and $\mathcal{P}_0$, but not $\mathcal{P}_1$, are active. 

We now give some examples of formulas in the fragments $\AAbarBbarEbar$, $\HSforall$, and $\AAbar$ that encode requirements for $\mathpzc{K}$.
As in Example~\ref{ExampleSched}, we force the validity of the considered property over all legal computation sub-intervals by using the modality $[E]$, or alternatively the modality $[A]$ (any computation sub-interval occurs after at least one initial track).

It can be checked that $\mathpzc{K}\not\models [E]\neg(e_0\wedge e_1)$ (the formula is in $\HSforall$), i.e.,
mutual exclusion is not guaranteed, as the faulty transition leading to $w_8$ may be taken at $w_3$, and  then both $\mathcal{P}_0$ and $\mathcal{P}_1$ access the resource in the interval $w_8w_9$ where $e_0\wedge e_1$ holds.

On the contrary, it holds that $\mathpzc{K}\models [A]\big(r_0\rightarrow \hsA e_0 \vee \hsA\hsA e_0\big)$ (the formula is in $\AAbar$ and $\AAbarBbarEbar$). Such a formula expresses the following reachability property: if $r_0$ holds over some interval, then it is always possible to reach an interval where $e_0$ holds. Obviously, this does not mean that all possible computations will necessarily lead to such an interval, but that the system is \emph{never} trapped in a state from which it is no more possible to satisfy requests from $\mathcal{P}_0$.
    
It also holds that $\mathpzc{K}\models [A]\big(r_0\wedge r_1\rightarrow [A](e_0\vee e_1\vee \bigwedge_{p\in\mathpzc{P}}\neg p)\big)$ (in $\AAbar$ and $\AAbarBbarEbar$). Indeed, if both processes send a request (state $w_3$), then $\mathpzc{S}$ immediately allocates the resource. In detail, if $r_0\wedge r_1$ holds over some tracks (the only possible intervals are $w_3w_4$, $w_3w_6$, and $w_3w_8$), then in any possible subsequent interval of length 2 $e_0\vee e_1$ holds, that is, $\mathcal{P}_0$ or $\mathcal{P}_1$ are executed, or, considering tracks longer than 2, $\bigwedge_{p\in\mathpzc{P}}\neg p$ holds.
On the contrary, if only one process asks for the resource, then $\mathpzc{S}$ can arbitrarily delay the allocation, and therefore $\mathpzc{K}\not\models [A]\big(r_0\rightarrow [A](e_0\vee \bigwedge_{p\in\mathpzc{P}}\neg p)\big)$.

    Finally, it holds that $\mathpzc{K}\models x_0\rightarrow\hsBt x_0$ (in $\AAbarBbarEbar$), that is, any initial track satisfying $x_0$ (any such track involves states $w_0$, $w_1$, $w_6$, and $w_7$ only) can be extended to the right in such a way that the resulting track still satisfies $x_0$. This amounts to say that there exists a computation in which $\mathcal{P}_1$ starves. Note that $\mathpzc{S}$ and $\mathcal{P}_0$ can continuously interact without waiting for $\mathcal{P}_1$. This is the case, for instance, when $\mathcal{P}_1$ is not asking for the shared resource at all. 

%% file: longproof.tex

As a preliminary step, we introduce some technical definitions.
Given a Kripke structure $\mathpzc{K}=(\mathpzc{AP},W, \delta,\mu,w_0)$ and an $\ABbar$ formula $\psi$, we denote by $p\ell(\psi)$ the set of proposition letters occurring in $\psi$ and by $\mathpzc{K}_{\,|p\ell(\psi)}$ the 
structure obtained from $\mathpzc{K}$ by restricting the labelling of each state to $p\ell(\psi)$, namely, the 
Kripke 
structure $(\overline{\mathpzc{AP}},W, \delta,\overline{\mu},w_0)$, where $\overline{\mathpzc{AP}}=\mathpzc{AP}\cap p\ell(\psi)$ and $\overline{\mu}(w)=\mu(w)\cap p\ell(\psi)$, for all $w\in W$.
Moreover, for $v\in W$, we denote by $reach(\mathpzc{K},v)$ the subgraph of $\mathpzc{K}$, with $v$ as its 
initial state, consisting of all and only the states which are reachable from $v$, namely, the 
Kripke 
structure $(\mathpzc{AP},W',\delta',\mu',v)$, where $W'=\{w\in W \mid \text{ there exists } \rho\in \Trk_\mathpzc{K} \text{ with } \fst(\rho)=v \text{ and } \lst(\rho)=w\}$, $\delta'=\delta \cap (W'\times W')$, and $\mu'(w)=\mu(w)$, for all $w\in W'$. 

As usual, two Kripke structures $\mathpzc{K}=(\mathpzc{AP},W, \delta,\mu,w_0)$ and $\mathpzc{K}'=(\mathpzc{AP}',W', \delta',\mu',w_0')$ are said to be \emph{isomorphic} ($\mathpzc{K}\sim \mathpzc{K}'$ for short) if and only if there is a \emph{bijection} $f:W\mapsto W'$ such that $(i)$~$f(w_0)=w_0'$; $(ii)$~for all $u,v\in W$, $(u,v)\in \delta$ if and only if $(f(u),f(v))\in\delta'$; $(iii)$~for all $v\in W$, $\mu(v)=\mu'(f(v))$.

Finally, if  $\mathpzc{A}_\mathpzc{K}=(\mathpzc{AP},\mathbb{I},A_\mathbb{I},B_\mathbb{I},E_\mathbb{I},\sigma)$ is the abstract interval model induced by a Kripke structure  $\mathpzc{K}$ and $\rho\in\Trk_{\mathpzc{K}}$,
 we denote $\sigma(\rho)$ by $\mathpzc{L}(\mathpzc{K},\rho)$.

Let $\mathpzc{K}$ and $\mathpzc{K}'$ be two Kripke structures. The following lemma states that, for any $\ABbar$ formula $\psi$, if the same set of proposition letters, restricted to $p\ell(\psi)$, holds over two tracks $\rho \in \Trk_{\mathpzc{K}}$ and $\rho' \in \Trk_{\mathpzc{K}'}$, and the subgraphs consisting of the states reachable from, respectively, $\lst(\rho)$ and $\lst(\rho')$ are isomorphic, then $\rho$ and $\rho'$ are equivalent with respect to $\psi$.

\begin{lemma}\label{lemmaABbar}
Given an $\ABbar$ formula $\psi$, two Kripke structures $\mathpzc{K}=(\mathpzc{AP},W, \delta,\mu,w_0)$ and $\mathpzc{K}'=(\mathpzc{AP}',W', \delta',\mu',w_0')$, and two tracks $\rho\in\Trk_\mathpzc{K}$ and $\rho'\in \Trk_{\mathpzc{K}'}$ such that \[\mathpzc{L}(\mathpzc{K}_{\,|p\ell(\psi)},\rho)=\mathpzc{L}(\mathpzc{K}'_{\,|p\ell(\psi)},\rho')\quad \text{and}\quad reach(\mathpzc{K}_{\,|p\ell(\psi)},\lst(\rho))\sim reach(\mathpzc{K}'_{\,|p\ell(\psi)},\lst(\rho')),\] 
it holds that 
$\mathpzc{K},\rho\models\psi \iff \mathpzc{K}',\rho'\models\psi$.
\end{lemma}
The proof of this lemma can be found in \ref{sec:lemmaABbarProof}.

\begin{theorem}\label{th:ABbarHard}
The model checking problem for $\ABbar$ formulas over finite Kripke structures is PSPACE-hard (under LOGSPACE reductions).
\end{theorem}
\begin{proof}
We prove that the quantified Boolean formula $\psi=Q_n x_n Q_{n-1} x_{n-1} \cdots Q_1 x_1 \phi(x_n,x_{n-1},\allowbreak \cdots ,x_1)$ is true if and only if $\mathpzc{K}_{QBF}^{x_n,\cdots , x_1}\models\xi$ by induction on the number of variables $n\geq 0$ of $\psi$. 
In the following, $\phi(x_n,x_{n-1},\cdots ,x_1)\{x_i/\upsilon\}$, with $\upsilon\in\{\top ,\bot\}$, denotes the formula obtained from $\phi(x_n,x_{n-1},\cdots ,x_1)$ by replacing all 
occurrences of $x_i$ by $\upsilon$. 
It is worth noticing that $\mathpzc{K}_{QBF}^{x_n,x_{n-1},\cdots , x_1}$ and $\mathpzc{K}_{QBF}^{x_{n-1},\cdots , x_1}$ are isomorphic when they are restricted to the states 
$w_{x_{n-1}}^{\top 1}$, $w_{x_{n-1}}^{\top 2}$, $w_{x_{n-1}}^{\bot 1}$, $w_{x_{n-1}}^{\bot 2}$, $\cdots$, $w_{x_1}^{\top 1}$, $w_{x_1}^{\top 2}$, $w_{x_1}^{\bot 1}$, $w_{x_1}^{\bot 2}$, $sink$ (i.e., the leftmost part of both Kripke structures is omitted), and the labelling of states is suitably restricted accordingly. Note that only the track $w_0w_1$ satisfies $start$ and, for $i=n,\cdots , 1$, the proposition letter $x_{i\, aux}$ is satisfied by the two tracks $w_{x_i}^{\top 1}w_{x_i}^{\top 2}$ and $w_{x_i}^{\bot 1}w_{x_i}^{\bot 2}$ only.

\smallskip

(Case $n=0$) $\psi$ equals $\phi$ and it has no variables. The states of $\mathpzc{K}_{QBF}^\emptyset$ are $W=\{w_0,w_1,sink\}$ and $\xi=start\rightarrow \phi$. 

    Let us assume $\phi$ to be true. All initial tracks of length greater than 2 trivially satisfy $\xi$, as $start$ does not hold on them. As for $w_0w_1$, it is true that $\mathpzc{K}_{QBF}^\emptyset,w_0w_1\models \phi$, since $\phi$ is true (its truth does not depend on the proposition letters that hold on $w_0w_1$, because it has no variables). Thus $\mathpzc{K}_{QBF}^\emptyset\models\xi$. Vice versa, if $\mathpzc{K}_{QBF}^\emptyset\models\xi$, then in particular $\mathpzc{K}_{QBF}^\emptyset,w_0w_1\models \phi$. But $\phi$ has no variables, hence it is true.
    
\smallskip
    
(Case $n\geq 1$) Let us consider the formula $\psi=Q_n x_n Q_{n-1} x_{n-1} \cdots Q_1 x_1 \phi(x_n,x_{n-1},\cdots ,x_1)$. We distinguish two cases, depending on whether $Q_n=\exists$ or $Q_n=\forall$, and for both we prove the two implications. 

    \medskip    
        \noindent $\circ$ Case $Q_n=\exists$:
        \medskip
        
        $(\Rightarrow)$ If the formula $\psi$ is true, then, by definition, there exists 
$\upsilon\in\{\top , \bot\}$ such that if we replace all 
occurrences of $x_n$ in $\phi(x_n,x_{n-1},\cdots ,x_1)$ by $\upsilon$, we get the formula $\phi'(x_{n-1},\cdots ,x_1)=\phi(x_n,x_{n-1},\cdots ,x_1)\{x_n/\upsilon\}$ such that $\psi'=Q_{n-1} x_{n-1} \cdots Q_1 x_1 \phi'(x_{n-1},\cdots ,x_1)$ is a true quantified Boolean formula. By the inductive hypothesis $\mathpzc{K}_{QBF}^{x_{n-1},\cdots , x_1}\models\xi'$, where $\xi'=start\rightarrow \xi_{n-1}'$ is obtained from $\psi'$ and $\xi_{n-1}'=\xi_{n-1}\{x_n/\upsilon\}$. It follows that $\mathpzc{K}_{QBF}^{x_{n-1},\cdots , x_1},w_0'w_1'\models\xi'_{n-1}$, where $w_0'$ and $w_1'$ are the two ``leftmost'' states of $\mathpzc{K}_{QBF}^{x_{n-1},\cdots , x_1}$ (corresponding to $w_0$ and $w_1$ of $\mathpzc{K}_{QBF}^{x_n,\cdots , x_1}$).
        
        We now prove that $\mathpzc{K}_{QBF}^{x_n,\cdots , x_1}\models\xi$. Let us consider a generic initial track $\rho$ in $\mathpzc{K}_{QBF}^{x_n,\cdots , x_1}$. If it does not satisfy $start$, then it trivially holds that $\mathpzc{K}_{QBF}^{x_n,\cdots , x_1},\rho\models\xi$. Otherwise $\rho=w_0w_1$, and we have to show that $\mathpzc{K}_{QBF}^{x_n,\cdots , x_1},w_0w_1\models \hsBt((\hsA x_{n\, aux}) \wedge \xi_{n-1})$ ($=\xi_n$). If $\upsilon=\top$, we consider $w_0w_1w_{x_n}^{\top 1}$; otherwise, we consider $w_0w_1w_{x_n}^{\bot 1}$. In the first case (the other is symmetric), we must prove that $\mathpzc{K}_{QBF}^{x_n,\cdots , x_1},w_0w_1w_{x_n}^{\top 1}\models (\hsA x_{n\, aux}) \wedge \xi_{n-1}$. It trivially holds that $\mathpzc{K}_{QBF}^{x_n,\cdots , x_1},w_0w_1w_{x_n}^{\top 1}\models \hsA x_{n\, aux}$. Hence, we only need to prove that $\mathpzc{K}_{QBF}^{x_n,\cdots , x_1},w_0w_1w_{x_n}^{\top 1}\models \xi_{n-1}$.
        
As we have shown, by the inductive hypothesis, it holds that $\mathpzc{K}_{QBF}^{x_{n-1},\cdots , x_1},w_0'w_1'\models\xi'_{n-1}$($=\xi_{n-1}\{x_n/\top\}$). 
        Now, since
        \begin{itemize}        
            \item $p\ell(\xi_{n-1}\{x_n/\top\})=\{x_1,\cdots ,x_{n-1},x_{1\, aux},\cdots ,x_{n-1\, aux}\}$,
            \item $\mathpzc{L}({\mathpzc{K}_{QBF}^{x_{n-1},\cdots , x_1}}_{|p\ell(\xi_{n-1}\{x_n/\top\})},w_0'w_1')= \{x_{n-1},\cdots ,x_1\}$,
            \item $\mathpzc{L}({\mathpzc{K}_{QBF}^{x_n,\cdots , x_1}}_{|p\ell(\xi_{n-1}\{x_n/\top\})},w_0w_1w_{x_n}^{\top 1}w_{x_n}^{\top 2})=\{x_{n-1}, \cdots ,x_1\}$, and
            \item $reach({\mathpzc{K}_{QBF}^{x_n,\cdots , x_1}}_{|p\ell(\xi_{n-1}\{x_n/\top\})},w_{x_n}^{\top 2})\sim reach({\mathpzc{K}_{QBF}^{x_{n-1},\cdots , x_1}}_{|p\ell(\xi_{n-1}\{x_n/\top\})},w_1')$,
        \end{itemize}
        by Lemma \ref{lemmaABbar}, $\mathpzc{K}_{QBF}^{x_n,\cdots , x_1},w_0w_1w_{x_n}^{\top 1}w_{x_n}^{\top 2}\models\xi'_{n-1}$. Hence, $\mathpzc{K}_{QBF}^{x_n,\cdots , x_1},w_0w_1w_{x_n}^{\top 1}w_{x_n}^{\top 2}\models\xi_{n-1}$ as $x_n$ is in the labelling of the track $w_0w_1w_{x_n}^{\top 1}w_{x_n}^{\top 2}$ and of any $\overline{\rho}$ such that $w_0w_1w_{x_n}^{\top 1}w_{x_n}^{\top 2}\in\Pref(\overline{\rho})$.
        
        Now, if $n=1$, then $\xi_{n-1}=\phi(x_n)$ and it holds that $\mathpzc{K}_{QBF}^{x_n,\cdots, x_1},w_0w_1w_{x_n}^{\top 1}\models\xi_{n-1}$. 
        If $n>1$, either $\xi_{n-1}=\hsBt ((\hsA x_{n-1\, aux}) \wedge \xi_{n-2})$ or $\xi_{n-1}=[\overline{B}] ((\hsA x_{n-1\, aux}) \rightarrow \xi_{n-2})$.
%
            In the first case, since $\mathpzc{K}_{QBF}^{x_n,\cdots , x_1},w_0w_1w_{x_n}^{\top 1}w_{x_n}^{\top 2}\models\hsBt ((\hsA x_{n-1\, aux})\wedge \xi_{n-2})$, there are only two possibilities: 
$\mathpzc{K}_{QBF}^{x_n,\cdots , x_1},w_0w_1w_{x_n}^{\top 1}w_{x_n}^{\top 2}w_{x_{n-1}}^{\top 1}\models (\hsA x_{n-1\, aux}) \wedge \xi_{n-2}$
or 
$\mathpzc{K}_{QBF}^{x_n,\cdots , x_1},w_0w_1w_{x_n}^{\top 1}w_{x_n}^{\top 2}w_{x_{n-1}}^{\bot 1}\models(\hsA x_{n-1\, aux}) \wedge \xi_{n-2}$.
            %
            In both cases, $\mathpzc{K}_{QBF}^{x_n,\cdots , x_1},w_0w_1w_{x_n}^{\top 1}\models\hsBt((\hsA x_{n-1\, aux}) \wedge \xi_{n-2})$. 
            
        Otherwise, it holds that $\mathpzc{K}_{QBF}^{x_n,\cdots , x_1},w_0w_1w_{x_n}^{\top 1}w_{x_n}^{\top 2}\models[\overline{B}] ((\hsA x_{n-1\, aux}) \rightarrow \xi_{n-2})$. It follows that $\mathpzc{K}_{QBF}^{x_n,\cdots , x_1},w_0w_1w_{x_n}^{\top 1}w_{x_n}^{\top 2}w_{x_{n-1}}^{\top 1}\models \xi_{n-2}$ and $\mathpzc{K}_{QBF}^{x_n,\cdots , x_1},w_0w_1w_{x_n}^{\top 1}w_{x_n}^{\top 2}w_{x_{n-1}}^{\bot 1}\models \xi_{n-2}$. As a consequence, $\mathpzc{K}_{QBF}^{x_n,\cdots , x_1},w_0w_1w_{x_n}^{\top 1}\models[\overline{B}] ((\neg \hsA x_{n-1\, aux}) \vee \xi_{n-2})$ $(=\xi_{n-1})$ (recall that the only successor of $w_{x_n}^{\top 1}$ in $\mathpzc{K}_{QBF}^{x_n,\cdots , x_1}$ is $w_{x_n}^{\top 2}$ and, in particular, $\mathpzc{K}_{QBF}^{x_n,\cdots , x_1},w_0w_1w_{x_n}^{\top 1}w_{x_n}^{\top 2}\models \neg \hsA x_{n-1\, aux}$).
                        
        \medskip
        $(\Leftarrow)$ If $\mathpzc{K}_{QBF}^{x_n,\cdots , x_1}\models\xi$, it holds that $\mathpzc{K}_{QBF}^{x_n,\cdots , x_1},w_0w_1\models \hsBt(\hsA x_{n\, aux} \wedge \xi_{n-1})$. Hence, either $\mathpzc{K}_{QBF}^{x_n,\cdots , x_1},w_0w_1w_{x_n}^{\top 1}\models (\hsA x_{n\, aux}) \wedge \xi_{n-1}$ or $\mathpzc{K}_{QBF}^{x_n,\cdots , x_1},w_0w_1w_{x_n}^{\bot 1}\models (\hsA x_{n\, aux}) \wedge \xi_{n-1}$. Let us consider the first case (the other is symmetric). It holds that $\mathpzc{K}_{QBF}^{x_n,\cdots , x_1},w_0w_1w_{x_n}^{\top 1}\models \xi_{n-1}\{x_n/\top\}$ and $\mathpzc{K}_{QBF}^{x_n,\cdots , x_1},w_0w_1w_{x_n}^{\top 1}w_{x_n}^{\top 2}\models \xi_{n-1}\{x_n/\top\}$ (as before). By Lemma \ref{lemmaABbar} we get that $\mathpzc{K}_{QBF}^{x_{n-1},\cdots , x_1},w_0'w_1'\models \xi_{n-1}\{x_n/\top\}(=\xi_{n-1}')$ and thus $\mathpzc{K}_{QBF}^{x_{n-1},\cdots , x_1}\models start \rightarrow \xi_{n-1}'$, namely, $\mathpzc{K}_{QBF}^{x_{n-1},\cdots , x_1}\models \xi'$. 
        By the inductive hypothesis, 
        $\psi'=Q_{n-1} x_{n-1} \cdots Q_1 x_1 \phi(x_n,x_{n-1},\allowbreak \cdots ,x_1)\{x_n/\top\}$ is true. Hence, $\psi=\exists x_n Q_{n-1} x_{n-1} \cdots Q_1 x_1 \phi(x_n,x_{n-1},\cdots ,x_1)$ is true.
            
            \medskip
            \noindent $\circ$ Case $Q_n=\forall$: 
            \medskip
            
            $(\Rightarrow)$ Assume that both $\psi'=Q_{n-1} x_{n-1} \cdots Q_1 x_1 \phi(x_n,x_{n-1},\cdots, x_1) \{x_n/\top\}$ and $\psi''=\linebreak Q_{n-1} x_{n-1} \cdots Q_1 x_1 \phi(x_n,x_{n-1}, \cdots, x_1)\{x_n/\bot\}$ are true quantified Boolean formulas. We show that $\mathpzc{K}_{QBF}^{x_n,\cdots , x_1},w_0w_1\models [\overline{B}]((\hsA x_{n\, aux}) \rightarrow \xi_{n-1})$. To this end, we prove that both $\mathpzc{K}_{QBF}^{x_n,\cdots , x_1},w_0w_1w_{x_n}^{\top 1}\models \xi_{n-1}$ and $\mathpzc{K}_{QBF}^{x_n,\cdots , x_1},w_0w_1w_{x_n}^{\bot 1}\models \xi_{n-1}$. This can be shown exactly as in the $\exists$ case.
        
            \medskip
            $(\Leftarrow)$ If $\mathpzc{K}_{QBF}^{x_n,\cdots , x_1}\models\xi$, then $\mathpzc{K}_{QBF}^{x_n,\cdots , x_1},w_0w_1\models [\overline{B}]((\hsA x_{n\, aux}) \rightarrow \xi_{n-1})$. Hence, both \linebreak $\mathpzc{K}_{QBF}^{x_n,\cdots , x_1},w_0w_1w_{x_n}^{\top 1}\models \xi_{n-1}$ and $\mathpzc{K}_{QBF}^{x_n,\cdots , x_1},w_0w_1w_{x_n}^{\bot 1}\models \xi_{n-1}$. Reasoning as in the $\exists$ case and by applying the inductive hypothesis twice, we get that the quantified Boolean formulas $Q_{n-1} x_{n-1} \cdots\allowbreak Q_1 x_1 \phi(x_n,x_{n-1},\cdots ,x_1)\{x_n/\top\}$ and $Q_{n-1} x_{n-1} \cdots Q_1 x_1 \phi(x_n,x_{n-1},\cdots ,\allowbreak x_1)\{x_n/\bot\}$ are true; thus $\forall x_n Q_{n-1} x_{n-1} \cdots Q_1 x_1 \phi(x_n,x_{n-1},\cdots ,x_1)$ is true.
\end{proof}

%% file: section05.tex
\section{The fragment $\HSforall$}\label{subsec:forallAAbarBE}

In this section, we introduce and study the complexity of the model checking problem for the universal fragment of $\AAbarBE$, denoted by
$\HSforall$. Its formulas are defined as follows:
\begin{equation*}
\psi ::= \beta \;\vert\; \psi \wedge \psi \;\vert\; [A]\psi \;\vert\; [B]\psi \;\vert\; [E]\psi \;\vert\; [\overline{A}]\psi ,
\end{equation*}
where $\beta$ is a pure propositional formula,
\[\beta ::= p \;\vert\; \beta\vee\beta \;\vert\; \beta\wedge\beta \;\vert\; \neg\beta \;\vert\; \bot \;\vert\; \top \mbox{ with } p \in \mathpzc{AP}.\]
Formulas of $\HSforall$ can thus be constructed starting from pure propositional formulas (a fragment of HS that we denote by $\HSprop$); subsequently, formulas with universal modalities $[A]$, $[B]$, $[E]$, and $[\overline{A}]$ can be combined only by conjunctions, but not by negations or disjunctions
(which may occur in pure propositional formulas only).

We will prove that the model checking problem for $\HSforall$ formulas (as well as for $\HSprop$) over finite Kripke structures is coNP-complete.

To start with, we need to introduce the (auxiliary) fragment $\HSexi$, which can be regarded as the ``dual'' of $\HSforall$. Its formulas are defined as:
\begin{equation*}
    \psi ::= \beta \;\vert\; \psi \vee \psi \;\vert\; \langle A\rangle\psi \;\vert\; \langle B\rangle\psi \;\vert\; \langle E\rangle\psi \;\vert\; \langle \overline{A}\rangle\psi .
\end{equation*}
$\HSexi$ formulas feature $\hsA$, $\hsB$, $\hsE$, and $\hsAt$ existential modalities; negation and conjunction symbols may occur only in pure propositional formulas.
The intersection of $\HSforall$ and $\HSexi$ is precisely $\HSprop$.
The negation of any $\HSforall$ formula can be transformed into an equivalent $\HSexi$  formula (of at most double length), and vice versa, by using De Morgan's laws and the equivalences $[X]\psi\equiv \neg\langle X\rangle\neg\psi$ and $\neg\neg\psi\equiv\psi$. 

In the following, we outline a \emph{non-deterministic} algorithm to decide the model checking problem for a $\HSforall$ formula $\psi$ (Algorithm~\ref{provCounter}). As usual, the algorithm searches for a counterexample to $\psi$, that is, an initial track satisfying $\neg\psi$. Since, as we already pointed out, $\neg\psi$ is equivalent to a suitable formula $\psi'$ of the dual fragment $\HSexi$, the algorithm looks for an initial track satisfying $\psi'$. 

Algorithm~\ref{provCounter} makes use of \emph{descriptor elements}: we remind that they are the labels of the nodes of $B_k$-descriptors. 
%
%
%
%
%
%
By Proposition~\ref{propzerohalt}, 
if a descriptor element $d$ is witnessed in $\mathpzc{K}$, i.e., there exists some $\rho\in\Trk_\mathpzc{K}$ associated with $d$, then there exists a track of length at most $2+|W|^2$ associated with $d$. Thus, to generate a (all) witnessed descriptor element(s) with initial state $v$, we just need to non-deterministically visit the unravelling of $\mathpzc{K}$ from $v$ up to depth $2+|W|^2$.
This property is fundamental for the completeness of the algorithm, and also for bounding the length of tracks we need to consider.

Before presenting Algorithm~\ref{provCounter}, we need to describe
the non-deterministic auxiliary procedure \texttt{Check$\exists$} (see Algorithm~\ref{ChkExist}), which takes as input a Kripke structure $\mathpzc{K}$, a formula $\psi$ of $\HSexi$, and a witnessed descriptor element $d=(v_{in},S,v_{fin})$ and it returns \textbf{Yes} if and only if there exists a track $\rho\in\Trk_\mathpzc{K}$, associated with $d$, such that $\mathpzc{K},\rho\models\psi$. 
\begin{algorithm}[tp]
\caption{\texttt{Check$\exists$}$(\mathpzc{K},\psi,(v_{in},S,v_{fin}))$}\label{ChkExist}
\begin{algorithmic}[1]
	\If{$\psi=\beta$}\Comment{$\beta$ is a pure propositional formula}
	    \If{$VAL(\beta,(v_{in},S,v_{fin}))=\top$}
	        \State{Yes \textbf{else} No}
	    \EndIf
    \ElsIf{$\psi=\varphi_1\vee\varphi_2$}
        \Either
            \Return{\texttt{Check$\exists$}$(\mathpzc{K},\varphi_1,(v_{in},S,v_{fin}))$} 
        \Or 
            \Return{\texttt{Check$\exists$}$(\mathpzc{K},\varphi_2,(v_{in},S,v_{fin}))$}
        \EndOr
	\ElsIf{$\psi=\hsA\varphi$}
	    \State{$(v_{fin},S',v_{fin}')\gets \texttt{aDescrEl}(\mathpzc{K},v_{fin},\textsc{forw})$}
	    \Return{\texttt{Check$\exists$}$(\mathpzc{K},\varphi,(v_{fin},S',v_{fin}'))$}
	\ElsIf{$\psi=\hsAt\varphi$}
	    \State{$(v_{in}',S',v_{in})\gets \texttt{aDescrEl}(\mathpzc{K},v_{in},\textsc{backw})$}
	    \Return{\texttt{Check$\exists$}$(\mathpzc{K},\varphi,(v_{in}',S',v_{in}))$}
	\ElsIf{$\psi=\hsB\varphi$}
	    \State{$(v_{in}',S',v_{fin}')\gets \texttt{aDescrEl}(\mathpzc{K},v_{in},\textsc{forw})$}\Comment{$v_{in}'=v_{in}$}
	    \Either
	        \If{$(v_{in}',S'\cup \{v_{fin}'\},v_{fin})=(v_{in},S,v_{fin})$ and $(v_{fin}',v_{fin})$ is an edge of $\mathpzc{K}$}
	            \Return{\texttt{Check$\exists$}$(\mathpzc{K},\varphi,(v_{in}',S',v_{fin}'))$}
	        \Else
	            \State{No}
	        \EndIf
	    \Or
	        \State{$(v_{in}'',S'',v_{fin}'')\gets \texttt{aDescrEl}(\mathpzc{K},v_{in}'',\textsc{forw})$, where $(v_{fin}',v_{in}'')$ is an edge of $\mathpzc{K}$ chosen non-deterministically}
	        \If{$\texttt{concat}\left((v_{in}',S',v_{fin}'),(v_{in}'',S'',v_{fin}'')\right)=(v_{in},S,v_{fin})$}
	            \Return{\texttt{Check$\exists$}$(\mathpzc{K},\varphi,(v_{in}',S',v_{fin}'))$}
	        \Else
	            \State{No}
	        \EndIf
	    \EndOr
	\ElsIf{$\psi=\hsE\varphi$}   
        \State{Symmetric to $\psi=\hsB\varphi$}
	\EndIf

\end{algorithmic}
\end{algorithm}
The procedure is recursively defined as follows. 

When it is called on
a pure propositional formula $\beta$ (base of the recursion), $VAL(\beta,d)$ evaluates $\beta$ over $d$ in the standard way. The evaluation can be performed in deterministic polynomial time, and
if $VAL(\beta,d)$ returns $\top$, then there exists a track associated with $d$ (of length at most quadratic in $|W|$) that satisfies $\beta$.

If $\psi =
\psi'\vee\psi''$, where $\psi'$ or $\psi''$ feature some temporal modality, the procedure non-deterministic\-ally calls itself on $\psi'$ or $\psi''$ (the 
construct \textbf{Either} $c_1$ \textbf{Or} $c_2$ \textbf{EndOr} denotes a non-deterministic choice between commands $c_1$ and $c_2$).

If  $\psi =
\hsA \psi'$ (respectively, $\hsAt \psi'$), the procedure looks for a new descriptor element for a track starting from the final state (respectively, leading to the initial state) of the current descriptor element $d$. To this aim, we use the procedure \texttt{aDescrEl}$(\mathpzc{K},v, \textsc{forw})$ (resp., \texttt{aDescrEl}$(\mathpzc{K},v,\textsc{backw})$) which 
non-deterministically returns a descriptor element $(v_{in}',S',v_{fin}')$, with $v_{in}'=v$ (resp., $v_{fin}'=v$), witnessed in $\mathpzc{K}$ by exploring forward (resp., backward) the unravelling of $\mathpzc{K}$ from $v_{in}'$ (resp., from $v_{fin}'$). 
Its complexity is polynomial in $|W|$, since it needs to examine the unravelling of $\mathpzc{K}$ from $v$ up to depth $2+|W|^2$.

If $\psi =
\hsB\psi'$, the procedure looks for a new descriptor element $d_1$ and eventually calls itself on $\psi'$ and $d_1$ only if the current descriptor element $d$ results from the ``concatenation'' of $d_1$ with a suitable descriptor element $d_2$:
if $d_1=(v_{in}',S',v_{fin}')$ and $d_2=(v_{in}'',S'',v_{fin}'')$, then \texttt{concat}$(d_1,d_2)$ returns 
$(v_{in}',S'\cup\{v_{fin}',v_{in}''\}\cup S'',v_{fin}'')$. Notice that if  $\rho_1$ and $\rho_2$ are tracks associated with $d_1$ and $d_2$, respectively, then $\rho_1\cdot\rho_2$ is associated with \texttt{concat}$(d_1,d_2)$.

The following theorem proves soundness and completeness of the \texttt{Check$\exists$} procedure.

\begin{theorem}\label{corrcomplth}
For any formula $\psi$ of the fragment $\HSexi$ and any witnessed descriptor element $d=(v_{in},S,v_{fin})$, the procedure \texttt{Check$\exists$}$(\mathpzc{K},\psi,d)$ has a successful computation if and only if there exists a track $\rho$, associated with $d$, such that $\mathpzc{K},\rho\models\psi$.
\end{theorem}

\input{corrcomplProof}

\begin{algorithm}[tbp]
\caption{\texttt{ProvideCounterex}$(\mathpzc{K},\psi)$}\label{provCounter}
\begin{algorithmic}[1]
    \State{$(v_{in},S,v_{fin})\gets \texttt{aDescrEl}(\mathpzc{K},w_0,\textsc{forw})$}\Comment{$v_{in}=w_0$ is the initial state of $\mathpzc{K}$}
    \Return{\texttt{Check$\exists$}$(\mathpzc{K},\texttt{to}\HSexi(\neg\psi),(v_{in},S,v_{fin}))$}
\end{algorithmic}
\end{algorithm}

We can finally introduce the procedure \texttt{ProvideCounterex}$(\mathpzc{K},\psi)$ (Algorithm \ref{provCounter}), which searches for counterexamples to the input $\HSforall$ formula $\psi$; indeed, it is possible to prove that it has a successful computation if and only if $\mathpzc{K}\not\models\psi$. In the pseudocode of procedure \texttt{ProvideCounterex}, $\texttt{to}\HSexi(\neg\psi)$ denotes the $\HSexi$ formula equivalent to $\neg\psi$.

On the one hand, if \texttt{ProvideCounterex}$(\mathpzc{K},\psi)$ has a successful computation, then there exists a witnessed descriptor element $d=(v_{in},S,v_{fin})$, where $v_{in}$ is $w_0$ (the initial state of $\mathpzc{K}$), such that \texttt{Check$\exists$}$(\mathpzc{K},\texttt{to}\HSexi(\neg\psi),d)$ has a successful computation. This means that there exists a track $\rho$, associated with $d$, such that $\mathpzc{K},\rho\models\neg\psi$, and thus $\mathpzc{K}\not\models\psi$.

On the other hand, if $\mathpzc{K}\not\models\psi$, then there exists an initial track $\rho$ such that $\mathpzc{K},\rho\models\neg\psi$. Let $d$ be the descriptor element for $\rho$: \texttt{Check$\exists$}$(\mathpzc{K},\texttt{to}\HSexi(\neg\psi),d)$ has a successful computation. Since $d$ is witnessed by an initial track, some non-deterministic instance of $\texttt{aDescrEl}(\mathpzc{K},w_0,\textsc{forw})$ returns $d$. Hence \texttt{ProvideCounterex}$(\mathpzc{K},\psi)$ has a successful computation.

As for the complexity, \texttt{ProvideCounterex}$(\mathpzc{K},\psi)$ runs in non-deterministic polynomial time (it is in NP), since the number of recursive invocations of the procedure \texttt{Check$\exists$}
is $O(|\psi|)$, and each invocation requires time polynomial in $|W|$ while generating descriptor elements. Therefore, the model checking problem for $\HSforall$ belongs to coNP.

We conclude the section by proving that the model checking problem for $\HSforall$ is coNP-complete. Such a result is an easy corollary of the following theorem.

\begin{theorem}
Let $\mathpzc{K}$ be a finite Kripke structure and $\beta\in\HSprop$ be a pure propositional formula. The problem of deciding whether $\mathpzc{K}\not\models \beta$ is NP-hard (under LOGSPACE reductions).
\end{theorem}
\begin{proof}
We provide a reduction from the NP-complete SAT problem to the considered problem.
Let $\beta$ be a Boolean formula over a set of variables $Var=\{x_1,\ldots , x_n\}$. 
We build a Kripke structure, $\mathpzc{K}_{SAT}^{Var}=(\mathpzc{AP},W,\delta,\mu,w_0)$, with:
\begin{itemize}
    \item $\mathpzc{AP}=Var$;
    \item $W= \{w_0\}\cup \{w_i^\ell \mid \ell \in \{\top,\bot\},\; 1\leq i \leq n\}$; 
    \item $\delta = \{(w_0,w_1^\top), (w_0,w_1^\bot)\}
    \cup\{(w_i^\ell,w_{i+1}^m) \mid  \ell,m \in \{\top,\bot\}, 1 \leq i \leq n-1\}
    \cup\{(w_n^\top,w_n^\top)\} \cup \{(w_n^\bot,w_n^\bot)\}$;
    \item $\mu(w_0)= \mathpzc{AP}$; 
    \item for $1\leq i \leq n$, $\mu(w_i^{\top})= \mathpzc{AP}$ and $\mu(w_i^{\bot})= \mathpzc{AP} \setminus \{x_i\}$.
\end{itemize}

See Figure~\ref{Kbool} for an example of $\mathpzc{K}_{SAT}^{Var}$, with $Var=\{x_1, \ldots ,x_4\}$.

\begin{figure}[tb]
\centering
\begin{tikzpicture}[->,>=stealth',shorten >=1pt,auto,node distance=2.8cm,semithick,every node/.style={circle,draw,inner sep=2pt},every loop/.style={max distance=8mm}]  

    \node[style={double}] (0) {$\stackrel{w_0}{x_1,x_2,x_3,x_4}$};
    \node (1a) [above right of=0] {$\stackrel{w_1^\top}{x_1,x_2,x_3,x_4}$};
    \node (1b) [below right of=0] {$\stackrel{w_1^\bot}{x_2,x_3,x_4}$};
    \node (2a) [right of=1a] {$\stackrel{w_2^\top}{x_1,x_2,x_3,x_4}$};
    \node (2b) [right of=1b] {$\stackrel{w_2^\bot}{x_1,x_3,x_4}$};
    \node (3a) [right of=2a] {$\stackrel{w_3^\top}{x_1,x_2,x_3,x_4}$};
    \node (3b) [right of=2b] {$\stackrel{w_3^\bot}{x_1,x_2,x_4}$};
    \node (4a) [right of=3a] {$\stackrel{w_4^\top}{x_1,x_2,x_3,x_4}$};
    \node (4b) [right of=3b] {$\stackrel{w_4^\bot}{x_1,x_2,x_3}$};

  \path
    (0) edge (1a)
        edge (1b)
    (1a) edge (2a)
        edge (2b)
    (1b) edge (2a)
        edge (2b)    
    (2a) edge (3a)
        edge (3b)
    (2b) edge (3a)
        edge (3b)
    (3a) edge (4a)
        edge (4b)
    (3b) edge (4a)
        edge (4b)
    (4a) edge [loop below] (4a)
    (4b) edge [loop above] (4b)
    ;
\end{tikzpicture}
\caption{Kripke structure $\mathpzc{K}_{SAT}^{Var}$ associated with a SAT formula with variables $Var=\{x_1,x_2,x_3,x_4\}$.}\label{Kbool}
\end{figure}
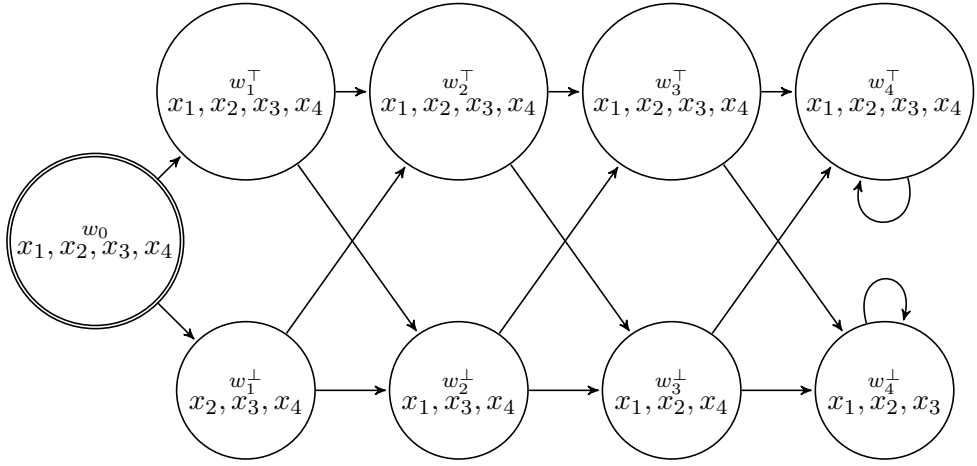

It is immediate to see that any initial track $\rho$ of any length induces a truth assignment to the variables of $Var$: for any $x_i \in Var$, $x_i$ evaluates to $\top$ if and only if $x_i \in \bigcap_{w\in\states(\rho)}\mu(w)$. Conversely, for any possible truth assignment to the variables in $Var$, there exists an initial track $\rho$ that induces such an assignment: we include in the track the state $w_i^{\top}$ if $x_i$ is assigned to $\top$, $w_i^{\bot}$ otherwise.


Let $\gamma=\neg\beta$. It holds that $\beta$ is satisfiable if and only if there exists an initial track $\rho\in\Trk_{\mathpzc{K}_{SAT}^{Var}}$ such that $\mathpzc{K}_{SAT}^{Var},\rho\models \beta$, that is, if and only if $\mathpzc{K}_{SAT}^{Var}\not\models \gamma$.
To conclude, we observe that $\mathpzc{K}_{SAT}^{Var}$ can be built with logarithmic working space.
\end{proof}

It immediately follows that checking whether $\mathpzc{K}\not\models\beta$ for $\beta\in\HSprop$ is NP-complete, thus model checking for formulas of $\HSprop$ is coNP-complete.
Moreover, since a pure propositional formula in $\HSprop$ is also a $\HSforall$ formula, \texttt{ProvideCounterex}$(\mathpzc{K},\psi)$ is at least as hard as checking whether $\mathpzc{K}\not\models \beta$ for $\beta\in\HSprop$. Thus, \texttt{ProvideCounterex}$(\mathpzc{K},\psi)$ is NP-complete, hence the model checking problem for $\HSforall$ is coNP-complete.

We conclude the section spending a few words about the complexity of the model checking problem for the fragment  $\AAbar$, also known as \emph{the logic of temporal neighborhood}. 
As a consequence of the lower bound for $\HSprop$, model checking for $\AAbar$ turns out to be coNP-hard as well. 
Moreover, the problem is in PSPACE, as $\AAbar$ is a subfragment of $\AAbarBbarEbar$.
Actually, in~\cite{MMPS16}, the authors proved that $\AAbar$ belongs to $\Thsq$ and is $\Th$-hard: the complexity class $\Th$ (respectively, $\Thsq$) contains the problems decided by a deterministic polynomial time algorithm which requires only $O(\log n)$ (respectively, $O(\log^2 n)$) queries to an NP oracle, being $n$ the input size~\cite{gottlob1995,schnoebelen2003}. Hence, such classes are higher than both NP and coNP in the polynomial time hierarchy.

%% file: corrcomplProof.tex
\begin{proof}
(Soundness)
The proof is by induction on the structure of the formula $\psi$. 

\begin{itemize}
    \item $\psi$ is a pure propositional formula $\beta$: let $\rho$ be a witness track for $d$; if \texttt{Check$\exists$}$(\mathpzc{K},\beta,d)$ has a successful computation, then $VAL(\beta,d)$ is true and so $\mathpzc{K},\rho\models\psi$.
    \item $\psi=\varphi_1\vee\varphi_2$: if \texttt{Check$\exists$}$(\mathpzc{K},\psi,d)$ has a successful computation, then, for some $i \in \{1,2\}$, \texttt{Check$\exists$}$(\mathpzc{K},\varphi_i,d)$ has a successful computation.
By the inductive hypothesis, there exists $\rho\in\Trk_\mathpzc{K}$ associated with $d$ such that $\mathpzc{K},\rho\models\varphi_i$, and thus $\mathpzc{K},\rho\models\varphi_1\vee\varphi_2$. 
    \item $\psi=\hsA\varphi$: if \texttt{Check$\exists$}$(\mathpzc{K},\psi,d)$ 
    has a successful computation, then there exists a witnessed $d'=(v_{in}',S',v_{fin}')$, with $v_{in}'=v_{fin}$, such that \texttt{Check$\exists$}$(\mathpzc{K},\varphi,d')$ has a successful computation. By the inductive hypothesis, there exists a track $\rho'$, associated with $d'$, such that $\mathpzc{K},\rho'\models\varphi$. If $\rho$ is a track associated with $d$ (which is witnessed by hypothesis), we have that $\lst(\rho)=\fst(\rho')=v_{fin}$ and, by definition, $\mathpzc{K},\rho\models\psi$.
    \item $\psi=\hsB\varphi$: if \texttt{Check$\exists$}$(\mathpzc{K},\psi,d)$ has a successful computation, then
    we must distinguish two possible cases.\newline
        $(i)$ There exists $d'=(v_{in},S',v_{fin}')$, witnessed by a track with $(v_{fin}',v_{fin})\in\delta$, such that $(v_{in},S'\cup \{v_{fin}'\},v_{fin})=d$, and \texttt{Check$\exists$}$(\mathpzc{K},\varphi,d')$ has a successful computation. By the inductive hypothesis, there exists a track $\rho'$, associated with $d'$, such that $\mathpzc{K},\rho'\models\varphi$. Hence $\mathpzc{K},\rho'\cdot v_{fin}\models\psi$ and $\rho'\cdot v_{fin}$ is associated with $d$.\newline 
        $(ii)$ There exist $d'=(v_{in},S',v_{fin}')$, witnessed by a track, and $d''=(v_{in}'',S'',v_{fin}'')$, witnessed by a track as well, such that $(v_{fin}',v_{in}'')\in\delta$, \texttt{concat}$(d',d'')=d$, and \texttt{Check$\exists$}$(\mathpzc{K},\varphi,d')$ has a successful computation. By the inductive hypothesis, there exists a track $\rho'$, associated with $d'$, such that $\mathpzc{K},\rho'\models\varphi$. Hence $\mathpzc{K},\rho'\cdot\rho''\models\psi$, where $\rho''$ is any track associated with $d''$ and  $\rho'\cdot\rho''$ is associated with $d$.
\end{itemize}
The case $\psi=\hsAt\varphi$ (respectively, $\psi=\hsE\varphi$) can be dealt with as $\psi=\hsA\varphi$ (respectively, $\psi=\hsB\varphi$).
%

\medskip
(Completeness) The proof is by induction on the structure of the formula $\psi$.
\begin{itemize}
    \item $\psi$ is a pure propositional formula $\beta$: if $\rho$ is associated with $d$ and $\mathpzc{K},\rho\models\beta$, then $VAL(\beta,d)=\top$, and thus \texttt{Check$\exists$}$(\mathpzc{K},\psi,d)$ has a successful computation.
    \item $\psi= \varphi_1\vee \varphi_2$: if there exists a track $\rho$, associated with $d$, such that $\mathpzc{K},\rho\models\varphi_1\vee \varphi_2$, then
    $\mathpzc{K},\rho\models\varphi_i$, for some $i \in \{1,2\}$. By the inductive hypothesis, \texttt{Check$\exists$}$(\mathpzc{K},\varphi_i,d)$ has a successful computation, and hence \texttt{Check$\exists$}$(\mathpzc{K},\psi,d)$ has a successful computation.
    \item $\psi=\hsA\varphi$: if there exists a track $\rho$, associated with $d$, such that $\mathpzc{K},\rho\models\hsA\varphi$, then, by definition, there exists a track $\overline{\rho}$, with $\fst(\overline{\rho})=\lst(\rho)=v_{fin}$, such that $\mathpzc{K},\overline{\rho}\models\varphi$. 
    If $d'=(v_{fin},S',v_{fin}')$ is the descriptor element for $\overline{\rho}$, then, by the inductive hypothesis, \texttt{Check$\exists$}$(\mathpzc{K},\varphi,d')$ has a successful computation. 
    Since there exists a computation where the non-deterministic call to $\texttt{aDescrEl}(\mathpzc{K},v_{fin}, \textsc{forw})$ returns the descriptor element $d'$ for $\overline{\rho}$, it follows that \texttt{Check$\exists$}$(\mathpzc{K},\psi,d)$ has a successful computation.
    \item $\psi=\hsB\varphi$: if there exists a track $\rho$, associated with $d$, such that $\mathpzc{K},\rho\models\hsB\varphi$, there are two possible cases.\newline
    %
        $(i)$ $\mathpzc{K},\overline{\rho}\models\varphi$, with $\rho=\overline{\rho}\cdot v_{fin}$ for some $\overline{\rho} \in \Trk_\mathpzc{K}$. If $d'=(v_{in},S',v_{fin}')$ is the descriptor element for $\overline{\rho}$, by the inductive hypothesis \texttt{Check$\exists$}$(\mathpzc{K},\varphi,d')$ has a successful computation. Since there is a computation where $\texttt{aDescrEl}(\mathpzc{K},v_{in},\textsc{forw})$ returns $d'$ and both $(v_{fin}',v_{fin}) \in\delta$ and $(v_{in},S'\cup \{v_{fin}'\},v_{fin})=d$, it follows that \texttt{Check$\exists$}$(\mathpzc{K},\psi,d)$ has a successful computation.\newline
        $(ii)$ $\mathpzc{K},\overline{\rho}\models\varphi$ with $\rho=\overline{\rho}\cdot \tilde{\rho}$ for some $\overline{\rho},\tilde{\rho}\in\Trk_\mathpzc{K}$. 
        Let $d'=(v_{in},S',v_{fin}')$ and $d''=(v_{in}'',S'',v_{fin}'')$ be the descriptor elements for $\overline{\rho}$ and $\tilde{\rho}$, respectively. Obviously, it holds that $\texttt{concat}(d',d'')=d$.
        By the inductive hypothesis, \texttt{Check$\exists$}$(\mathpzc{K},\varphi,d')$ has a successful computation. Since both $\overline{\rho}$ and $\tilde{\rho}$ are witnessed, there is a computation where the calls to $\texttt{aDescrEl}(\mathpzc{K},v_{in},\textsc{forw})$ and $\texttt{aDescrEl}(\mathpzc{K},v_{in}'',\textsc{forw})$ non-deterministically return $d'$ and $d''$, respectively, and $(v_{fin}',v_{in}'')\in\delta$ is non-deterministically chosen. Hence, \texttt{Check$\exists$}$(\mathpzc{K},\psi,d)$ has a successful computation.
\end{itemize}
The case $\psi=\hsAt\varphi$ (respectively, $\psi=\hsE\varphi$) can be dealt with as $\psi=\hsA\varphi$ (respectively, $\psi=\hsB\varphi$).
\end{proof}

It is worth pointing out that \texttt{Check$\exists$}$(\mathpzc{K},\psi,d)$ cannot deal with $\hsBt$ and $\hsEt$ modalities. To cope with them, descriptor elements are not enough: the whole descriptors must be considered.

%% file: conclus.tex
\section{Conclusions and future work}

In this paper, we have studied the model checking problem for some fragments of Halpern and Shoham's modal logic of time intervals.
First, we have considered the large fragment $\AAbarBBbarEbar$, and devised an EXPSPACE model checking algorithm for it, which rests on a contraction method that allows us to restrict the verification 
of the input formula to a finite subset of tracks of bounded size, called track representatives.
We have also proved that the problem is PSPACE-hard, NEXP-hard if a suitable succinct encoding of formulas is allowed. As a matter of fact, in the latter case, the problem can also be proved coNEXP-hard, and thus we conjecture that a tighter lower bound can be established (for instance, EXPSPACE-hardness).
Then, we identified some other HS fragments, namely, $\AAbarBbarEbar$, $\HSforall$, and 
$\AAbar$, whose model checking problem turns out to be (computationally) much simpler than that
of full HS and of $\AAbarBBbarEbar$, and comparable to that of point-based temporal 
logics (as an example, the model checking problem for $\AAbarBbarEbar$ is PSPACE-complete, and has thus the same complexity as 
LTL).
Luckily, these fragments are expressive enough to capture meaningful properties of state-transition systems, such as, for instance, mutual exclusion, state reachability, and non-starvation.

One may wonder whether, given the homogeneity assumption, there is the possibility to reduce the model checking problem for HS fragments over finite Kripke structures to a point-based setting. Such an issue has been systematically dealt with in \cite{DBLP:conf/fsttcs/BozzelliMMPS16}. Together with Laura Bozzelli and Pietro Sala, we consider three semantic variants of HS: the one we introduced in \cite{DBLP:conf/time/MontanariMPP14} and we used in the subsequent papers, including the present one, called state-based semantics, which allows branching in the past and in the future, the computation-tree-based semantics, allowing branching only in the future, and the linear semantics, disallowing branching. These variants are compared, as for their expressiveness, among themselves and to standard temporal logics, getting a complete picture. In particular, we show that (i) HS with computation-tree-based semantics is equivalent to finitary CTL* and strictly included in HS with state-based semantics, and (ii) HS with linear semantics is equivalent to LTL and incomparable to HS with state-based semantics.

As for future work, we are currently exploring two main research directions. On the one hand, we are looking for other well-behaved fragments of HS; on the other hand, we are thinking of possible ways of relaxing the homogeneity assumption. As for the latter, a promising direction has been recently outlined by Lomuscio and Michaliszyn, who proposed to use regular expressions to define the behavior of proposition letters over intervals in terms of the component states \cite{LM15}. Our ultimate goal is to be able to deal with interval properties that can only be predicated over time intervals considered as a whole. This is the case, for instance, of temporal aggregations (think of a constraint on the average speed of a moving device during a given time period).  In this respect, the existing work on Duration Calculus (DC) model checking seems to be relevant. DC extends interval temporal logic with an explicit notion of state: states are denoted by state expressions and characterized by a duration (the time period during which the system remains in a given state). Recent results on DC model checking and an account of related work can be found in~\cite{HPB14}.
%

%% file: appendix.tex
\section{Proofs}

\subsection{Proof of Lemma \ref{symmextlemma}}\label{symmextlemmaProof}
In the proof, we will exploit the fact that if two tracks in $\Trk_\mathpzc{K}$ have the same $B_{k+1}$-descriptor, then they also have the same $B_k$-descriptor. The latter can indeed be obtained from the former by removing the nodes at depth $k+1$ (leaves) and then deleting isomorphic subtrees possibly originated by the removal.

\begin{proof}
By induction on $k\geq 0$.

\emph{Base case} ($k=0$): let us assume $\rho_1$ and $\rho_2$ are associated with the descriptor element $(v_{in},S,v_{fin})$ and $\rho_1'$ and $\rho_2'$ with $(v_{in}',S',v_{fin}')$. Thus $\rho_1\cdot \rho_1'$ and $\rho_2\cdot\rho_2'$ are both described by the descriptor element $(v_{in},S\cup\{v_{fin},v_{in}'\}\cup S',v_{fin}')$.

\emph{Inductive step} ($k>0$): let $\mathpzc{D}_{B_k}$ be the $B_k$-descriptor for $\rho_1\cdot \rho_1'$ and $\mathpzc{D}_{B_k}'$ be the one for $\rho_2\cdot\rho_2'$: their roots are the same, as for $k=0$; let us now consider a prefix $\rho$ of $\rho_1\cdot \rho_1'$:
\begin{itemize}
    \item if $\rho$ is a proper prefix of $\rho_1$, since $\rho_1$ and $\rho_2$ have the same $B_k$-descriptor, there exists a prefix $\overline{\rho}$ of $\rho_2$ associated with the same subtree as $\rho$ of depth $k-1$ in the descriptor for $\rho_1$ (and $\rho_2$); 
    \item for $\rho=\rho_1$, it holds that $\rho_1$ and $\rho_2$ have the same $B_{k-1}$-descriptor because they have the same $B_k$-descriptor;
    \item if $\rho$ is a proper prefix of $\rho_1\cdot \rho_1'$ such that $\rho=\rho_1\cdot\tilde{\rho}_1$ for some prefix $\tilde{\rho}_1$ of $\rho_1'$, then two cases have to be taken into account:
    \begin{itemize}
        \item if $|\tilde{\rho}_1|=1$, then $\tilde{\rho}_1=v_{in}'$; but also $\fst(\rho_2')=v_{in}'$. Let us now consider the $B_{k-1}$-descriptors for $\rho_1\cdot v_{in}'$ and $\rho_2\cdot v_{in}'$: the labels of the roots are the same, namely $(v_{in},S\cup \{v_{fin}\},v_{in}')$, then the subtrees of depth $k-2$ are exactly the same as in $\rho_1$ and $\rho_2$'s $B_{k-1}$-descriptor, (possibly) with the addition of the $B_{k-2}$-descriptor for $\rho_1$ (which is equal to that for $\rho_2$). Thus $\rho_1\cdot v_{in}'$ and $\rho_2\cdot v_{in}'$ have the same $B_{k-1}$-descriptor;
        \item otherwise, since $\tilde{\rho}_1$ is a prefix of $\rho_1'$ of length at least 2, and $\rho_1'$ and $\rho_2'$ have the same $B_k$-descriptor, there exists a prefix $\tilde{\rho}_2$ of $\rho_2'$ associated with the same subtree of depth $k-1$ as $\tilde{\rho}_1$ (in the $B_k$-descriptor for $\rho_1'$). Hence, by inductive hypothesis, $\rho_1\cdot \tilde{\rho}_1$ and $\rho_2\cdot\tilde{\rho}_2$ have the same $B_{k-1}$-descriptor.
    \end{itemize}
\end{itemize}
 
Therefore we have shown that for any proper prefix of $\rho_1\cdot \rho_1'$ there exists a proper prefix of $\rho_2\cdot\rho_2'$ having the same $B_{k-1}$-descriptor. The inverse can be shown by symmetry. Thus $\mathpzc{D}_{B_k}$ is equal to $\mathpzc{D}_{B_k}'$.
\end{proof}

\subsection{Proof of Theorem \ref{thdesc}}\label{thdescProof}
\begin{proof}
The proof is by induction on $i \geq u+1$.\\
(Case $i=u+1$) We consider two cases: 
\begin{enumerate}
    \item if $\rho_{ds}(u)=\rho_{ds}(u+1)=d\in\mathpzc{C}$, then we have $Q_{-2}(u)=\mathpzc{C}\setminus \{d\}$, and $Q_{-1}(u)=\{d\}$, $Q_0(u)=Q_1(u)=\cdots =Q_s(u)=\emptyset$. Moreover, it holds that $Q_{-2}(u+1)=\mathpzc{C}\setminus \{d\}$, $Q_{-1}(u)=\emptyset$, $Q_0(u)=\{d\}$, and $Q_1(u)=Q_2(u)=\cdots = Q_s(u)=\emptyset$. $c(u)>_{lex}c(u+1)$ and the thesis follows.
    \item if $d,d'\in\mathpzc{C}$, with $d\neq d'$, $\rho_{ds}(u)=d$, and $\rho_{ds}(u+1)=d'$, then we have $Q_{-2}(u)=\mathpzc{C}\setminus \{d\}$, $Q_{-1}(u)=\{d\}$, and $Q_0(u)=Q_1(u)=\cdots =Q_s(u)=\emptyset$. Moreover, it holds that $Q_{-2}(u+1)=\mathpzc{C}\setminus \{d,d'\}$, $Q_{-1}(u)=\{d,d'\}$, $Q_0(u)=Q_1(u)=\cdots =Q_s(u)=\emptyset$, and $c(u)>_{lex}c(u+1)$, implying the thesis.
\end{enumerate}
(Case $i>u+1$) In the following, we say that $\rho_{ds}(\ell)$ and $\rho_{ds}(m)$ ($\ell<m$) are consecutive occurrences of a descriptor element $d$ if there are no other occurrences of $d$ in $\rho_{ds}(\ell+1, m-1)$. We consider the following cases:
\begin{enumerate}
    \item If $\rho_{ds}(i)$ is the first occurrence of $d\in\mathpzc{C}$, then $d\in Q_{-2}(i-1)$, $d\in Q_{-1}(i)$, and it holds that $c(i-1)>_{lex}c(i)$.
    \item If $\rho_{ds}(i)$ is the second occurrence of $d\in\mathpzc{C}$, according to the definition, $\rho_{ds}(i)$ can not be 1-indistinguishable from the previous occurrence of $d$, and thus $d\in Q_{-1}(i-1)$ ($\rho_{ds}(u,i-1)$ contains the first occurrence of $d$) and $d\in Q_0(i)$, proving that $c(i-1)>_{lex}c(i)$.
    \item If $\rho_{ds}(i)$ is at least the third occurrence of $d\in\mathpzc{C}$, but $\rho_{ds}(i)$ is \emph{not} $1$-indistinguishable from  the immediately preceding occurrence of $d$, $\rho_{ds}(i')$, with $i'<i$, then $DElm(\rho_{ds}(u,\allowbreak i'-1))\subset DElm(\rho_{ds}(u,i-1))$. Hence, there exists a first occurrence of some $d'\in\mathpzc{C}$ in $\rho_{ds}(i'+1, i-1)$, say $\rho_{ds}(j)=d'$, for $i'+1\leq j\leq i-1$. Thus, $d\in Q_{-1}(j)$, $\cdots$ , $d\in Q_{-1}(i-1)$, and $d\in Q_0(i)$, proving that $c(i-1)>_{lex}c(i)$.
    \item In the remaining cases, we assume that $\rho_{ds}(i)$ is \emph{at least the third occurrence} of $d\in\mathpzc{C}$. If $\rho_{ds}(i-1)$ and $\rho_{ds}(i)$ are both occurrences of $d\in\mathpzc{C}$ and $\rho_{ds}(i-1)$ is $t$-indistinguishable, for some $t>0$, and not $(t+1)$-indistinguishable, from the immediately preceding occurrence of $d$, then $\rho_{ds}(i-1)$ and $\rho_{ds}(i)$ are exactly $(t+1)$-indistinguishable. Thus, $d\in Q_{t}(i-1)$ and $d\in Q_{t+1}(i)$, implying that $c(i-1)>_{lex}c(i)$ (as a particular case, if $\rho_{ds}(i-1)$ and the immediately preceding occurrence are not 1-indistinguishable, then $\rho_{ds}(i-1)$ and $\rho_{ds}(i)$ are at most 1-indistinguishable).
    \item\label{b1b} If $\rho_{ds}(i)$ is exactly $1$-indistinguishable from the immediately preceding occurrence of $d$, $\rho_{ds}(j)$, with $j<i-1$, then $DElm(\rho_{ds}(u,j-1))= DElm(\rho_{ds}(u,i-1))$, and there are no first occurrences of any $d'\in\mathpzc{C}$ in $\rho_{ds}(j,i-1)$. 
    If $\rho_{ds}(j)$ is not 1-indistinguishable from its previous occurrence of $d$, it immediately follows that $d\in Q_0(j)$, $\cdots$, $d\in Q_0(i-1)$ and $d\in Q_1(i)$, implying that $c(i-1)>_{lex}c(i)$.
    
    Otherwise, there exists $j< i'<i$ such that $\rho_{ds}(i')=d''\in\mathpzc{C}$ is not 1-indistinguishable from any occurrence of $d''$ before $j$ (as a matter of fact, if this was not the case, $\rho_{ds}(i)$ and $\rho_{ds}(j)$ would be 2-indistinguishable); in particular, $\rho_{ds}(i')$ is not 1-indistinguishable from the last occurrence of $d''$ before $j$, say $\rho_{ds}(j')$, for some $j'<j$ (such a $j'$ exists since there are no first occurrences in $\rho_{ds}(j+1,i-1)$). 
    Now, if by contradiction every pair of consecutive occurrences of $d''$ in $\rho_{ds}(j',i')$ were 1-indistinguishable, then by Corollary \ref{propC} $\rho_{ds}(j')$ and $\rho_{ds}(i')$ would be 1-indistinguishable. Thus, a pair of consecutive occurrences of $d''$ exists, where the second element in the pair is $\rho_{ds}(\ell)=d''$, with $j<\ell <i$, such that they are not 1-indistinguishable. By inductive hypothesis, $d''\in Q_{-1}(\ell -1)$ and $d''\in Q_0(\ell)$. Therefore, $d\in Q_0(\ell)$, $\cdots$, $d\in Q_0(i-1)$ (recall that there are no first occurrences between $j$ and $i$) and $d\in Q_1(i)$, proving that $c(i-1)>_{lex}c(i)$.
    \item If $\rho_{ds}(j)=d\in\mathpzc{C}$ is at most $t$-indistinguishable (for some $t\geq 1$) from a preceding occurrence of $d$ and $\rho_{ds}(j)$ and $\rho_{ds}(i)=d$, with $j<i-1$, are $(t+1)$-indistinguishable consecutive occurrences of $d$ (by definition of indistinguishability, $\rho_{ds}(j)$ and $\rho_{ds}(i)$ can not be more than $(t+1)$-indistinguishable), any occurrence of $d'\in\mathpzc{C}$ in $\rho_{ds}(j+1,i-1)$ is (at least) $t$-indistinguishable from another occurrence of $d'$ before $j$. By Proposition \ref{propA}, all pairs of consecutive occurrences of $d'$ in $\rho_{ds}(j+1,i-1)$ are (at least) $t$-indistinguishable, hence $d\in Q_t(j)$, $\cdots$, $d\in Q_t(i-1)$ and finally $d\in Q_{t+1}(i)$, proving that $c(i-1)>_{lex}c(i)$.
    \item If $\rho_{ds}(j)=d\in\mathpzc{C}$ is at most $t$-indistinguishable (for some $t\geq 1$) from a preceding occurrence of $d$, and $\rho_{ds}(j)$ and $\rho_{ds}(i)=d$, with $j<i-1$, are consecutive occurrences of $d$ which are at most $\overline{t}$-indistinguishable, for some $1\leq\overline{t}\leq t$, we preliminarily observe that $DElm(\rho_{ds}(u,j-1))= DElm(\rho_{ds}(u,i-1))$. Then, if some $d''\in\mathpzc{C}$, with $d'' \neq d$, occurs in $\rho_{ds}(j+1,i-1)$ and it is not $1$-indistinguishable from any occurrence of $d''$ before $j$, then $\overline{t}=1$ and we are again in case~\ref{b1b}.
    
    Otherwise, all the occurrences of descriptor elements in $\rho_{ds}(j+1,i-1)$ are (at least) 1-indistinguishable from other occurrences before $j$. Moreover, there exists $j<i'<i$ such that $\rho_{ds}(i')=d'\in\mathpzc{C},d\neq d'$, and it is at most $(\overline{t}-1)$-indistinguishable from another occurrence of $d'$ before $j$. Analogously to case \ref{b1b}, by Proposition \ref{propA}, $\rho_{ds}(i')$ must be \mbox{$(\overline{t}-1)$-indistinguishable} from the last occurrence of $d'$ before $j$, say $\rho_{ds}(j')$, with $j'<j$. 
    But two consecutive occurrences of $d'$ in $\rho_{ds}(j',i')$ must then be at most $(\overline{t}-1)$-indistinguishable (if all pairs of occurrences of $d'$ in $\rho_{ds}(j', i')$ were $\overline{t}$-indistinguishable, $\rho_{ds}(i')$ and $\rho_{ds}(j')$ would be $\overline{t}$-indistinguishable as well), where the second occurrence is $\rho_{ds}(\ell)=d'$ for some $j<\ell\leq i'$. By applying the inductive hypothesis, we have $d'\in Q_{\overline{t}-2}(\ell-1)$ and $d'\in Q_{\overline{t}-1}(\ell)$. As a consequence, we have $d\in Q_{\overline{t}-1}(\ell)$, $\cdots$, $d\in Q_{\overline{t}-1}(i-1)$ (all descriptor elements in $\rho_{ds}(j, i)$ are at least $(\overline{t}-1)$-indistinguishable from other occurrences before $j$) and finally $d\in Q_{\overline{t}}(i)$, implying that $c(i-1)>_{lex}c(i)$.\qedhere
\end{enumerate}
\end{proof}
It is worth pointing out that, from the proof of the theorem, it follows that the definition of $f$ is in fact redundant: cases (c) and (e) never occur.

\subsection{Proof of Lemma~\ref{lemmamdc}}\label{explCheck}
\begin{proof}
The proof is by induction on the structure of $\psi$.
 The cases in which $\psi=\top$, $\psi=\bot$, $\psi=p\in\mathpzc{AP}$ are trivial. 
 The cases in which $\psi=\neg\varphi$, $\psi=\varphi_1\wedge\varphi_2$ are also trivial and omitted. We focus on the remaining cases. 
\begin{itemize}
    \item $\psi=\hsA\varphi$. If $\mathpzc{K},\tilde{\rho}\models \psi$, then there exists $\rho\in \Trk_\mathpzc{K}$ such that $\lst(\tilde{\rho})=\fst(\rho)$ and $\mathpzc{K},\rho\models \varphi$. By Theorem \ref{corrunr} the unravelling procedure returns $\overline{\rho}\in \Trk_\mathpzc{K}$ such that $\fst(\overline{\rho})=\fst(\rho)$ and $\overline{\rho}$ and $\rho$ have the same $B_k$-descriptor, thus $\mathpzc{K},\overline{\rho}\models \varphi$. By the inductive hypothesis, \texttt{Check}$(\mathpzc{K},k,\varphi,\overline{\rho})=1$, hence \texttt{Check}$(\mathpzc{K},k,\psi,\tilde{\rho})=1$.
    
    Vice versa, if \texttt{Check}$(\mathpzc{K},k,\psi,\tilde{\rho})=1$, there exists $\rho\in \Trk_\mathpzc{K}$ such that $\lst(\tilde{\rho})=\fst(\rho)$ and \texttt{Check}$(\mathpzc{K},k,\varphi,\rho)=1$. By the inductive hypothesis, $\mathpzc{K},\rho\models \varphi$, hence $\mathpzc{K},\tilde{\rho}\models \psi$.
    
    \item $\psi=\hsAt\varphi$. The proof is symmetric to the case $\psi=\hsA\varphi$.
    
    \item $\psi=\hsB\varphi$. If $\mathpzc{K},\tilde{\rho}\models \psi$, there exists $\rho\in\Pref(\tilde{\rho})$ such that $\mathpzc{K},\rho\models \varphi$. By the inductive hypothesis, \texttt{Check}$(\mathpzc{K},k-1,\varphi,\rho)=1$. Since all prefixes of $\tilde{\rho}$ are checked, \texttt{Check}$(\mathpzc{K},k,\psi,\tilde{\rho})=1$. 
    \emph{Note that, by definition of descriptor, if $\tilde{\rho}$ is a track representative of a $B_k$-descriptor $\mathpzc{D}_{B_k}$, a prefix of $\tilde{\rho}$ is a representative of a $B_{k-1}$-descriptor, whose root is a child of the root of $\mathpzc{D}_{B_k}$.}
    
    Vice versa, if \texttt{Check}$(\mathpzc{K},k,\psi,\tilde{\rho})=1$, then for some track $\rho\in\Pref(\tilde{\rho})$, we have  \texttt{Check}$(\mathpzc{K},k-1,\varphi,\rho)=1$. By the inductive hypothesis $\mathpzc{K},\rho\models \varphi$, hence $\mathpzc{K},\tilde{\rho}\models \psi$.
    
    \item $\psi=\hsBt\varphi$. If $\mathpzc{K},\tilde{\rho}\models \psi$, then there exists $\rho$ such that $\tilde{\rho}\cdot \rho \in \Trk_\mathpzc{K}$ for which $\mathpzc{K},\tilde{\rho}\cdot\rho\models \varphi$. If $|\rho|=1$, since by the inductive hypothesis \texttt{Check}$(\mathpzc{K},k,\varphi,\tilde{\rho}\cdot \rho)=1$, then \texttt{Check}$(\mathpzc{K},k,\psi,\tilde{\rho})=1$. Otherwise, the unravelling algorithm returns a track $\overline{\rho}$ with the same $B_k$-descriptor as $\rho$. Thus, by the extension Proposition \ref{extBk}, $\tilde{\rho}\cdot \rho$ and $\tilde{\rho}\cdot \overline{\rho}$ have the same $B_k$-descriptor. Thus $\mathpzc{K},\tilde{\rho}\cdot\overline{\rho}\models \varphi$. So (by inductive hypothesis) \texttt{Check}$(\mathpzc{K},k,\varphi,\tilde{\rho}\cdot\overline{\rho})=1$ implying that \texttt{Check}$(\mathpzc{K},k,\psi,\tilde{\rho})=1$. 
    \emph{Note that, given two tracks $\rho,\rho'$ of $\mathpzc{K}$, if we are considering $\overline{\rho}$ as the track representative of the $B_k$-descriptor of $\rho$, and the unravelling algorithm returns $\overline{\rho}'$ as the representative of the $B_k$-descriptor of $\rho'$, since by Lemma \ref{symmextlemma} $\rho\cdot\rho'$ and $\overline{\rho}\cdot\overline{\rho}'$ have the same $B_k$-descriptor, we have that $\overline{\rho}\cdot\overline{\rho}'$ is the representative of the $B_k$-descriptor of $\rho\cdot\rho'$.}

     Vice versa, if \texttt{Check}$(\mathpzc{K},k,\psi,\tilde{\rho})=1$, there exists $\rho$ such that $\tilde{\rho}\cdot \rho \in \Trk_\mathpzc{K}$ and \texttt{Check}$(\mathpzc{K},k,\varphi,\tilde{\rho}\cdot\rho)=1$. By the inductive hypothesis, $\mathpzc{K},\tilde{\rho}\cdot\rho\models \varphi$, hence $\mathpzc{K},\tilde{\rho}\models \psi$.
     
     \item $\psi=\hsEt\varphi$. The proof is symmetric to the case $\psi=\hsBt\varphi$.\qedhere
\end{itemize} 
\end{proof}

\subsection{Proof of Theorem~\ref{thModcheck}}\label{proofModCheck2}
\begin{proof}
	If $\mathpzc{K}\models \psi$, then for all $\rho\in\Trk_\mathpzc{K}$ such that $\fst(\rho)=w_0$ is the initial state of $\mathpzc{K}$, we have $\mathpzc{K},\rho\models \psi$.
	By  Lemma \ref{lemmamdc}, 
it follows that $\texttt{Check}(\mathpzc{K},\nestb(\psi),\psi,\rho)=1$. Now, the unravelling procedure returns a subset of the initial tracks. This implies that \texttt{ModCheck}$(\mathpzc{K},\psi)=1$.

	On the other hand, if \texttt{ModCheck}$(\mathpzc{K},\psi)=1$, then for any track $\rho$ with $\fst(\rho)=w_0$ \emph{returned by the unravelling algorithm}, $\texttt{Check}(\mathpzc{K},\nestb(\psi),\psi,\rho)=1$ and, by Lemma \ref{lemmamdc}, $\mathpzc{K},\rho\models \psi$. 
	Assume now that a track $\tilde{\rho}$, with $\fst(\tilde{\rho})=w_0$, is \emph{not} returned by the unravelling algorithm. By Theorem~\ref{corrunr}, there exists a track $\overline{\rho}$, with $\fst(\overline{\rho})=w_0$, which is returned in place of $\tilde{\rho}$ and $\overline{\rho}$ has the same $B_k$-descriptor as $\tilde{\rho}$ (with $k=\nestb(\psi)$). Since $\mathpzc{K},\tilde{\rho}\models \psi \iff \mathpzc{K},\overline{\rho}\models \psi$ (by Theorem~\ref{satPresB}) and $\mathpzc{K},\overline{\rho}\models \psi$, we get that $\mathpzc{K},\tilde{\rho}\models \psi$. So all tracks starting from state $w_0$ model $\psi$, implying that $\mathpzc{K}\models \psi$.
\end{proof}

\input{NEXPhardness}

\subsection{Proof of Lemma~\ref{lemmaABbar}}\label{sec:lemmaABbarProof}
\begin{proof}
The proof is by induction on the complexity of $\psi$.
\begin{itemize}
    \item $\psi =p$, with $p\in\mathpzc{AP}$ ($p\ell(p)=\{p\}$). If $\mathpzc{K},\rho\models p$, then $p\in \mathpzc{L}(\mathpzc{K},\rho)$ and hence $p\in \mathpzc{L}(\mathpzc{K}_{\,|p\ell(\psi)},\rho)$. By hypothesis, it immediately follows that $p\in \mathpzc{L}(\mathpzc{K}'_{\,|p\ell(\psi)},\rho')$, and thus $p\in \mathpzc{L}(\mathpzc{K}',\rho')$ and $\mathpzc{K}',\rho'\models p$.
    \item $\psi=\neg\phi$ ($p\ell(\phi)=p\ell(\psi)$). If $\mathpzc{K},\rho\models \neg\phi$, then $\mathpzc{K},\rho\not\models \phi$. By the inductive hypothesis, $\mathpzc{K}',\rho'\not\models \phi$ and thus $\mathpzc{K}',\rho'\models \neg\phi$.
    \item $\psi=\phi_1\wedge\phi_2$. If $\mathpzc{K},\rho\models \phi_1\wedge\phi_2$, then in particular $\mathpzc{K},\rho\models \phi_1$. Since, by hypothesis, $\mathpzc{L}(\mathpzc{K}_{\,|p\ell(\psi)},\rho)=\mathpzc{L}(\mathpzc{K}'_{\,|p\ell(\psi)},\rho')$ and $reach(\mathpzc{K}_{\,|p\ell(\psi)},\lst(\rho))\sim reach(\mathpzc{K}'_{\,|p\ell(\psi)},\lst(\rho'))$, it holds that $\mathpzc{L}(\mathpzc{K}_{\,|p\ell(\phi_1)},\rho)=\mathpzc{L}(\mathpzc{K}'_{\,|p\ell(\phi_1)},\rho')$ and $reach(\mathpzc{K}_{\,|p\ell(\phi_1)},\lst(\rho))\sim reach(\mathpzc{K}'_{\,|p\ell(\phi_1)},\lst(\rho'))$, as $p\ell(\phi_1)\subseteq p\ell(\psi)$. By the inductive hypothesis, $\mathpzc{K}',\rho'\models \phi_1$. 
The same argument works for $\phi_2$. The thesis 
follows.
    \item $\psi=\hsA \phi$. If $\mathpzc{K},\rho\models \hsA \phi$, there exists a track $\overline{\rho}\in\Trk_\mathpzc{K}$ such that $\fst(\overline{\rho})=\lst(\rho)$ and $\mathpzc{K},\overline{\rho}\models \phi$, with $p\ell(\phi) = p\ell(\psi)$.     
    By hypothesis, it holds that $reach(\mathpzc{K}_{\,|p\ell(\psi)},\lst(\rho))\sim reach(\mathpzc{K}'_{\,|p\ell(\psi)},\lst(\rho'))$.
    Hence, there exists a track $\overline{\rho}'\in\Trk_{\mathpzc{K}'}$, with $\fst(\overline{\rho}')=\lst(\rho')$, such that $|\overline{\rho}|=|\overline{\rho}'|$ and for all $0\leq i \leq |\overline{\rho}|-1$, $f(\overline{\rho}(i)) = \overline{\rho}'(i)$, where $f$ is the (an) isomorphism between $reach(\mathpzc{K}_{\,|p\ell(\psi)},\lst(\rho))$ and $reach(\mathpzc{K}'_{\,|p\ell(\psi)},\lst(\rho'))$. It immediately follows that 
    $\mathpzc{L}(\mathpzc{K}_{\,|p\ell(\phi)},\overline{\rho})=\mathpzc{L}(\mathpzc{K}'_{\,|p\ell(\phi)},\overline{\rho}')$.
    
    We now prove that $reach(\mathpzc{K}_{\,|p\ell(\phi)},\lst(\overline{\rho}))\sim reach(\mathpzc{K}'_{\, |p\ell(\phi)}, \lst(\overline{\rho}'))$. To this end, it suffices to prove that the restriction of 
the isomorphism 
$f$ to the states of $reach(\mathpzc{K}_{\,|p\ell(\phi)},\lst(\overline{\rho}))$, say $f'$, is an isomorphism between $reach(\mathpzc{K}_{\,|p\ell(\phi)},\lst(\overline{\rho}))$ and $reach(\mathpzc{K}'_{\,|p\ell(\phi)},\lst(\overline{\rho}'))$ (note that  $reach(\mathpzc{K}_{\,|p\ell(\phi)},\lst(\overline{\rho}))$ is a subgraph of $reach(\mathpzc{K}_{\,|p\ell(\psi)},\lst(\rho))$). 
    First, it holds that $f(\lst(\overline{\rho}))=f'(\lst(\overline{\rho}))=\lst(\overline{\rho}')$.
    Next, if $w$ is any state of $reach(\mathpzc{K}_{\,|p\ell(\phi)},\lst(\overline{\rho}))$, then $f(w)=f'(w)=w'$ is a state of $reach(\mathpzc{K}'_{\,|p\ell(\phi)},\lst(\overline{\rho}'))$, as from the existence of a track from $\lst(\overline{\rho})$ to $w$, it follows that there is an isomorphic track (w.r.t. $f$) from $\lst(\overline{\rho}')$ to $w'$. 
    Moreover, if $(w,\overline{w})\in\delta$, then $\overline{w}$ 
belongs to $reach(\mathpzc{K}_{\,|p\ell(\phi)},\lst(\overline{\rho}))$, and thus $(w',f(\overline{w}))\in\delta'$ and $f(\overline{w})=f'(\overline{w})$ 
belongs to $reach(\mathpzc{K}'_{\,|p\ell(\phi)},\lst(\overline{\rho}'))$. 
We can conclude that,
for any two states  $v, v'$ of $reach(\mathpzc{K}_{\,|p\ell(\phi)},\lst(\overline{\rho}))$, it holds that $(v,v')$ is an edge if and only if $(f'(v),f'(v'))$ is an edge of $reach(\mathpzc{K}'_{\,|p\ell(\phi)},\lst(\overline{\rho}'))$.    
    
    By the inductive hypothesis, $\mathpzc{K}',\overline{\rho}'\models \phi$ and hence $\mathpzc{K}',\rho'\models \hsA\phi$.
    
    \item $\psi=\hsBt \phi$. If $\mathpzc{K},\rho\models \hsBt \phi$, then $\mathpzc{K},\rho\cdot\overline{\rho}\models \phi$, with $p\ell(\psi)=p\ell(\phi)$, where $\rho\cdot\overline{\rho}\in\Trk_\mathpzc{K}$ and $\overline{\rho}$ is either a single state or a proper track. 
In analogy to the previous case, let $\overline{\rho}'\in\Trk_{\mathpzc{K}'}$ such that $|\overline{\rho}|=|\overline{\rho}'|$ and, for all $0\leq i <|\overline{\rho}|$, $f(\overline{\rho}(i))=\overline{\rho}'(i)$, where $f$ is the isomorphism between $reach(\mathpzc{K}_{\,|p\ell(\psi)},\lst(\rho))$ and $reach(\mathpzc{K}'_{\,|p\ell(\psi)},\lst(\rho'))$. 
    Since $f(\lst(\rho))=\lst(\rho ')$, by definition of isomorphism, $(\lst(\rho),\fst(\overline{\rho}))\in\delta$ implies $(\lst(\rho'),\fst(\overline{\rho}'))\in\delta'$.
    Therefore 
    $\mathpzc{L}(\mathpzc{K}_{\,|p\ell(\phi)},\overline{\rho})=\mathpzc{L}(\mathpzc{K}'_{\,|p\ell(\phi)},\overline{\rho}')$ and $reach(\mathpzc{K}_{\,|p\ell(\phi)},\lst(\overline{\rho}))\sim reach(\mathpzc{K}'_{\,|p\ell(\phi)},\lst(\overline{\rho}'))$.
    Finally, 
    \begin{multline*}    
    \mathpzc{L}(\mathpzc{K}_{\,|p\ell(\phi)},\rho\cdot\overline{\rho})=\mathpzc{L}(\mathpzc{K}_{\,|p\ell(\phi)},\rho)\cap\mathpzc{L}(\mathpzc{K}_{\,|p\ell(\phi)},\overline{\rho})= \\
    \mathpzc{L}(\mathpzc{K}'_{\,|p\ell(\phi)},\rho')\cap\mathpzc{L}(\mathpzc{K}'_{\,|p\ell(\phi)},\overline{\rho}')=\mathpzc{L}(\mathpzc{K}'_{\,|p\ell(\phi)},\rho'\cdot\overline{\rho}')
    \end{multline*}
    and 
$reach(\mathpzc{K}_{\,|p\ell(\phi)},\lst(\rho\cdot\overline{\rho}))\sim reach(\mathpzc{K}'_{\,|p\ell(\phi)},\lst(\rho'\cdot\overline{\rho}'))$.
By the inductive hypothesis, $\mathpzc{K}',\rho'\cdot\overline{\rho}'\models \phi$ and thus $\mathpzc{K}',\rho'\models \hsBt \phi$.\qedhere
\end{itemize}
\end{proof}

%% file: NEXPhardness.tex
\subsection{NEXP-hardness of succinct $\AAbarBBbarEbar$}\label{sec:succAAbarBBbarEbarHard}

In Section~\ref{sec:representatives}, we proved that the model checking problem for $\AAbarBBbarEbar$ formulas is in EXPSPACE, and, in Section~\ref{subsec:AAbarBbarEbar}, that it is PSPACE-hard.
Here we prove that the model checking problem for $\AAbarBBbarEbar$ is in between EXPSPACE and NEXP when a suitable encoding of formulas is exploited. Such an encoding is \emph{succinct}, in the sense that the following binary-encoded shorthands are used: $\hsB^k\psi$ stands for $k$ repetitions of $\hsB$ before $\psi$,
where $k$ is represented in binary (the same for all the other HS modalities); 
moreover, $\bigwedge_{i=l,\cdots ,r} \psi(i)$
denotes a conjunction of formulas which contain some occurrences of the index $i$ as exponents ($l$ and $r$ are binary encoded naturals), e.g., $\bigwedge_{i=1,\cdots,5}\hsB^i \top$.
Finally, we denote by $\expand(\psi)$ the expanded form of $\psi$, where all exponents $k$ are removed from $\psi$, by explicitly repeating $k$ times each HS modality with such an exponent, and big conjunctions are replaced by conjunctions of formulas without indexes. 

It is not difficult to show that there exists a constant $c>0$ such that, for all succinct $\AAbarBBbarEbar$ formulas $\psi$, $|\expand(\psi)|\leq 2^{|\psi|^c}$.
Therefore the model checking algorithm \texttt{ModCheck} of Section~\ref{sec:representatives} still runs in \emph{exponential working space} with respect to the succinct input formula $\psi$---by preliminarily expanding $\psi$ to $\expand(\psi)$---as $\tau(|W|,\nestb(\expand(\psi)))$ is exponential in $|W|$ and $|\psi|$.

Moreover, the following result holds:
\begin{theorem}\label{threduction}
The model checking problem for succinctly encoded formulas of $\AAbarBBbarEbar$ over finite Kripke structures is NEXP-hard (under polynomial-time reductions).
\end{theorem}
The theorem is proved by means of a reduction from the acceptance problem for a (generic) language $L$ decided by a \emph{non-deterministic one-tape} Turing machine $M$ (w.l.o.g.) that halts in $O(2^{n^{k}})$ computation steps on any input of size $n$, where $k>0$ is a constant. We suitably define  a Kripke structure $\mathpzc{K}=(\mathpzc{AP},W,\delta,\mu,w_0)$ and a succinct $\AAbarBBbarEbar$ formula $\psi$ such that $\mathpzc{K}\models\psi$ if and only if
$M$ accepts its input string $c_{0}c_{1}\cdots c_{n-1}$.

This allows us to conclude that the model checking problem for succinct $\AAbarBBbarEbar$ formulas over finite Kripke structures is between NEXP and EXPSPACE.
We end this section by proving Theorem~\ref{threduction}.

\begin{proof}
Let us consider a language $L$ decided by a \emph{non-deterministic one-tape} Turing machine $M$ (w.l.o.g.) that halts after no more than $2^{n^{k}}-3$ computation steps on an input of size $n$ (assuming a sufficiently high constant $k\in\mathbb{N}$). Hence, $L$ belongs to NEXP. 

Let $\Sigma$ and $Q$ be the alphabet and the set of states of $M$, respectively, and let $\#$ be a special symbol
not in $\Sigma$ used as separator for configurations (in the following we let $\Sigma'=\Sigma\cup\left\{ \#\right\} $).
The alphabet $\Sigma$ is assumed to contain the blank symbol $\sqcup$.

As usual, a computation of $M$ is a sequence of configurations of $M$, where each configuration fixes the content of the tape, the position of the head on the tape and the internal state of $M$. 
We use a standard encoding for computations called \emph{computation table} (or tableau) (see \cite{Pap94,Sip12} for further details).
Each configuration of $M$ is a sequence over the alphabet $\Gamma=\Sigma'\cup(Q\times\Sigma)$; a symbol in $(q,c) \in Q\times\Sigma$ occurring in the $i$-th position encodes the fact that the machine has internal state $q$ and its head is currently on the $i$-th position of the tape (obviously exactly one occurrence of a symbol in $Q\times\Sigma$ occurs in each configuration). 
Since $M$ halts after no more than $2^{n^{k}}-3$ computation steps, $M$ uses at most $2^{n^{k}}-3$ cells on its tape, so the size of a configuration is $2^{n^{k}}$ (we need 3 occurrences of the auxiliary symbol $\#$, two for delimiting the beginning of the configuration, and one for the end; additionally $M$ never overwrites delimiters $\#$). If a configuration is actually shorter than $2^{n^k}$, it is padded with $\sqcup$ symbols in order to reach length $2^{n^k}$ (which is a fixed number, once the input length is known). 
Moreover, since $M$ halts after no more than $2^{n^{k}}-3$ computation steps, the number of configurations
is $2^{n^{k}}-3$. The computation table is basically a matrix of $2^{n^{k}}-3$ rows and $2^{n^{k}}$ columns, where the $i$-th row records the configuration of $M$ at the $i$-th computation step.

\begin{figure}[tb]
\begin{equation*}
\underbrace{
\begin{array}{|c|c|c|c|c|c|c|c|c|c|c|c|c|c|}
\hline
\# & \# & (q_{0},c_{0}) & c_{1} & c_{2} & \cdots & \cdots & c_{n-1} & \sqcup & \sqcup & \cdots & \cdots & \sqcup & \#\\
\hline
\# & \# & c_{0}' & (q_{1},c_{1}) & c_{2} & \cdots & \cdots & c_{n-1} & \sqcup & \sqcup & \cdots & \cdots & \sqcup & \#\\
\hline
\vdots & \vdots &  &  &  & \ddots & \ddots &  &  &  &  &  &  & \vdots\\
\hline
\vdots & \vdots &  &  &  & \ddots & \ddots &  &  &  &  &  &  & \vdots\\
\hline
\# & \# & \cdots & \cdots & (q_{yes},c_{k}) & \cdots & \cdots & \cdots & \cdots & \cdots & \cdots & \cdots & \cdots & \#\\
\hline
\end{array}
}_{2^{n^k}}
\end{equation*}
\caption{An example of computation table (tableau).}\label{table}
\end{figure}
As an example, a possible table is depicted in Figure~\ref{table}. In the first configuration (row) the head is in the leftmost position (on the right of delimiters $\#$) and $M$ is in state $q_0$. In addition, we have the string symbols $c_0c_1\cdots c_{n-1}$ padded with occurrences of $\sqcup$ to reach length $2^{n^k}$. In the second configuration, the head has moved one position to the right, $c_0$ has been overwritten by $c_0'$, and $M$ is in state $q_1$.
From the first two rows, we can deduce that the tuple $(q_{0},c_{0},q_{1},c_{0}',\rightarrow)$ belongs to
the transition relation $\delta_M$ of $M$ (we assume that $\delta_M \subseteq Q  \times \Sigma \times Q \times \Sigma \times \{\rightarrow,\leftarrow,\bullet\}$ with the obvious standard meaning).

Following \cite{Pap94,Sip12}, 
we now introduce the notion of (legal) window. A window is a $2\times 3$ matrix, in which the first row
represents three consecutive symbols of a possible configuration. The second row represents the three symbols which are placed exactly in the 
same position in the next configuration. A window is legal when the changes from the first to the second row 
are coherent with $\delta_M$ in the obvious sense. Actually, the set of legal windows, which we denote by $Wnd\subseteq\left(\Gamma^{3}\right)^{2}$, is a tabular representation of the transition relation $\delta_M$.

For example, two legal windows associated with the table of the previous example are: 
\begin{equation*}
\begin{array}{|c|c|c|}
\hline
\# & (q_{0},c_{0}) & c_{1} \\
\hline
\# & c_{0}' & (q_{1},c_{1}) \\
\hline
\end{array}
\hspace{1cm}
\begin{array}{|c|c|c|}
\hline
(q_{0},c_{0}) & c_{1} & c_2\\
\hline
c_{0}' & (q_{1},c_{1}) & c_2 \\
\hline
\end{array}
\end{equation*}

Formally, a $((x,y,z),(x',y',z'))\in Wnd$ can be represented as 
 \begin{equation*}
\begin{array}{|c|c|c|}
\hline
x & y & z \\
\hline
x' & y' & z' \\
\hline
\end{array}\qquad \text{with } x,x',y,y',z,z'\in\Gamma ,
\end{equation*}
where the following constraints must hold:
\begin{enumerate}
    \item if all $x,y,z\in\Sigma'$ ($x$, $y$, $z$ are not state-symbol pairs), then $y=y'$;
    \item if one of $x$, $y$ and $z$ belongs to $Q\times\Sigma$, then $x'$, $y'$ and $z'$ are coherent with $\delta_M$, and
    \item $(x=\#\Rightarrow x'=\#)\wedge(y=\#\Rightarrow y'=\#)\wedge(z=\#\Rightarrow z'=\#)$.
\end{enumerate}
As we said, $M$ never overwrites a $\#$ and we can assume that the head never visits a $\#$, as well (some more windows can be possibly added if necessary, see~\cite{Pap94}). 

In the following we define a Kripke structure $\mathpzc{K}=(\mathpzc{AP},W,\delta,\mu,w_0)$ and a (succinct) formula $\psi$ of $\AAbarBBbarEbar$  such that $\mathpzc{K}\models\psi$ if and only if
$M$ accepts its input string $c_{0}c_{1}\cdots c_{n-1}$. 
The set of propositional letters is $\mathpzc{AP}=\Gamma\cup\Gamma^{3}\cup\left\{ start\right\}$. 
The Kripke structure  $\mathpzc{K}$ is obtained by suitably composing a basic
pattern called \emph{gadget}.   An instance of the gadget is associated with a triple of symbols  $(a,b,c)\in\Gamma^{3}$ (i.e., a sequence of three adjacent symbols in a configuration) and consists of 3 states: $q_{(a,b,c)}^{0}$, $q_{(a,b,c)}^{1}$, $q_{(a,b,c)}^{2}$ such that 
\begin{equation*}
\mu\left(q_{(a,b,c)}^{0}\right)=\mu\left(q_{(a,b,c)}^{1}\right)=\left\{ (a,b,c),c\right\}
\text{ and }
\mu\left(q_{(a,b,c)}^{2}\right)=\emptyset .
\end{equation*}
Moreover, 
\begin{equation*}
\delta\left(q_{(a,b,c)}^{0}\right)=\left\{ q_{(a,b,c)}^{1}\right\} \text{ and } \delta\left(q_{(a,b,c)}^{1}\right)=\left\{ q_{(a,b,c)}^{2}\right\} .
\end{equation*}
(See Figure \ref{gadget}.) The underlying idea is that a gadget associated with $(x,y,z)\in\Gamma^{3}$ ``records'' the current proposition letter $z$, as well as two more ``past'' letters ($x$ and $y$).

\begin{figure}[t] 
\centering 
\begin{tikzpicture}[->,>=stealth,thick,shorten >=1pt,auto,minimum width=2cm,node distance=3.5cm,every node/.style={circle,draw}]     
\node [style={double}](v0) {$\stackrel{q^0_{(a,b,c)}}{(a,b,c),c}$};
\node (v1) [left of=v0] {$\stackrel{q^1_{(a,b,c)}}{(a,b,c),c}$};
\node (v2) [left of=v1] {$\stackrel{q^2_{(a,b,c)}}{\emptyset}$};
\node (v3) [minimum width=0, style={draw=none}, left of=v2] {$\vdots$};
\node (v4) [minimum width=0, style={draw=none}, above left of=v2] {$\vdots$};
\node (v5) [minimum width=0, minimum width=0, style={draw=none}, below left of=v2] {$\vdots$};
\draw (v0) to (v1);
\draw (v1) to (v2);
\draw (v2) to (v3);
\draw (v2) to (v4);
\draw (v2) to (v5);
\draw[black,thick,dashed] ($(v0.north east)+(0.5,0.5)$)  rectangle ($(v2.south west)+(-0.4,-0.5)$);
\end{tikzpicture}
\caption{An instance of the described gadget for $(a,b,c)\in\Gamma^{3}$.} 
\label{gadget}
\end{figure}
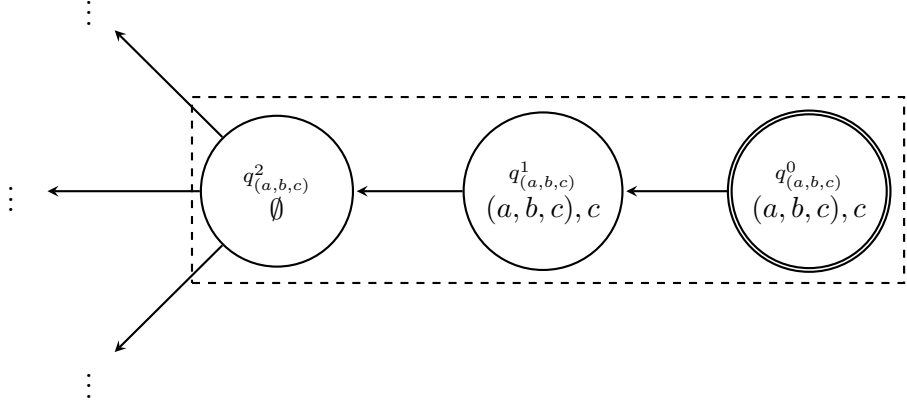

The Kripke structure $\mathpzc{K}$ has (an instance of) a gadget for every $(x,y,z)\in\Gamma^{3}$ and for all $(x,y,z)$ and $(x',y',z')$ in $\Gamma^{3}$, we have $q_{(x',y',z')}^{0}\in\delta\left(q_{(x,y,z)}^{2}\right)$ if and only if $x'=y$ and $y'=z$. Moreover, $\mathpzc{K}$ has some additional (auxiliary) states $w_0,\cdots , w_6$ described in Figure \ref{initK} and $\delta(w_{6})=\left\{ q_{(\#,\#,x)}^{0}\mid x\in\Gamma\right\}$. Note that the overall size of  $\mathpzc{K}$ only depends on $|\Gamma|$ and it is constant w.r.\ to the input string $c_{0}c_{1}\cdots c_{n-1}$ of $M$.

\begin{figure}[t] 
\centering 
\begin{tikzpicture}[->,>=stealth,thick,shorten >=1pt,auto,minimum width=1.7cm,node distance=2.5cm,every node/.style={circle,draw}]     
\node [style={double}](v0) {$\stackrel{w_0}{start}$};
\node (v1) [left of=v0] {$\stackrel{w_1}{start,\#}$};
\node (v2) [left of=v1] {$\stackrel{w_2}{\#}$};
\node (v3) [left of=v2] {$\stackrel{w_3}{\emptyset}$};
\node (v4) [above of=v1] {$\stackrel{w_4}{\#}$};
\node (v5) [left of=v4] {$\stackrel{w_5}{\#}$};
\node (v6) [left of=v5] {$\stackrel{w_6}{\emptyset}$};
\node (v10) [minimum width=0, style={draw=none}, left of=v6] {$\vdots$};
\node (v11) [minimum width=0, style={draw=none}, above left of=v6] {$\vdots$};
\node (v12) [minimum width=0, style={draw=none}, below left of=v6] {$\vdots$};
\draw (v0) to (v1);
\draw (v1) to (v2);
\draw (v2) to (v3);
\draw (v3) to (v4);
\draw (v4) to (v5);
\draw (v5) to (v6);
\draw (v6) to (v10);
\draw (v6) to (v11);
\draw (v6) to (v12);
\end{tikzpicture}
\caption{Initial part of $\mathpzc{K}$.}\label{initK}
\end{figure}

Now we want to decide whether an input string belongs to the language $L$ by solving the model checking problem $\mathpzc{K}\models start\rightarrow\hsA\xi$ where $\xi$ is satisfied only by tracks which represent a successful computation of $M$. Since the only (initial) track which satisfies $start$ is $w_{0}w_{1}$, we are actually verifying the existence of a track which begins with $w_{1}$ and satisfies $\xi$.

As for $\xi$, it requires that a track $\rho$, for which $\mathpzc{K},\rho\models\xi$ (with $\fst(\rho)=w_1$), mimics a successful computation of $M$ in this way: every interval $\rho(i,i+1)$, for $i\mod3=0$, satisfies the proposition letter $p\in\mathpzc{AP}$ if and only if the $\frac{i}{3}$-th character of the computation represented by $\rho$ is $p$ (note that as a consequence of the gadget structure, only $\rho$'s subtracks $\overline{\rho}=\rho(i, i+1)$ for $i\mod3=0$ can satisfy some proposition letters). A symbol of a configuration is mapped to an occurrence of an instance of a gadget in $\rho$; $\rho$, in turn, encodes a computation of $M$ through the concatenation of the first, second, third\dots rows of the computation table (two consecutive configurations are separated by 3 occurrences of $\#$, which require 9 states overall).

Let us now define the  HS formula $\xi=\psi_{accept}\wedge\psi_{input}\wedge\psi_{window}$,
where
\begin{equation*}
\psi_{accept}=\hsB\hsA\bigvee_{a\in\Sigma}(q_{yes},a)
\end{equation*}
requires a track to contain an occurrence of the accepting state of $M$, $q_{yes}$; $\psi_{input}$ is a bit more involved and demands that the subtrack corresponding to the first configuration of $M$ actually ``spells'' the input $c_0c_1\cdots c_{n-1}$, suitably padded with occurrences of $\sqcup$ and terminated by a $\#$ (in the following, $\ell(k)$, introduced in Example~\ref{ex:ellk}, is satisfied only by those tracks whose length equals $k$ ($k\geq 2$) and it has a binary encoding of $O(\log k)$ bits): 
\begin{multline*}
\psi_{input}=[B]\Big(\ell(7)\rightarrow\hsA(q_{0},c_{0})\Big)\wedge
[B]\Big(\ell(10)\rightarrow\hsA c_1\Big)\wedge
[B]\Big(\ell(13)\rightarrow\hsA c_2\Big)\wedge\\
\vdots\\
[B]\Big(\ell(7+3(n-1))\rightarrow\hsA c_{n-1}\Big)\wedge\\
[B]\Bigg(\hsB^{5+3n}\top \wedge [B]^{3\cdot 2^{n^{k}}-6}\bot\rightarrow \hsA\bigg(\Big(\ell(2)\wedge\bigwedge_{a\in\Gamma}\neg a \Big)\vee\sqcup\bigg)\Bigg)\wedge\\
[B]\Big(\ell\big(3\cdot 2^{n^{k}}-2\big)\rightarrow\hsA\#\Big).
\end{multline*}

Finally $\psi_{window}$ enforces the window constraint: if the proposition $(d,e,f)\in\Gamma^3$ is witnessed in a subinterval (of length 2) in the subtrack of $\rho$ corresponding to the $j$-th configuration of $M$, then in the same position of (the subtrack of $\rho$ associated with) configuration $j-1$, some $(a,b,c)\in\Gamma^3$ must be there, such that $((a,b,c),(d,e,f))\in Wnd$.
\begin{multline*}
\psi_{window}= [B]\Bigg(\bigwedge_{i=2,\cdots ,t}\bigwedge_{(d,e,f)\in\Gamma^3}\Big(\ell(3\cdot 2^{n^k}+3i+1)\wedge\hsA (d,e,f)\\
\rightarrow [B] \big(\ell(3i+1)\rightarrow\bigvee_{((a,b,c),(d,e,f))\in Wnd}\hsA (a,b,c)\big)\Big)\Bigg).
\end{multline*}
where $t=2^{n^k}\cdot (2^{n^k}-4)-1$ is encoded in binary.

All the integers which must be stored in the formula are less than $(2^{n^k})^2$, thus they need $O(n^{k})$ bits to be encoded; in this way the formula can be generated in polynomial time. 
\end{proof}